\documentclass[11pt]{article}
\usepackage[left=1in,top=1in,right=1in,bottom=1in,head=.1in,nofoot]{geometry}

\usepackage{setspace,url,bm,amsmath} 

\usepackage{titlesec} 
\titlelabel{\thetitle.\quad} 

\usepackage[numbers]{natbib}
\usepackage{graphicx} 
\usepackage{bbm}
\usepackage{latexsym}
\usepackage{caption}
\usepackage[margin=20pt]{subcaption}
\usepackage{hyperref}
\usepackage[]{algorithm2e}
\usepackage{booktabs}

\usepackage{mathtools}

\usepackage[table]{xcolor}

\newcommand{\GG}[1]{}

\usepackage{amsthm}
\usepackage{amssymb}
\usepackage{amsmath}
\usepackage{color}

\usepackage{booktabs}

\usepackage{comment}
\theoremstyle{definition}
\newtheorem{assumption}{Assumption}
\newtheorem*{theorem*}{Theorem}
\newtheorem*{rmk*}{Remark}

\newtheorem{theorem}{Theorem}
\newtheorem{proposition}{Proposition}
\newtheorem{lemma}{Lemma}

\newtheorem{definition}{Definition}
\newtheorem{corollary}{Corollary}
\newtheorem*{corollary*}{Corollary}
\newtheorem{condition}{Condition}

\usepackage{mathrsfs}

\usepackage{natbib} 
\bibpunct{(}{)}{;}{a}{}{,} 

\usepackage{color}
\usepackage{listings}

\usepackage[sc]{mathpazo}

\DeclareMathOperator*{\argmax}{arg\,max}

\newcommand{\lxr}{\color{red}}

\def\Var{\text{Var}}
\def\Cov{\text{Cov}}

\def\converged{\stackrel{d}{\longrightarrow}}

\def\convergep{\stackrel{p}{\longrightarrow}}
\def\converged{\stackrel{d}{\longrightarrow}}

\def\yscale{c}
\def\yshift{\xi_{[a]}}

\def\Yhhrrprime{\overline{Y_{[a]}(r)Y_{[a]}(r')}}
\def\pr{\text{pr}}

\DeclareFontFamily{U}{mathx}{\hyphenchar\font45}
\DeclareFontShape{U}{mathx}{m}{n}{
      <5> <6> <7> <8> <9> <10>
      <10.95> <12> <14.4> <17.28> <20.74> <24.88>
      mathx10
      }{}
\DeclareSymbolFont{mathx}{U}{mathx}{m}{n}
\DeclareFontSubstitution{U}{mathx}{m}{n}
\DeclareMathAccent{\widecheck}{0}{mathx}{"71}
\DeclareMathAccent{\wideparen}{0}{mathx}{"75}

\RequirePackage[normalem]{ulem}

\allowdisplaybreaks

\begin{document}
\onehalfspacing 
\title{\bf Randomization Inference for Peer Effects}

\author{
Xinran Li, Peng Ding, Qian Lin, Dawei Yang and Jun S. Liu
\footnote{
Xinran Li is Doctoral Graduate, Department of Statistics, Harvard University,  Cambridge, MA 02138 (E-mail: \url{xinranli@fas.harvard.edu}). Peng Ding is Assistant Professor, Department of Statistics, University of California, Berkeley, CA 94720 (E-mail: \url{pengdingpku@berkeley.edu}).
Qian Lin is Assistant Professor, 
Center for Statistical Science, 
Center for Statistical Science, Department of Industrial Engineering, Tsinghua University, Beijing, 100084, P. R. China (E-mail: \url{qianlin@tsinghua.edu.cn}). 
Dawei Yang 
is Research Associate, Bureau of Personnel of Chinese Academy of Sciences \& School of Education of Peking University, Beijing, 100871, P. R. China (Email: \url{yangdw@cashq.ac.cn} or \url{yangdawei@pku.edu.cn}).
Jun S. Liu is Professor, Department of Statistics, Harvard University, Cambridge, MA 02138 (E-mail: \url{jliu@stat.harvard.edu}). 
}
}

\date{}
\maketitle

\begin{abstract}
Many previous causal inference studies require no interference, that is, the potential outcomes of a unit do not depend on the treatments of other units. However, this no-interference assumption becomes unreasonable when a unit interacts with other units in the same group or cluster. In a motivating application, a university in China admits students through two channels: the college entrance exam (also known as Gaokao) and recommendation (often based on Olympiads in various subjects). The university randomly assigns students to dorms, each of which hosts four students. Students within the same dorm live together and have extensive interactions. Therefore, it is likely that peer effects exist and the no-interference assumption does not hold. It is important to understand peer effects, because they give useful guidance for future roommate assignment to improve the performance of students. We define peer effects using potential outcomes. We then propose a randomization-based inference framework to study peer effects with arbitrary numbers of peers and peer types. Our inferential procedure does not assume any parametric model on the outcome distribution. Our analysis gives useful practical guidance for policy makers of the university in China.

\medskip
\noindent {\it Key Words}: Causal inference; Design-based inference; Grade point average (GPA); Interference; Optimal treatment assignment; Spillover effect
\end{abstract}

 \newpage 
\section{Introduction}

\subsection{Causal inference, interference, and peer effects}

The classical potential outcomes framework \citep{Neyman:1923} assumes no interference among experimental units \citep{cox1958planning}, i.e., the potential outcomes of a unit are functions of its own treatment but not others' treatments. This constitutes an important part of \citet{Rubin:1980}'s Stable Unit Treatment Value Assumption (SUTVA). In some experiments, interference is a nuisance, and researchers try to avoid it by isolating units. Interference, however, is unavoidable in many studies when units have interactions with each other. Examples include vaccine trials for infectious diseases in epidemiology \citep{halloran1991study, halloran1995causal, perez2014assessing}, group-randomized trials in education \citep{hong2006evaluating, vanderweele2013mediation}, and
interventions on networks in sociology \citep{an2011models, vanderweele2013social}, 
political science \citep{nickerson2008voting, ichino2012deterring, bowers2013reasoning} and 
economics \citep{manski1993identification, Sacerdote01052001, miguel2004worms, graham2010measuring, goldsmith2013social, arpino2016assessing}. 
\citet{ogburn2014causal} discussed different types of interference.
\citet{forastiere2016identification} showed that ignoring interference can lead to biased inferences. It is important to study the pattern of interference in some applications, because it is of scientific interest and useful for decision making. For example, \citet{Sacerdote01052001} found significant peer effects in student outcomes (e.g., GPA and fraternity membership) among students living in the same dorm of Dartmouth College. Based on this, \citet{bhattacharya2009inferring} discussed the optimal peer assignment.

 \subsection{Motivating application in education}
 
Our motivation comes from a data set of a university in China. It contains a rich set of variables of the students: family background, the ways they were admitted, roommates' information, GPAs, etc.

The university admits students through two primary channels: the college entrance exam (also known as Gaokao) and recommendation. Gaokao is an annual test in China to assess students' knowledge in various subjects. Every university has its own minimal test score threshold to admit students. Students from Gaokao study all subjects and often have broader knowledge. Students from recommendation do not need to take Gaokao. They win awards in national or international Olympiads in mathematics, physics, chemistry, biology, or informatics. They concentrate on a certain subject for the corresponding Olympiad. 
They may even take some college courses on that subject during their high school years. 
Most of them choose majors related to the subject they focused on in high schools. Students admitted through these two channels have different training and thus different attributes. Students from recommendation generally perform better in GPAs than students from Gaokao.

After entering the university, students usually live in four-person rooms for four years. They often study together and interact with each other. We know that two types of students, from Gaokao and recommendation, have different training in high schools. It is then natural to ask the following questions. Is it beneficial for students from Gaokao to live with students from recommendation, or vice versa? Is there an optimal combination of roommate types for the performance of a certain student? Is there an optimal roommate assignment to maximize the performance of all students? These questions are all about peer effects among students.

\subsection{Literature review and contribution}
 
With interference, the potential outcomes of a unit can depend on its own treatment and others' treatments in various ways. Therefore, causal inference with interference has different mathematical forms. Like many other causal inference problems, there are at least two inferential frameworks for causal inference with interference: the Fisherian and Neymanian perspectives. Under the Fisherian view, \citet{rosenbaum2007interference}, \citet{luo2012inference}, \citet{aronow2012general}, \citet{bowers2013reasoning}, \citet{rigdon2015exact}, \citet{athey2015exact} and \citet{basse2017exact} proposed exact randomization tests for detecting causal effects with interference, and constructed confidence intervals for certain causal parameters by inverting tests. \citet{choi2016estimation} discussed a related approach under the monotone treatment effect assumption. Under the Neymanian view, \citet{hudgens2008toward} discussed point and interval estimation for several causal estimands with interference under two-stage randomized experiments on both the group and individual levels. \citet{liu2014large} then established the large sample theory for these estimators. \citet{aronow2015estimating}, \citet{bassefeller2017}, and \citet{savje2017average} extended the discussion to other general contexts. The Fisherian and Neymanian views are both randomization-based in the sense that the uncertainty in testing or estimation comes solely from the treatment assignment mechanism, and all the potential outcomes are fixed constants. When two-stage randomization is infeasible, we need certain unconfoundedness assumptions. \citet{tchetgen2012causal} proposed an inverse probability weighting estimator. \citet{perez2014assessing} applied this methodology to assess effects of cholera vaccination.  \citet{liu2016inverse} studied the theoretical properties. Other studies \citep{Sacerdote01052001, toulis2013estimation} relied on parametric modeling assumptions.

Our framework for peer effects furthers the literature in several ways. First, we define peer effects using potential outcomes. Unlike in previous work \citep[e.g.,][]{sobel2006randomized, hudgens2008toward}, our estimands do not involve averages over the treatment assignment. 
We separate the causal estimands from the treatment assignment.
As \citet{rubin2005causal} argued, the former are functions of the potential outcomes, and the latter induces randomness and governs the statistical inference. 
Second, previous works discussed external interventions with known networks, clusters, or groups.
Our hypothetical intervention is the roommate assignment in the motivating application. It forms a ``network'' among units, which further causes interference and peer effects. 
We explain the distinction between the two types of interference in detail in Section \ref{sec:distinction}.
Our setting is similar to \citet{Sacerdote01052001}'s. However, we formalize the problem using potential outcomes instead of linear models and allow for causal interpretations without imposing model assumptions. 
Third, we propose randomization-based point estimators, prove their asymptotic Normalities, and construct confidence intervals.
We further derive the optimal roommate assignment to maximize the performance of students. The inferential framework is Neymanian, similar to those of \citet{hudgens2008toward} and \citet{aronow2015estimating}. Fourth, we apply the new method to the data set from a university in China and find important policy implications. We relegate all the technical details to the Supplementary Material.

\section{Notation and framework for peer effects}\label{sec:notation_assump}

\subsection{Potential outcomes with peers}\label{sec::potentialoutcomes}
We consider an experiment with $n=m(K+1)$ units, where $m$ is the number of groups and $K+1$ is the size of each group. Each unit has $K$ peers in the same group. The group and peers correspond to room and roommates in our motivating application, where $K=3$ is the number of roommates for each student. Let $Z_i$ be the treatment assignment for unit $i$, which is a set consisting of the identity numbers of his/her $K$ peers, i.e., $Z_i=\{j: \text{units }j \text{ and } i \text{ are in the same group}\}$. In the motivating application,
$Z_i$ is a set consisting of three roommates of unit $i$. Let $Z=(Z_1,Z_2,\ldots,Z_n)$ be the treatment assignment for all units, and $\mathcal{Z}$ be the set of all possible values of the assignment $Z$. Let $Y_i(z)$ be the potential outcome of unit $i$ under treatment assignment $z=(z_1,\ldots, z_n)$. This potential outcome depends on treatment assignments of all other units. Let $A_i\in \{1,2,\ldots, H\}$ be the attribute or type of unit $i$. In the motivating application, $H=2$, and $A_i=1$ if unit $i$ is from Gaokao, and $A_i=2$ if unit $i$ is from recommendation. Under treatment assignment $z$, let $R_i(z_i)=\{A_j: j \in z_i\}$ be the set consisting of the attributes of unit $i$'s $K$ peers, and $G_i(z_i)=R_i(z_i)\cup \{A_i\}$ be the set consisting of the attributes of all units in the group that unit $i$ belongs to. 
We call $R_i(z_i)$ and $G_i(z_i)$ the peer attribute set and group attribute set. Both of them contain unordered but replicable elements. Therefore, $|R_i(z_i)| = K$ and $|G_i(z_i)| = K+1,$ where $|\cdot|$ denotes the cardinality of a set. In the motivating application,
if unit $i$  is from recommendation and has $2$ roommates from Gaokao and $1$ from recommendation, then $R_i(z_i)=\{1,1,2\}\equiv 112$ and $G_i(z_i)=\{1,1,2,2\}\equiv 1122$, where we use 112 and 1122 for notational simplicity. In this case, $R_i$ or $G_i$ has a one-to-one mapping to the number of students from Gaokao within the room of unit $i$.

Let $I(\cdot)$ be the indicator function. For unit $i$, $Y_i = \sum_{z\in \mathcal{Z}} I(Z=z)Y_i(z)$ is the observed outcome, and $R_i = \sum_{z_i} I(Z_i=z_i) R_i(z_i)$ is the observed peer attribute set. These summations are over all possible values of the treatment assignments for all units.

\subsection{Group-level SUTVA and exclusion-restriction-type assumptions}
Without further assumptions, the potential outcome $Y_i(z)$ depends on the treatments of all units. This makes statistical inference intractable. We invoke the following two assumptions to reduce the number of potential outcomes.

\begin{assumption}\label{asp:sutva_pe}
If $z_i = z'_i$, then $Y_i(z) = Y_i(z')$, for any two treatment assignments $(z, z')$ and any unit $i$.
\end{assumption}

Assumption \ref{asp:sutva_pe} states that if a unit's peers do not change, then its potential outcome will not change. 
This assumption requires no interference between groups but allows for interference within groups. Under Assumption \ref{asp:sutva_pe}, each unit's potential outcomes depend only on its peers in the same group. Therefore, we can write $Y_i(z)$ as $Y_i(z_i)$, a function of the peers of unit $i$. Assumption \ref{asp:sutva_pe} is a group-level SUTVA, which is similar to the ``partial interference'' assumption \citep{sobel2006randomized, hudgens2008toward}.

\begin{assumption}\label{asp:exres_pe}
If $R_i(z_i) = R_i(z'_i)$, then $Y_i(z_i)=Y_i(z_i')$, for any two treatment assignments $(z, z' ) $ and any unit $i$.
\end{assumption}

Assumption \ref{asp:exres_pe} states that if the treatment assignment does not affect the attributes of the peers of unit $i$, then it does not affect the outcome of unit $i$. Therefore, the potential outcomes of each unit depend only on its peers' attributes instead of its peers' identities. Assumption \ref{asp:exres_pe} is similar to ``anonymous interaction'' \citep{manski2013identification}. Assumption \ref{asp:exres_pe} implies that the peer attribute set of a unit is the {\it ultimate treatment} of interest. We are inferring the treatment effects of the peer attribute set. Previous works often invoked Assumption \ref{asp:exres_pe}, or a slightly weaker form, for inferring peer effects among college roommates. For example, in \cite{Sacerdote01052001}'s study from Dartmouth College, the ultimate treatment was peers' academic indices created by the admission office, and in \cite{rubinpeersmoking}'s study from Harvard College, the ultimate treatment was peers' smoking behaviors.

Both Assumptions \ref{asp:sutva_pe} and \ref{asp:exres_pe} are untestable based on the observed data from a single experiment. 
They are strong identifying assumptions. We will relax them in Section \ref{sec:discussion}.

Under Assumptions \ref{asp:sutva_pe} and \ref{asp:exres_pe}, $Y_i(z)$ simplifies to $Y_i(R_i(z_i))$, a function of the peer attribute set of unit $i$. Recall that $R_i(z_i)$ contains $K$ unordered but replicable elements from $\{1,2,\ldots, H\}$.
Let $\mathcal{R}$ be the set consisting of all possible values of $R_i(z_i)$. Potential outcome of unit $i$, $Y_i(z)$, simplifies to $Y_i(r)$ for some $r\in \mathcal{R}$. 
Then 
the potential outcome is $Y_i(z) = Y_i(R_i(z_i))=\sum_{r \in \mathcal{R}}I\{R_i(z_i)=r  \}Y_i(r)$, and 
the observed outcome is $Y_i = \sum_{r \in \mathcal{R}}  I(R_i = r) Y_i(r) $.
Therefore, we can view the elements in $\mathcal{R}$ as hypothetical treatments, with
$
|\mathcal{R}|=  \binom{K+H-1}{H-1} = \frac{(K+H-1)! } { (H-1)!K! }
$
possible values. 
In our motivating application,
$\mathcal{R}=
\{r_1, r_2, r_3, r_4\}
=\{111, 112, 122, 222\}$ and 
$|\mathcal{R}|=
\frac{(3+2-1)!}{(2-1)!3!}
=4.$

As a side note, motivated by the example of the university in China,
we consider the case with equal group sizes $K+1$. When groups have different sizes, we need to modify Assumption \ref{asp:exres_pe}. For example, we can assume that the potential outcomes of a unit depend on the proportions of his/her peers' attributes. The plausibility of this assumption depends on the context of the application, and we leave it to future work.

\subsection{Causal estimands for peer effects}

For units with attribute $1\leq a\leq H$, 
let $n_{[a]}$ and $w_{[a]} = n_{[a]}/n$ be the number and proportion, and
$\bar{Y}_{[a]}(r) = n_{[a]}^{-1} \sum_{i:A_i=a}Y_i(r) $ be the subgroup average potential outcome under treatment $r$. Let $\bar{Y}(r) = n^{-1} \sum_{i=1}^n Y_i(r)  $ be the average potential outcome for all units under treatment $r$. Therefore, $\bar{Y}(r)=\sum_{a=1}^H w_{[a]} \bar{Y}_{[a]}(r)$ is a weighted average of $\bar{Y}_{[a]}(r)$'s. 
Comparing treatments $r,r'\in \mathcal{R}$,
we define $\tau_{i}(r,r') = Y_i(r)-Y_i(r')$  as the individual peer effect, 
\begin{align}\label{eq:sub_ave_peer_effect}
\tau_{[a]}(r,r') = n_{[a]}^{-1}\sum_{i:A_i=a}\tau_{i}(r,r')
=
\bar{Y}_{[a]}(r) - \bar{Y}_{[a]}(r')
\end{align}
as the subgroup average peer effect for units with attribute $a$, and
\begin{align}\label{eq:ave_peer_effect}
\tau(r,r') = {n^{-1}}\sum_{i=1}^n \tau_{i}(r,r')
=
\bar{Y}(r)-\bar{Y}(r')
= \sum_{a=1}^H w_{[a]} \tau_{[a]} (r,r')
\end{align}
as the average peer effect for all units. 
We are interested in estimating the average peer effects $\tau_{[a]}(r,r')$ and $\tau(r,r')$. They are functions of the fixed potential outcomes and do not depend on the treatment assignment mechanism.

For ease of reading, we summarize the key notation in Table \ref{tab:notation}.  

\begin{table}[t]
\centering
\caption{Notation and explanations}\label{tab:notation}
\begin{tabular}{lll}
\toprule
notation & definition & meaning, properties or possible values \\
\midrule
$z_i$ & peer assignment of unit $i$ & a set of the identity numbers of his/her $K$ peers\\
$A_i$ &  unit $i$'s attribute & $A_i\in \{1,2,\ldots, H\}$\\
$n_{[a]}$ & number of units with attribute $a$&  $\sum_{a=1}^H n_{[a]} = n$
\\
$w_{[a]}$ & proportion of units with attribute $a$&   $\sum_{a=1}^H w_{[a]} = 1$ and $   0<w_{[a]}<1\quad  (a=1,\ldots, H)$
\\
$R_i(z_i)$ & unit $i$'s peer attribute set & a set of the attributes of unit $i$'s $K$ peers \\
$\mathcal{R}$ &  a set of all possible values of $R_i(z_i)$ & $|  \mathcal{R}|  = \binom{K+H-1}{H-1}$
\\
$G_i(z_i)$ & unit $i$'s group attribute set & a set of attributes of all units in unit $i$'s group\\
$\mathcal{G}$ & a set of all possible values of $G_i(z_i)$ & $\mathcal{G} = \{g_1, \ldots, g_T\}$ with $T=  \binom{K+H}{H-1}$
\\
$Y_i(z)$ &  unit $i$'s potential outcome &under original treatment $z \in \mathcal{Z}$\\
$Y_i(r)$ & unit $i$'s potential outcome& under ultimate treatment $r\in \mathcal{R}$ \\ 
\bottomrule
\end{tabular}
\end{table}

\subsection{Treatment assignment mechanism}\label{sec:treat_assign_mecha}
The treatment assignment mechanism is important for identifying and estimating peer effects. 
We consider treatment assignment mechanisms satisfying some symmetry conditions. First, units with the same attribute must have the same probability to receive all treatments. Second, pairs of units with the same pair of attributes must have the same probability to receive all pairs of treatments. Formally, we require that the treatment assignment mechanism satisfies the following two conditions.

\begin{assumption}\label{asmp:indist}
For any $r,r'\in \mathcal{R}$,
\begin{itemize}
\item[(a)] $\pr(R_i = r) = \pr(R_{j} = r),$ if $A_i = A_{j}$;
\item[(b)] $\pr (
R_{i} = r,R_{j} = {r}'
)  = \pr (
R_{k} = r,R_{q} = {r}'
),$ if $A_i=A_{k}$ and $A_j = A_{q}$ for $i\neq j, k \neq q.$
\end{itemize}
\end{assumption}

We will give two examples of treatment assignment mechanisms satisfying Assumption \ref{asmp:indist}. 

\subsubsection{Random partitioning}\label{sec:rand_part}

Under random partitioning, we randomly assign units to $m$ groups of size $K+1$, and all possible partitions of units have equal probability.  
To be more specific, 
if a treatment assignment 
$z$ is compatible with a partition of units into $m$ groups of size $K+1$, then 
$
\pr(Z=z) =
 m!\{(K+1)!\}^m / \{m(K+1)\}! ;
$
otherwise, $\pr(Z=z)=0$. 
This formula follows from counting all possible random partitions. 
To generate a random partition, we can randomly permute $n = m(K+1)$ units and divide them into $m$ groups of equal size $K+1$ sequentially.

Random partitioning, however, can result in unlucky realizations of the randomization. We may have too few units with attributes and treatments of interest. For illustration, we consider the motivating education example
with $8$ students, $5$ from Gaokao and $3$ from recommendation. Assume that we are interested in $\tau_{[1]}(r_2, r_3)$, the treatment effect of $r_2=112$ versus $r_3=122$ for students from Gaokao.
Under random partitioning, it is possible that no students from Gaokao 
receives treatment $r_2$ or $r_3$. In that case, it is impossible to estimate $\tau_{[1]}(r_2, r_3)$ precisely. 
An example of such a realization is that 4 students from Gaokao live in one room and the remaining 1 student from Gaokao and 3 students from recommendation live in the other room.

\subsubsection{Complete randomization}\label{sec:comp_rand}

We propose another treatment assignment mechanism to avoid the drawback of random partitioning. It requires predetermined number of units for each attribute receiving each treatment. 
We achieve this goal by fixing the numbers of groups.
Recall that the group attribute set $G_i(z_i)$ contains $K+1$ unordered but replicable elements from $\{1,\ldots, H\}$.  Consider the same education example with 5 students from Gaokao and 3 students from recommendation. 
Under random partitioning, 
we may hope that one room has group attribute set $1112$ and thus the other room has group attribute set $1122$. This results in 3 and 2 students from Gaokao receiving treatments $r_2$ and $r_3$, respectively. Therefore, this avoids other assignments with no students from Gaokao receiving these treatments of interest.

We need additional symbols to describe complete randomization.  Let $\mathcal{G}=\{g_1, \ldots, g_T\}$ be the set consisting of all possible group attribute sets, with cardinality
$
T=|\mathcal{G}| = \binom{H+K}{H-1}  
= \frac{(H+K)!}{(H-1)!(K+1)!}. 
$
In our motivating application, 
$\mathcal{G}=\{g_1, \ldots, g_5\} = \{1111,$ $1112,$ $1122,$ $1222,$ $2222\}$ with 
$T=|\mathcal{G}|=
\frac{(2+3)!}{(2-1)!(3+1)!}
=5.$ 
Under treatment assignment $z$, the number of groups with attribute set $g_t \in \mathcal{G}$ is
$$
L_{t}(z) = (K+1)^{-1}\sum_{i=1}^n I\left\{ G_i(z_i)=g_t \right\},
$$
where the divisor $K+1$ appears because all $K+1$ units in the same group must have the same group attribute set.
Let $L(z)=(L_{1}(z), L_{2}(z), \ldots, L_{T}(z))$ be the 
vector of numbers of groups corresponding to group attribute sets 
$(g_1, \ldots, g_T)$ 
under assignment $z$.

Under complete randomization, the assignment $z$ must satisfy $L(z)= l =(l_{1}, \ldots, l_{T})$ for a predetermined constant vector $l$, and all such assignments must have equal probability. For any $g_t\in \mathcal{G}$, let $g_t(a)$ be the number of elements in set $g_t$  that are equal to $a$.
If $z$ is compatible with a partition of units into $m$ groups and $L(z)=l$, then  
\begin{align}\label{eq:dist_Z_com_rand}
\pr(Z=z) = 
\frac{
\prod_{t=1}^T l_{t}! 
\times
\prod_{a=1}^H \prod_{t=1}^T \{g_t(a)!\}^{l_{t}}
}{
\prod_{a=1}^H n_{[a]}!
};
\end{align}
otherwise, $\pr(Z=z)=0.$ 
The above formula \eqref{eq:dist_Z_com_rand} follows from counting all possible complete randomizations. 
To generate a complete randomization, 
we can first randomly partition the $n_{[a]}$ units with attribute $a$ into $m$ groups, where each of the first $l_1$ groups has $g_1(a)$ units, each of the next $l_2$ groups has $g_2(a)$ units, $\ldots$, each of the last $l_T$ groups has $g_T(a)$ units.  The partitions for units with different attributes are mutually independent.  
Finally, the first $l_1$ groups will have group attribute set $g_1$, $\ldots,$ and the last $l_T$ groups will have group attribute set $g_T$, satisfying the requirement $L(z)=l$.

We revisit the education example with 5 students from Gaokao and 3 students from recommendation. The treatment of 
complete randomization has predetermined vector $l=(l_1,l_2,l_3,l_4,l_5)=(0,1,1,0,0)$. Thus, one group has attribute set $g_2 = 1112$ and the other group has attribute set $g_3=1122.$ 
We need to randomly assign 3 students from Gaokao and 1 student from recommendation to group $g_2$, and assign the remaining students to group $g_3$.  
Equivalently, 
for the 5 students from Gaokao, we randomly assign $3$ of them to group $g_2$ and the remaining 2 to group $g_3$; for the 3 students from recommendation, we randomly assign $1$ of them to group $g_2$ and the remaining 2 to group $g_3$, independently of the group assignments for students from Gaokao.

For $1\leq a\leq H$ and $r\in \mathcal{R}$,
let $n_{[a]r}=|\{i:A_i=a, R_i=r\}|$ be the number of units with attribute $a$ receiving treatment $r$.  
First, the units with attribute $a$ receiving treatment $r$ must have group attribute set $\{a\}\cup r$, which equals $g_{t_0}$ for some $1\leq t_0\leq T$. Second,  
each group with attribute set $g_{t_0}$ 
contains $g_{t_0}(a)$ units with attribute $a$.
Third, all of these $g_{t_0}(a)$ units receive the same treatment $r$. 
These facts imply that
\begin{align}\label{eq:n_hr}
n_{[a]r}=L_{t_0}(z)g_{t_0}(a) = \sum_{t=1}^T I(g_t=\{a\}\cup r) \cdot L_t(z)g_t(a)
\end{align}
depends only on the vector $L(z)$. Thus, the $n_{[a]r}$'s are constants under complete randomization. 
In the previous education example with $8$ students, consider complete randomization with predetermined vector $l=(0,1,1,0,0)$. The numbers of units from Gaokao receiving treatments $r_2=112$ and $r_3=122$ are constants $n_{[1]r_2}=3$ and $n_{[1]r_3}=2$. 
Therefore, complete randomization can guarantee that at least some students from Gaokao receive the treatments of interest.

Moreover, under random partitioning, if we conduct inference conditional on $L(Z)$, then the treatment assignment mechanism becomes complete randomization with $L(z)$ fixed at the observed vector $L(Z).$ 
Therefore, even under random partitioning, we can still conduct inference under complete randomization if we condition on $L(Z)$.

\subsection{Connection and distinction between existing literature and our paper}\label{sec:distinction}

We comment on the difference between the majority of the existing literature and our paper. We compare two types of interference.

Figure \ref{fig:illstrate}(a) illustrates the first type. 
The grey or white color of each unit denotes the external treatment (e.g., receiving vaccine or not). Each unit's outcome depends not only on its own treatment but also on treatments of other units in its circle. Thus, units interfere with each other in the same dashed circle.  Importantly, the network structure is fixed.

Figure \ref{fig:illstrate}(b) illustrates the second type. The grey or white color denotes the units' attributes (e.g., from Gaokao or recommendation in the motivating application). The outcome of each unit depends on the attributes of other units
in its circle. Thus, units interfere with each other in the same dashed circle.  Unlike the first type, the units' attributes are fixed but the network structure is random.

\begin{figure}[htb]
\centering
\begin{subfigure}{1\textwidth}
  \centering
\begin{subfigure}{.35\textwidth}
\centering
\includegraphics[width=1\linewidth]{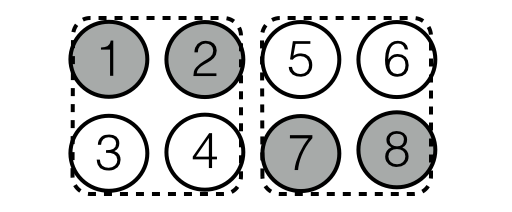}
\caption*{(a1)}
\end{subfigure}
\begin{subfigure}{.35\textwidth}
\centering
\includegraphics[width=1\linewidth]{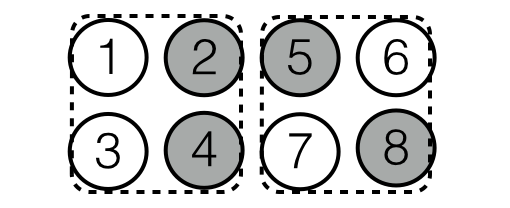}
\caption*{(a2)}
\end{subfigure}
   \caption{
The first type of interference with a fixed network and random external interventions. 
(a1) and (a2) are two possible realizations of random external interventions (colors of the units). 
}
\end{subfigure}
\begin{subfigure}{1\textwidth}
  \centering
\begin{subfigure}{.35\textwidth}
\centering
\includegraphics[width=1\linewidth]{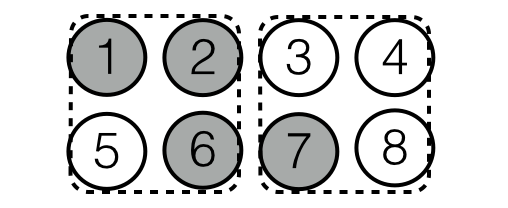}
\caption*{(b1)}
\end{subfigure}
\begin{subfigure}{.35\textwidth}
\centering
\includegraphics[width=1\linewidth]{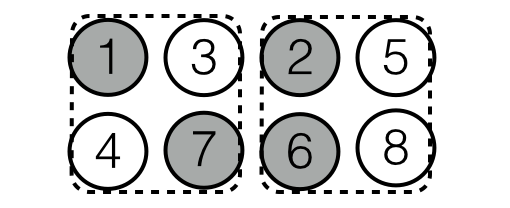}
\caption*{(b2)}
\end{subfigure}
   \caption{
The second type of interference with fixed attributes of all units and a random network. 
(b1) and (b2) are two possible realizations of random networks (dashed circles).
}
\end{subfigure}
\caption{Two types of interference with dashed circles indicating networks. 
}\label{fig:illstrate}
\end{figure}

A main difference between these two types comes from the source of randomness. For the first type, the colors are random and the dashed circles are fixed. For the second type, the colors are fixed and the dashed circles are randomly formed. The recent causal inference literature focused on the first type \citep{hudgens2008toward, aronow2012general, liu2014large, athey2015exact}. In this paper, we formalize the second type and propose inferential procedures based on the treatment assignment mechanism.

\section{Inference for peer effects under general treatment assignment}\label{sec:design_based}
\subsection{Point estimators for peer effects}
Throughout the paper, we invoke, unless otherwise stated, Assumptions \ref{asp:sutva_pe}--\ref{asmp:indist}.
For $1\leq a\leq H$ and $r,r'\in \mathcal{R}$, let $\pi_{[a]}(r)=\pr(R_i=r)$ be the probability that a unit $i$ with attribute $a$ receives treatment $r$. Define 
\begin{align}
\label{eq:sub_ave_pot_hat}
&\hat{Y}_{[a]}(r) = \{n_{[a]} \pi_{[a]}(r)\}^{-1}\sum_{i:A_i=a} I(R_i=r)Y_i, \\
\label{eq:tau_hat}
&\hat{\tau}_{[a]} (r,r') = \hat{Y}_{[a]}(r) - \hat{Y}_{[a]}(r'), \quad 
\hat{\tau}(r,r') = \sum_{a=1}^H w_{[a]} \hat{\tau}_{[a]}(r,r').
\end{align}

\begin{proposition}\label{thm:unbiased_gen}
For $1\leq a, a'\leq H$ and $r, r'\in \mathcal{R}$,   
the estimators
$\hat{Y}_{[a]}(r), \hat{\tau}_{[a]} (r,r')$ and $\hat{\tau}(r,r')$ are unbiased for $\bar{Y}_{[a]}(r), \tau_{[a]}(r,r')$ and $\tau(r,r')$, respectively. 
\end{proposition}

The unbiasedness of $\hat{Y}_{[a]}(r)$ follows from the Horvitz--Thompson-type inverse probability weighting,  
and the unbiasedness of $\hat{\tau}_{[a]}(r,r')$ and $\hat{\tau}(r,r')$ then follows directly from the linearity of expectation.

\subsection{Sampling variances of the peer effect estimators}

For units with attribute 
$1\leq a\leq H$ and $r\neq r'\in \mathcal{R}$, define   
\begin{align*}
S_{[a]}^2({r}) & =  (n_{[a]}-1)^{-1}\sum_{i:A_i =a }\left\{ Y_i({r})
- \bar{Y}_{[a]}({r})
\right\}^2,\\
S_{[a]}^2({r}\text{-}{r}') & = (n_{[a]}-1)^{-1}\sum_{i:A_i=a}\left\{
\tau_i(r,r') - \tau_{[a]}(r,r') 
\right\}^2\nonumber
\end{align*}
as the finite population variances of the potential outcomes and individual peer effects, 
and 
\begin{align*}
\Yhhrrprime
 & = \{n_{[a]}(n_{[a]}-1)\}^{-1} \mathop{\sum\sum}_{i\neq j:A_i=A_j=a} Y_i(r)Y_j(r') 
\end{align*}
as the average of the  
products of the potential outcomes 
for pairs of units with attribute $a$.

For $1\leq a,a'\leq H$ and $r,r'\in \mathcal{R},$ 
if $i\neq j$ are two units with attributes $a$ and $a'$, then $\pi_{[a][a']}(r,r')=\pr(R_i=r, R_j=r')$ is the joint treatment assignment probability, and
\begin{align}\label{eq:d_def}
d_{[a][a']}(r,r') = 
\sqrt{n_{[a]} n_{[a']}}
\left\{
\frac{\pr(R_i=r, R_j=r')}{\pr(R_i=r)\pr(R_j=r')}-1\right\}
=
\sqrt{n_{[a]} n_{[a']}}
\left\{\frac{\pi_{[a][a']}(r,r')}{\pi_{[a]}(r) \pi_{[a']}(r')}-1\right\}
\end{align}
measures the dependence between the two events $\{ R_i=r \}$ and $\{ R_j = r' \}$.
We further need a few known constants depending only on the treatment assignment mechanism. For $1\leq a,a'\leq H$ and $r,r'\in \mathcal{R},$ define
\begin{align}\label{eq:c_def}
c_{[a][a']}(r,r') & = 
\begin{cases}
d_{[a][a']}({r}, {r}'), & \text{if } a\neq a',\\
(1-n_{[a]}^{-1})d_{[a][a]}(r,r') -1, & \text{if } a=a', {r} \neq {r}',\\
(1-n_{[a]}^{-1})d_{[a][a]}(r,r)+\pi_{[a]}^{-1}(r)-1,
 & \text{if } a=a', r=r',
\end{cases}
\end{align}
and 
\begin{align}\label{eq:b_def}
b_{[a]}({r}) & =
(1-n_{[a]}^{-1})\left\{
c_{[a][a]}(r,r) - d_{[a][a]}(r,r)
\right\}+1 .
\end{align}
These constants are useful for expressing the sampling variances of the estimators.

\begin{theorem}\label{thm:var_gen}
Under Assumptions \ref{asp:sutva_pe}--\ref{asmp:indist}, 
for treatments $r\neq r'\in \mathcal{R},$ the sampling variance of the subgroup average peer effect estimator is
\begin{align}\label{eq:var_gen_h}
\Var\left\{
\hat{\tau}_{[a]}({r}, {r}')
\right\} 
& =  n_{[a]}^{-1}\left\{ b_{[a]}({r})
S_{[a]}^2({r})
+ 
b_{[a]}({r}')
S_{[a]}^2({r}')
-
S_{[a]}^2({r}\text{-}{r}')
 \right\} \nonumber\\
&  \  + n_{[a]}^{-1}\left\{ 
c_{[a][a]}({r},{r}) 
\bar{Y}_{[a]}^2({r})
 + c_{[a][a]}(r',r')
\bar{Y}_{[a]}^2({r}') - 2c_{[a][a]}({r},{r}')
\Yhhrrprime
 \right\}, 
\end{align}
and the sampling variance of the average peer effect estimator is
\begin{align}
\label{eq:var_gen}
\Var\left\{
\hat{\tau}({r}, {r}')
\right\} 
& =  n^{-1}\sum_{a=1}^H w_{[a]} \left\{ 
b_{[a]}({r})
S_{[a]}^2({r})
+ 
b_{[a]}({r}')
S_{[a]}^2({r}') - S_{[a]}^2({r}\text{-}{r}')
 \right\} \nonumber\\
& \quad \  + n^{-1}\sum_{a=1}^H w_{[a]} \left\{ 
c_{[a][a]}({r},{r}) 
\bar{Y}_{[a]}^2({r})
 + c_{[a][a]}(r',r')
\bar{Y}_{[a]}^2({r}') - 2c_{[a][a]}({r},{r}')
\Yhhrrprime
 \right\} \nonumber \\
& \quad \   + n^{-1}\sum_{a=1}^H \sum_{a'\neq a}  (w_{[a]}w_{[a']})^{1/2}\left\{ 
c_{[a][a']}({r},{r}) 
\bar{Y}_{[a]}({r})  \bar{Y}_{[a']}({r})
 + c_{[a][a']}(r',r')
\bar{Y}_{[a]}({r}') \bar{Y}_{[a']}({r}')
 \right. \nonumber\\
& \quad \ \quad  \quad \quad \quad \quad  \quad \quad   \quad  \left. - c_{[a][a']}({r},{r}')
\bar{Y}_{[a]}({r}) \bar{Y}_{[a']}({r}') - c_{[a][a']}({r}',{r})
\bar{Y}_{[a]}({r}') \bar{Y}_{[a']}({r})
 \right\}.
\end{align}
\end{theorem}

From Theorem \ref{thm:var_gen}, the sampling variances of the peer effect estimators depend on the finite population variances of potential outcomes and individual peer effects, the products of two subgroup average potential outcomes, and 
the  product averages  $\Yhhrrprime$'s. 
In contrast to $\bar{Y}_{[a]}(r)\bar{Y}_{[a]}(r')$, 
the average $\Yhhrrprime$ 
excludes the product of two potential outcomes of the same unit.  
Note that we cannot unbiasedly estimate quantities involving $Y_i(r)Y_i(r')$ in general because we cannot jointly observe the potential outcomes, $Y_i(r)$ and $Y_i(r')$, for any unit $i$ and any treatments $r\neq r'$. 

Moreover, the sampling variance of $\hat{\tau}(r,r')$ is a weighted summation of the sampling variances of the $\hat{\tau}_{[a]}(r,r')$'s, corresponding to the first two terms in \eqref{eq:var_gen}, and the sampling covariances between $\hat{\tau}_{[a]}(r,r')$ and $\hat{\tau}_{[a']}(r,r')$, corresponding to the last double summation in \eqref{eq:var_gen}.

\subsection{Estimating the sampling variances}\label{sec:est_var_gen}
From Theorem \ref{thm:var_gen}, to estimate the sampling variances, we need to estimate the 
population quantities in \eqref{eq:var_gen_h} and \eqref{eq:var_gen}. 
For $1\leq a\leq H$, 
define
\begin{align}\label{eq:s_hr^2}
s_{[a]}^2(r) & =
\frac{n_{[a]}\pi_{[a]}^2(r)}{(n_{[a]}-1) \pi_{[a][a]}(r,r)}
\left\{
 \frac{n_{[a]}+c_{[a][a]}(r,r)}{n_{[a]}^2\pi_{[a]}(r)} \sum_{i:A_i=a}I(R_i=r)  Y_i^2-
\hat{Y}^2_{[a]}(r)
\right\}.
\end{align}

\begin{theorem}\label{thm:est_var_gen}
Under Assumptions \ref{asp:sutva_pe}--\ref{asmp:indist}, 
for $1\leq a, a'\leq H$ and $r, r'\in \mathcal{R}$, 
\begin{align*}
S_{[a]}^2(r) & = E\left\{ s_{[a]}^2(r) \right\},\\
\bar{Y}_{[a]}^2(r) & = E\left[
\frac{n_{[a]}\hat{Y}_{[a]}^2(r)-\{b_{[a]}(r)-1\}s_{[a]}^2(r)
	}{n_{[a]}+c_{[a][a]}(r,r)}
\right],\\
\Yhhrrprime
& = 
E\left\{  
\frac{n_{[a]}}{n_{[a]}-1}
\frac{\pi_{[a]}(r)\pi_{[a]}(r')}{\pi_{[a][a]}(r,r')}\hat{Y}_{[a]}(r)\hat{Y}_{[a]}(r') \right\},  & \text{if } r\neq r',\\
\bar{Y}_{[a]}(r)\bar{Y}_{[a']}(r') & =
E\left\{  \frac{\pi_{[a]}(r) \pi_{[a']}(r')}{\pi_{[a][a']}(r,r')}\hat{Y}_{[a]}(r)\hat{Y}_{[a']}(r') \right\},  & \text{if } a\neq a'.
\end{align*}
\end{theorem}

The estimators in Theorem \ref{thm:est_var_gen} correspond to the sample analogues of these finite population quantities, with carefully chosen coefficients to ensure unbiasedness.
Theorem \ref{thm:est_var_gen} guarantees that we have unbiased estimators for all terms in $\Var\{\hat{\tau}_{[a]}({r}, {r}')\}$ and $\Var\{\hat{\tau}({r}, {r}')\}$ except the variance of the individual peer effects  $S_{[a]}^2({r}\text{-}{r}')$. We cannot unbiasedly estimate $S_{[a]}^2({r}\text{-}{r}')$ from the observed data. This is analogous to other finite population causal inference \citep{Neyman:1923}. 
Because the coefficients of $S_{[a]}^2(r\text{-}r')$ in the variance formulas \eqref{eq:var_gen_h} and \eqref{eq:var_gen} are both negative, we can ignore the terms involving $S_{[a]}^2(r\text{-}r')$ and conservatively estimate the sampling variances by simply plugging in the estimators in Theorem \ref{thm:est_var_gen}. 
Note that $S_{[a]}^2(r\text{-}r')=0$ holds under \textit{additivity} defined below.

\begin{definition}\label{def:add_h}
The individual peer effects for units with attribute $a$ are additive if and only if  $\tau_i(r,r')=Y_i(r)-Y_i(r')$ is constant for each unit $i$ with attribute $a$, or, equivalently, $S_{[a]}^2(r\text{-}r')=0$. 
\end{definition}

Therefore, the final estimator for $\Var\{\hat{\tau}_{[a]}(r,r')\}$ is unbiased under additivity for $a$, and the final estimator for $\Var\{\hat{\tau}(r,r')\}$ is unbiased under additivity for all $1\leq a\leq H$.

\section{Inference for peer effects under complete randomization}\label{sec:design_based_complete_rand}

Under random partitioning, the formulas of $\pi_{[a]}(r)$, $\pi_{[a][a']}(r,r'), d_{[a][a']}(r,r'), b_{[a]}(r)$ and $c_{[a][a']}(r,r')$ are complicated, and so are the sampling variances of peer effect estimators. We relegate them to the Supplementary Material.
Fortunately, they have much simpler forms under complete randomization. In this section, we will focus on the inference under complete randomization.

\subsection{Treatment assignment under complete randomization}

The randomness in the peer effect estimators comes solely from the treatment assignments for all units, $(R_1,\ldots,R_n)$. Therefore, we need to first characterize the distribution of the treatments under complete randomization. 
Intuitively, the symmetry of complete randomization suggests that $(R_1,\ldots,R_n)$ has the same distribution as the treatment of a stratified randomized experiment.  The following proposition states this equivalence formally.

\begin{proposition}\label{prop:CR_stratified}
Under Assumptions \ref{asp:sutva_pe} and \ref{asp:exres_pe}, the 
complete randomization defined in  
Section \ref{sec:comp_rand} 
induces a stratified randomized experiment, in the sense that (1) for each $1\leq a\leq H$, in the stratum consisting of $n_{[a]}$ units with attribute $a$, $n_{[a]r}$ units receive treatment $r$ for any $r\in \mathcal{R}$, and any realization of treatments for these $n_{[a]}$ units has the same probability; and (2) the treatments of units are independent across strata. 
\end{proposition}

Proposition \ref{prop:CR_stratified} follows from the numerical implementation of the complete randomization described in Section \ref{sec:comp_rand}. It implies the formulas of $\pi_{[a]}(r)$, $\pi_{[a][a']}(r,r'),$  $d_{[a][a']}(r,r'),$ $b_{[a]}(r)$ and $c_{[a][a']}(r,r')$. We give a formal proof in the Supplementary Material. The group assignment for units with the same attribute $a$ induces a completely randomized experiment, with $n_{[a]r}$ units receiving treatment $r$. Moreover, the group assignments for units with different attributes are mutually independent.

\subsection{Point estimators for peer effects}
Proposition \ref{prop:CR_stratified} characterizes the treatment assignment of complete randomization, which allows us to express the peer effect estimators in simpler forms. 

\begin{corollary}\label{eg:tam_fix}
Under Assumptions \ref{asp:sutva_pe} and \ref{asp:exres_pe}, and under the complete randomization defined in 
Section \ref{sec:comp_rand}, 
for $1\leq a\leq H$ and $r\neq r'\in \mathcal{R}$, 
\begin{align}\label{eq:est_pe_fix}
\hat{Y}_{[a]}(r) & =  n_{[a]r}^{-1}\sum_{i:A_i=a,R_i=r}Y_i, \quad 
\hat{\tau}_{[a]}(r,r')  = \hat{Y}_{[a]}(r)
- \hat{Y}_{[a]}(r'), \quad 
\hat{\tau}(r,r')  = \sum_{a=1}^H w_{[a]} \hat{\tau}_{[a]}(r,r').
\end{align}
\end{corollary}

Therefore, under complete randomization, the unbiased estimator of the subgroup average peer effect, $\hat{\tau}_{[a]}(r,r')$, is the observed difference in outcome means under treatments $r$ and $r'$ for units with attribute $a$.

\subsection{Sampling variances of the peer effect estimators}

The sampling variances also have simpler forms under complete randomization. 

\begin{corollary}\label{cor:var_com_res}
Under Assumptions \ref{asp:sutva_pe} and \ref{asp:exres_pe}, and under the complete randomization defined in 
Section \ref{sec:comp_rand}, 
for $1\leq a\leq H$ and $r\neq r'\in \mathcal{R}$, 
\begin{align*}
\Var\left\{
\hat{\tau}_{[a]}(r,r')
\right\} & = 
\frac{S_{[a]}^2({r})}{n_{[a]r}}
+ 
\frac{S_{[a]}^2({r}')}{n_{[a]r'}}
  -\frac{S_{[a]}^2({r}\text{-}{r}')}{n_{[a]}},\\
 \Var\left\{
\hat{\tau}({r}, {r}')
\right\}
& = \sum_{a=1}^H w_{[a]}^2
\left\{
\frac{S_{[a]}^2({r})}{n_{[a]r}}
+ 
\frac{S_{[a]}^2({r}')}{n_{[a]r'}}
  -\frac{S_{[a]}^2({r}\text{-}{r}')}{n_{[a]}}
\right\}.
\end{align*}
\end{corollary}

From Corollary \ref{cor:var_com_res}, the variance formula of the subgroup average peer effect estimator 
under complete randomization is the same as that for classical completely randomized experiments with multiple treatments \citep{Neyman:1923}. This follows from the equivalence relationship in Proposition \ref{prop:CR_stratified}. 
Corollary \ref{cor:var_com_res} also implies that  $\Var\{\hat{\tau}(r,r')\} \equiv
\Var\{\sum_{a=1}^{H} w_{[a]} \hat{\tau}_{[a]}(r,r') \} = 
\sum_{a=1}^H w_{[a]}^2 \Var\{\hat{\tau}_{[a]}(r,r')\}$. This follows from the mutual independence of $\{ \hat{\tau}_{[a]}(r,r'): 1\leq a\leq H \}$ in an experiment stratified on attributes. 

From Corollary \ref{cor:var_com_res}, the $n_{[a]r}$'s are the effective sample sizes. One the one hand, this is intuitive because they are the sample sizes of the stratified experiment described in Proposition \ref{prop:CR_stratified}. One the other hand, this is counterintuitive because units in the same group have correlated observed outcomes. However, this correlation does not diminish the effective sample sizes in contrast to the correlation in standard group-randomized experiments. Units in the same group could potentially be in a different group under a different realization of the treatment assignment. The probability that two given units are in the same group decreases as $n$ increases, and so does the correlation between their observed outcomes.

\subsection{Estimating the sampling variances}\label{sec:var_est_com_rand}

From Proposition \ref{prop:CR_stratified}, $\hat{Y}_{[a]}(r) = \sum_{i:A_i=a,R_i=r}Y_i/n_{[a]r}$ reduces to the sample mean, and
\begin{align}\label{eq:observed_sample_variance}
s_{[a]}^2(r)
& =  \frac{n_{[a]r}}{n_{[a]r}-1}\left\{
\frac{1}{n_{[a]r}}\sum_{i:A_i=a,R_i=r} Y_i^2 -\hat{Y}^2_{[a]}(r)
\right\} = (n_{[a]r}-1)^{-1}\sum_{i:A_i=a, R_i=r}\left\{
Y_i - \hat{Y}_{[a]}(r)
\right\}^2
\end{align}
reduces to the sample variance of the observed outcomes for units with attribute $a$ receiving treatment $r$.
Formula \eqref{eq:observed_sample_variance}, coupled with Corollary \ref{cor:var_com_res}, simplifies the variance estimators under complete randomization, which coincide with \citet{Neyman:1923}'s conservative variance estimators under classical completely randomized experiments with multiple treatments.

\begin{corollary}\label{cor:var_est_com_res}
Under Assumptions \ref{asp:sutva_pe} and \ref{asp:exres_pe}, and under the complete randomization defined in 
Section \ref{sec:comp_rand}, 
for $1\leq a\leq H$ and $r\neq r'\in \mathcal{R}$,
 the variance estimators become
\begin{align}\label{eq:var_est_CRFG}
\hat{V}_{[a]}({r}, {r}')
& = 
\frac{s_{[a]}^2({r})}{n_{[a]r}}
+ 
\frac{s_{[a]}^2({r}')}{n_{[a]r'}}, \quad 
\hat{V}({r}, {r}')
 = \sum_{a=1}^H w_{[a]}^2
\left\{
\frac{s_{[a]}^2({r})}{n_{[a]r}}
+ 
\frac{s_{[a]}^2({r}')}{n_{[a]r'}}
\right\}.
\end{align}
Moreover, $E\{\hat{V}_{[a]}({r}, {r}')\} - \Var\{\hat{\tau}_{[a]}(r,r')\} = n_{[a]}^{-1}S_{[a]}^2(r\text{-}r')\geq 0$, which becomes zero under additivity for $a$, 
and $E\{\hat{V}({r}, {r}')\} - \Var\{\hat{\tau}(r,r')\} = n^{-1}\sum_{a=1}^{H}w_{[a]} S_{[a]}^2(r\text{-}r') \geq 0$, which becomes zero under additivity for all $1\leq a\leq H$. 
\end{corollary}

\subsection{Asymptotic distributions and confidence intervals for peer effects}\label{sec:dist_ci_single_estimator}
The asymptotic analysis embeds the $n$ units into a sequence of finite populations with increasing sizes. See  \citet{fcltxlpd2016} for a review of finite population asymptotics in causal inference.
Under complete randomization, if some regularity conditions hold, then $\hat{\tau}_{[a]}(r,r')$ is asymptotically Normal. We can then construct 
a $1-\alpha$ Wald-type confidence interval for $\tau_{[a]}(r,r')$: $\hat{\tau}_{[a]}(r,r') \pm q_{1-\alpha/2}\hat{V}^{1/2}_{[a]}({r}, {r}')$, with $q_{1-\alpha/2}$ being the $(1-\alpha/2)$th quantile of $\mathcal{N}(0,1)$. 
Because the variance estimator $\hat{V}_{[a]}({r}, {r}')$ in \eqref{eq:var_est_CRFG} overestimates the true sampling variance on average, the  confidence interval is asymptotically conservative, with the limit of coverage probability larger than or equal to the nominal level. 
Analogously, 
we can construct asymptotically conservative confidence intervals for $\tau({r}, {r}')$. 
We formally state the regularity condition as follows.

\begin{condition}\label{cond:fp}
For any $1\leq a\leq H, r\neq r'\in \mathcal{R}$, as $n\rightarrow \infty$,
\begin{itemize}
\item[(i)] the proportions,  
$w_{[a]}$ and $n_{[a]r}/n_{[a]},$ have positive limits, 
\item[(ii)] the finite population variances of potential outcomes and individual peer effects, $S_{[a]}^2(r)$ and $S_{[a]}^2(r\text{-}r')$, have limits, and at least one of the limits of $\{  S_{[a]}^2({r}): {r}\in\mathcal{R} \} $ are non-zero, 
\item[(iii)] $\max_{i:A_i=a}|Y_i(r)-\bar{Y}_{[a]}(r)|^2/n_{[a]}\rightarrow 0$. 
\end{itemize}
\end{condition}

Conditions (i) and (ii) are natural in most applications. In our motivating application, 
GPA is bounded within $[0,4]$, and therefore condition (iii) holds automatically \citep{fcltxlpd2016}.
We summarize the asymptotic results below. 

\begin{theorem}\label{thm:clt_comp_rand}
Under Assumptions \ref{asp:sutva_pe} and \ref{asp:exres_pe}, and under the complete randomization defined in Section \ref{sec:comp_rand}, 
if Condition \ref{cond:fp} holds, 
then, 
for any $1\leq a\leq H$ and $r\neq r'\in \mathcal{R}$,  
\begin{itemize}
\item[(a)] $\hat{\tau}_{[a]}(r,r')$ and $\hat{\tau}(r,r')$ are asymptotically Normal, 
\item[(b)] the Wald-type confidence intervals for ${\tau}_{[a]}(r,r')$ and $\tau(r,r')$ are asymptotically conservative, unless the peer effects are additive for units with the same attribute.  
\end{itemize}
\end{theorem}

\subsection{Randomization-based and regression-based analyses}
Theorem \ref{thm:clt_comp_rand} is purely randomization-based without any modeling assumptions of the outcomes. Regression-based analysis is also popular in practice. Suppose we fit a linear model for the observed outcomes
\begin{align}\label{eq:linear_model}
Y_i = \mu + \alpha_{[A_i]} + \beta_{R_i} + \lambda_{[A_i]R_i} + \varepsilon_i,
\end{align}
where $\mu$ is the intercept, 
$\alpha_{[a]}$ represents the main ``effect" of attribute $a$, 
$\beta_{r}$ represents the main effect of treatment $r$, 
and $\lambda_{[a]r}$ represents the interaction between attribute $a$ and treatment $r$.
The traditional linear regression assumes that the error terms follow independent zero-mean (Normal) distributions and generate
the randomness of the observed outcomes. 

Under model \eqref{eq:linear_model}, we need some constraints to avoid over-parameterization:
$
\sum_{a=1}^H \alpha_{[a]} = 0,$
$
\sum_{r\in \mathcal{R}} \beta_r = 0,
$
$
\sum_{a=1}^H \lambda_{[a]r}=0,$
and
$
\sum_{r\in \mathcal{R}} \lambda_{[a]r}=0$,
for $ 1\leq a\leq H$ and $ r\in \mathcal{R}$.
Let $\mu_{[a]r}=\mu+\alpha_{[a]}+\beta_r+\lambda_{[a]r}$. Then we can interpret $\mu_{[a]r}-\mu_{[a]r'}$ as the subgroup average peer effect of treatment $r$ versus $r'$ for units with attribute $a$. 
The least squares estimators of the coefficients are 
\begin{align*}
\hat{\mu} & =   \frac{1}{H|\mathcal{R}|}\sum_{a=1}^H \sum_{r\in \mathcal{R}} \hat{Y}_{[a]}(r), & 
\hat{\alpha}_{[a]}  & =  \frac{1}{|\mathcal{R}|}\sum_{r\in \mathcal{R}} \hat{Y}_{[a]}(r) -\hat{\mu},\\
\hat{\beta}_r & =  \frac{1}{H}\sum_{a=1}^H  \hat{Y}_{[a]}(r) - \hat{\mu}, &   
\hat{\lambda}_{[a]r}  & =  \hat{Y}_{[a]}(r)-(\hat{\mu}+\hat{\alpha}_{[a]}+\hat{\beta}_r), \quad (1\leq a\leq H, r\in\mathcal{R}).
\end{align*}
Then we have the following proposition. 

\begin{proposition}\label{prop:equiva_linear}
Under the linear model (\ref{eq:linear_model}), for any $1\leq a\leq H$ and $r\neq r'\in \mathcal{R}$,  the least squares estimator for the subgroup average peer effect is
\begin{align*}
\hat{\mu}_{[a]r} - \hat{\mu}_{[a]r'} &  = (\hat{\mu}+\hat{\alpha}_{[a]}+\hat{\beta}_r+\hat{\lambda}_{[a]r})-(\hat{\mu}+\hat{\alpha}_{[a]}+\hat{\beta}_{r'}+\hat{\lambda}_{[a]r'})= \hat{Y}_{[a]}(r) -\hat{Y}_{[a]}(r') = \hat{\tau}_{[a]}(r,r'),
\end{align*}
with the Huber--White variance estimator  
\begin{align*}
\hat{V}_{[a],\text{HW}}( r,r') 
& = \frac{n_{[a]r}-1}{n_{[a]r}}\frac{s_{[a]}^2(r)}{n_{[a]r}}+
\frac{n_{[a]r'}-1}{n_{[a]r'}}\frac{s_{[a]}^2(r')}{n_{[a]r'}}
 \approx \frac{s_{[a]}^2(r)}{n_{[a]r}} + \frac{s_{[a]}^2(r')}{n_{[a]r'}} = \hat{V}_{[a]}(r,r') . 
\end{align*}
\end{proposition}

The linear outcome model (\ref{eq:linear_model}) includes the {\it interaction} between the unit's attribute $A_i$ and the treatment received $R_i$. Under (\ref{eq:linear_model}), both the point estimator and Huber--White variance estimator for the subgroup average peer effect are (nearly) identical to the randomization-based ones under complete randomization. Complete randomization justifies this regression-based analysis for peer effects. This result extends \citet{lin2013} for classical completely randomized experiments. However, such equivalence generally does not hold if the treatment assignment is not complete randomization, nor if we use the conventional variance estimator in linear models assuming homoscedasticity of the error terms.

Related to the discussion of effective sample sizes after Proposition \ref{cor:var_com_res}, we do {\it not} need to use cluster-robust standard errors even though some units are in the same group or cluster. Our inference depends solely on the random assignment of peers in contrast to model-based inferences \citep[e.g.,][]{carrell2013natural}. In our setting, randomization does {\it not} justify cluster-robust standard errors. Our view is similar to \citet{abadie2017should} in a different context. 

Many econometric analyses of peer effects did not include the interaction term \citep[e.g.,][]{Sacerdote01052001, carrell2013natural}. Complete randomization does {\it not} justify them in the presence of treatment effect heterogeneity. Sometimes, peers' outcomes also enter the right-hand side of the regression in \eqref{eq:linear_model}. It is then more difficult to interpret their least squares coefficients as causal effects estimators \citep{manski1993identification, angrist2014perils}.

\subsection{Asymptotic distributions and confidence sets for multiple peer effects}

Below we study the joint asymptotic sampling distribution of multiple average peer effect estimators. It is useful for constructing confidence sets and testing significance of multiple average peer effects simultaneously.
For mathematical convenience, we center the potential outcomes:
\begin{align*}
& \theta_i(r) = Y_i(r) - |\mathcal{R}|^{-1}\sum_{r'\in \mathcal{R}}Y_i(r'),\quad 
& & \theta_{[a]}(r) = n_{[a]}^{-1}\sum_{i:A_i=a}\theta_i(r) = \bar{Y}_{[a]}(r) -|\mathcal{R}|^{-1}\sum_{r'\in \mathcal{R}}\bar{Y}_{[a]}(r'), \\
& \theta_i(\mathcal{R}) =(\theta_i(r_1), \ldots, \theta_i(r_{|\mathcal{R}|}) )^\top, \quad 
& & \theta_{[a]}(\mathcal{R}) =(\theta_{[a]}(r_1), \ldots, \theta_{[a]}(r_{|\mathcal{R}|}) )^\top.
\end{align*}
Let
$\hat{\theta}_{[a]}(r) = \hat{Y}_{[a]}(r) -|\mathcal{R}|^{-1}\sum_{r'\in \mathcal{R}}\hat{Y}_{[a]}(r')$ 
be the centered subgroup average potential outcome estimator, vectorized as
$\hat{\theta}_{[a]}(\mathcal{R})=(\hat{\theta}_{[a]}(r_1), \ldots, \hat{\theta}_{[a]}(r_{|\mathcal{R}|}) )^\top$. For any $r,r'\in \mathcal{R}$, the individual peer effect $\tau_i(r,r')$, the subgroup average peer effect $\tau_{[a]}(r,r')$, and the subgroup average peer effect estimator $\hat{\tau}_{[a]}(r,r')$ are the same linear transformations of $\theta_i(\mathcal{R}), \theta_{[a]}(\mathcal{R})$ and $\hat{\theta}_{[a]}(\mathcal{R})$, respectively. 
Therefore, it suffices to study the joint asymptotic sampling distribution of the $\hat{\theta}_{[a]}(\mathcal{R})$'s for all $a$, and construct confidence sets for $\theta_{[a]}(\mathcal{R})$'s and their linear transformations.

Define ${\Gamma} = {I}_{|\mathcal{R}|} - |\mathcal{R}|^{-1}{1}_{|\mathcal{R}|}{1}_{|\mathcal{R}|}^\top$ as an  $|\mathcal{R}|\times |\mathcal{R}|$ projection matrix orthogonal to $1_{|\mathcal{R}|}$. The theorem below summarizes the results for the joint inference of multiple peer effects. 

\begin{theorem}\label{thm:asym_adjust_treat_effect}
Under Assumptions \ref{asp:sutva_pe} and \ref{asp:exres_pe}, the complete randomization defined in Section \ref{sec:comp_rand}, 
and Condition \ref{cond:fp}, (a) $\hat{\theta}_{[1]}(\mathcal{R}), \ldots,$  $\hat{\theta}_{[H]}(\mathcal{R})$ are mutually independent; (b) $\hat{\theta}_{[a]}(\mathcal{R})$ is unbiased for $\theta_{[a]}(\mathcal{R})$ with sampling covariance $\Cov\{\hat{{\theta}}_{[a]}(\mathcal{R})\}$ as follows: 
\begin{align*} 
\Gamma ~
\text{diag}\left\{ 
n_{[a]r_1}^{-1}S^2_{[a]}(r_1),\ldots,  
n_{[a]r_{|\mathcal{R}|}}^{-1}S_{[a]}^2(r_{|\mathcal{R}|})\right\} 
 \Gamma
- 
\frac{1}{n_{[a]}(n_{[a]}-1)} 
\sum_{i:A_i=a} 
\left\{\theta_i(\mathcal{R}) - \theta_{[a]}(\mathcal{R})\right\}\left\{\theta_i(\mathcal{R}) - \theta_{[a]}(\mathcal{R})\right\}^\top; 
\end{align*}
(c) 
$\hat{{\theta}}_{[a]}(\mathcal{R})-{\theta}_{[a]}(\mathcal{R})$ is asymptotically Normal with mean 0 and covariance $\Cov\{\hat{{\theta}}_{[a]}(\mathcal{R})\}$; 
(d) the covariance estimator 
\begin{align}\label{eq:est_var_adj_treat_effect}
\widehat{\Cov}\{\hat{{\theta}}_{[a]}(\mathcal{R})\} = 
\Gamma~
\text{diag}\left\{
n_{[a]r_1}^{-1} s^2_{[a]}(r_1), \ldots, n_{[a]r_{|\mathcal{R}|}}^{-1} s_{[a]}^2(r_{|\mathcal{R}|})
\right\}
\Gamma 
\end{align}
is conservative in expectation, unless the peer effects are additive for units with attribute $a$. 
\end{theorem}

From Theorem \ref{thm:asym_adjust_treat_effect}, we can then obtain the Wald-type asymptotic conservative confidence sets for $(\theta_{[1]}(\mathcal{R})^\top, \ldots, \theta_{[H]}(\mathcal{R})^\top)^\top$ and their linear transformations, including multiple average or subgroup average peer effects as special cases.

\section{Optimal treatment assignment mechanism}\label{sec:opt_tre_asign}
\subsection{Point estimator for the optimal treatment assignment mechanism}\label{sec:opt_tre_asign_est}

The results in previous sections are useful for decision making. We can use them to find the optimal treatment assignment mechanism for a new population of size $n'=m'(K+1)$. 
We need to assume that the new population is similar to the one in our data in some way. Otherwise, we cannot draw any conclusions in general. For instance, we assume that the subgroup average potential outcomes in the new population are linear transformations of those in our data, i.e. $
\bar{Y}'_{[a]}(r)=\yscale \bar{Y}_{[a]}(r)+\yshift
$
for some $\yscale > 0$ and $\yshift$, for all $1\leq a\leq H$ and $r \in \mathcal{R}$. 
In our motivating application, the new population usually consists of the students coming next year. The scale parameter $\yscale$ and shift parameter $\yshift$ can explain the proportional change and the absolute change of average GPAs across different years. These changes are possibly due to the difference in qualities of students and difficulties of exams across years.

We use complete randomization with ${L'}(z)$ fixed at some vector ${l'}$ for the new population. 
Our goal is to find $l'=(l'_{1}, \ldots, l'_{T})$ to maximize the expected total outcome.
We are looking for the optimal $l'$ of complete randomization, and the final assignment $Z'$ is still random. 
For any $1\leq a\leq H$ and $r\in \mathcal{R}$, let $n'_{[a]r}$ be the number of units  with attribute $a$ receiving treatment $r$ in the new population under complete randomization with ${L'}(z)$ fixed at $l'$. Proposition \ref{prop:CR_stratified} and \eqref{eq:n_hr} have the following useful implications.
First, $n'_{[a]r}$ is a deterministic function of ${l'}$. Second, within each stratum consisting of $n_{[a]}'$ units with attribute $a$, we randomly assign $n'_{[a]r}$ units to treatment $r$ for any $r\in \mathcal{R}$. Third, the treatments for units in different strata are mutually independent. Based on these, the expected total outcome under complete randomization  with ${L'}(z)=l'$ is
\begin{align}\label{eq:ex_total_outcome_new}
E\left(\sum_{i=1}^{{n'}} Y'_i\right) 
& = 
\sum_{a=1}^H E\left(\sum_{i:A_i=a} Y'_i\right) 
=
\sum_{a=1}^H \sum_{r\in \mathcal{R}} n'_{[a]r} \bar{Y}'_{[a]}(r) 
= \yscale \sum_{a=1}^H \sum_{r\in \mathcal{R}} n'_{[a]r} \bar{Y}_{[a]}(r) +
\sum_{a=1}^H n'_{[a]} \yshift, 
\end{align}
where the last equality follows from  
$
\bar{Y}'_{[a]}(r)=\yscale \bar{Y}_{[a]}(r)+ \yshift
$ and $\sum_{r\in \mathcal{R}}n_{[a]r}'=n'_{[a]}$. Although the expected total outcome (\ref{eq:ex_total_outcome_new}) of the new population depends on the unknown constants $\yscale >  0 $ and $\yshift$'s, 
the maximizer $l'_{\text{opt}}$
for this expected total outcome is the same as that for $\sum_{a=1}^H \sum_{r\in \mathcal{R}} n'_{[a]r} \bar{Y}_{[a]}(r)$. Moreover, we can unbiasedly estimate $\sum_{a=1}^H \sum_{r\in \mathcal{R}} n'_{[a]r} \bar{Y}_{[a]}(r)$ by replacing $\bar{Y}_{[a]}(r)$ with the corresponding unbiased estimator $\hat{Y}_{[a]}(r)$, and then use $\hat{l}'_{\text{opt}}$ that maximizes the unbiased estimator $\sum_{a=1}^H \sum_{r\in \mathcal{R}} n'_{[a]r} \hat{Y}_{[a]}(r)$ as an estimator for $l'_{\text{opt}}$. 
From \eqref{eq:n_hr}, the objective function reduces to 
\begin{align}\label{eq:obj_linear_fun}
\sum_{a=1}^H \sum_{r\in \mathcal{R}} n'_{[a]r} \hat{Y}_{[a]}(r)
& = \sum_{a=1}^H \sum_{r\in \mathcal{R}} 
\left\{
\sum_{t=1}^T I(g_t=\{a\}\cup r)l_t'g_t(a)
\right\}
\hat{Y}_{[a]}(r) \nonumber \\
& =  \sum_{t=1}^T \left\{\sum_{a=1}^H \sum_{r\in \mathcal{R}}I(g_t=\{a\}\cup r)g_t(a)\hat{Y}_{[a]}(r) \right\} l_t',
\end{align}
which is a linear function of $l'=(l'_1, \ldots, l'_T)$. 
In \eqref{eq:obj_linear_fun}, 
the coefficient of $l_t'$ is the estimated total outcome of $K+1$ units in the group with group attribute $g_t.$ 
The constraints on $l'$ include that all the $l_t'$'s are non-negative integers, and the number of units with attribute $a$ implied by $l'$ is fixed:  
\begin{align}\label{eq:obj_constrain_linear}
\begin{cases}
\sum_{t=1}^T g_t(a) l_t' = n_{[a]}', & (a = 1, \ldots, H), \\
l_t' \geq 0, \qquad   l_t' \text{ is an integer, } & (t = 1, \ldots, T). 
\end{cases}
\end{align}
Both the objective function \eqref{eq:obj_linear_fun} and the constraints \eqref{eq:obj_constrain_linear} are linear in $l'$. Therefore, finding the maximizer $\hat{l}'$ is a linear integer programming problem. When the sample size is not too large, we can enumerate all possible values of $l'$ to obtain the maximizer.

\citet{bhattacharya2009inferring} discussed the optimal peer assignment in a super population scenario where each unit has only one peer $(K=1)$. In that case, the optimization problem becomes a linear programming problem without the integer constraint.

\subsection{Inference for the optimal assignment mechanism}

Section \ref{sec:opt_tre_asign_est} gives a point estimator of the optimal treatment assignment mechanism. 
The uncertainty of the point estimator comes from the uncertainty of the $\hat{Y}_{[a]}(r)$'s.  
Below we construct confidence sets for the optimal $l'_{\text{opt}}$ of complete randomization. 
Note that $\sum_{r\in \mathcal{R}}n_{[a]r}'=n_{[a]}'$ is fixed for all $a$. 
By the definitions of $\theta_{[a]}(r)$ and $\hat{\theta}_{[a]}(r)$, 
the true and estimated optimal treatment assignment mechanisms satisfy 
\begin{align}\label{eq:relation_l_theta}
l'_{\text{opt}} & = 
\argmax_{l'\in \mathcal{L}'}\sum_{a=1}^H \sum_{r\in \mathcal{R}} n'_{[a]r} \bar{Y}_{[a]}(r) = 
\argmax_{l'\in \mathcal{L}'}\sum_{a=1}^H \sum_{r\in \mathcal{R}} n'_{[a]r} \theta_{[a]}(r), \\ 
\hat{l}'_{\text{opt}} & = 
\argmax_{l'\in \mathcal{L}'}\sum_{a=1}^H \sum_{r\in \mathcal{R}} n'_{[a]r} \hat{Y}_{[a]}(r)
=
\argmax_{l'\in \mathcal{L}'}\sum_{a=1}^H \sum_{r\in \mathcal{R}} n'_{[a]r} \hat{\theta}_{[a]}(r),
\nonumber
\end{align}
where $\mathcal{L}'$ denotes the set of all possible $l'$ satisfying the constraint \eqref{eq:obj_constrain_linear}. 
Here we represent $l'_{\text{opt}}$ using the centered subgroup average potential outcomes, because the $\theta_{[a]}(r)$'s have simpler asymptotically conservative confidence sets, as shown in Theorem \ref{thm:asym_adjust_treat_effect}. 
For any $\alpha\in (0,1)$, 
let $\mathcal{C}_{[a]}(\alpha)$ be the $1-\alpha$ Wald-type asymptotic conservative confidence set for ${\theta}_{[a]}(\mathcal{R})$. 
Then a $1-\alpha$ asymptotic conservative confidence set for $l'_{\text{opt}}$ is 
\begin{align}\label{eq:conf_set_l}
\left\{
\argmax_{l'\in \mathcal{L}'}\sum_{a=1}^H \sum_{r\in \mathcal{R}} n'_{[a]r} \bar{\theta}_{[a]}(r):   \bar{{\theta}}_{[1]}(\mathcal{R}) \in \mathcal{C}_{[1]}(\bar{\alpha}), \ldots, \bar{{\theta}}_{[H]}(\mathcal{R})  \in \mathcal{C}_{[H]}(\bar{\alpha}), \bar{\alpha}= 1 - (1-\alpha)^{1/H}
\right\}. 
\end{align}

The confidence set in \eqref{eq:conf_set_l} involves solving infinite linear integer programming problems, which are computationally intensive. 
More importantly, the interpretation of the confidence set in \eqref{eq:conf_set_l} seems unnatural for making decisions in the future because the ``confidence'' statement is a property over repeated sampling of the previous experiment. Ideally, we need to make future decisions conditioning on the observed data rather than
averaging over them. 
Below we use the ``fiducial distribution" \citep{Fisher:1935, dasgupta2014causal} of $l_{\text{opt}}'$. 

We start with the asymptotic sampling distribution in Theorem \ref{thm:asym_adjust_treat_effect}, and then swap the roles of the estimators and estimands. 
Let $\tilde{\theta}_{[a]}(\mathcal{R}) \equiv 
(\tilde{\theta}_{[a]}(r_1), \ldots, \tilde{\theta}_{[a]}(r_{|\mathcal{R}|}))$
be 
a multivariate Normal distribution with mean $\hat{\theta}_{[a]}(\mathcal{R})$ and covariance  $\widehat{\Cov}\{\hat{{\theta}}_{[a]}(\mathcal{R})\}$ in \eqref{eq:est_var_adj_treat_effect}, independently for $1\leq a\leq H$. 
The fiducial distribution of $l_{\text{opt}}'$ is the distribution of $\tilde{l}_{\text{opt}}'\equiv \argmax_{l'\in \mathcal{L}'}\sum_{a=1}^H \sum_{r\in \mathcal{R}} n'_{[a]r} \tilde{\theta}_{[a]}(r)$, with $\theta_{[a]}(r)$ in \eqref{eq:relation_l_theta} replaced by the {\it random vector}
$\tilde{\theta}_{[a]}(\mathcal{R})$. 
When there exist multiple maximizers for $\tilde{l}_{\text{opt}}'$, we randomly choose one of them with equal probability. 
Computationally, to simulate from $\tilde{l}_{\text{opt}}'$, 
we can first simulate the $\tilde{\theta}_{[a]}(\mathcal{R})$'s independently from Normal distributions and then calculate $\tilde{l}_{\text{opt}}'$. 
Compared to confidence sets, the "fiducial distribution" not only acts as a computational compromise but also has a natural Bayesian interpretation. We can view $\tilde{\theta}_{[a]}(\mathcal{R})$ as the Bayesian posterior distribution of $\theta_{[a]}(\mathcal{R})$ based on the sampling distribution in Theorem \ref{thm:asym_adjust_treat_effect} under a flat prior. Consequently, $\tilde{l}_{\text{opt}}'$ is the Bayesian posterior distribution  of $l_{\text{opt}}'$.
Rigorously, this is not a full Bayesian procedure but only a limited information Bayesian procedure \citep{kwan1999asymptotic, sims2006example}. It uses ``limited information'' from the asymptotic randomization distribution, but it does not impose a full outcome model for all units. Therefore, this ``fiducial distribution'' enjoys not only the robustness of randomization inference without outcome modeling but also the interpretability of Bayesian inference for decision making.

\section{Application to roommate assignment in a university in China}
\label{sec:edu_example}

\subsection{Overview of the data}
The data set consists of college students graduating in 2013 and 2014 from 25
departments of a university in China. The university assigns dorms to 
departments, and the 
departments then assign students to dorms. 
The roommate assignment is close to random partitioning for students of the same 
department, gender and graduating year.

Among students graduating in 2013 and 2014, $73.9\%$ of them were from Gaokao, $21.7\%$ were from recommendation, and $4.4\%$ were from other ways (omitted in our analysis). 
For example, $43.2\%$ students are from recommendation in the mathematics department, $45.2\%$ in the physics department, 
$52.5\%$ in the chemistry department, 
$22.4\%$ in the biology department, and 
$28.4\%$ in the informatics department. 
The outcome is the freshman year GPA. 
The average freshman GPAs are $3.26$ and $3.37$ for students from Gaokao and recommendation, respectively. On average, students from recommendation do better  than students from Gaokao during the freshman year. 
We first want to understand the effect of roommate types on students' academic performance. We then need to design an optimal roommate assignment.

\subsection{Point and interval estimators for peer effects}\label{sec:edu_point_interval}

The student room assignment is conditional on department, gender and the graduating year.
We focus on the male students graduating in 2013 from the Departments of Informatics and Physics separately. These two departments have larger sample sizes. Moreover, there is a close connection between the training in high school Olympiads and the freshman introductory courses in these two departments. 
The informatics department has 104 students from Gaokao and 52 students from recommendation. The physics department has 49 students from Gaokao and 43 students from recommendation.
If we conduct inference conditioning on the numbers of groups with different group attribute sets, 
the treatment assignment mechanism is equivalent to 
complete randomization. 
Recall that students from Gaokao have attribute 1, and students from recommendation have attribute 2. 
Under Assumptions \ref{asp:sutva_pe} and \ref{asp:exres_pe}, we have four treatments, contained in
$\mathcal{R}=\{r_1,r_2,r_3,r_4\}=\{111, 112, 122, 222\}$.

\begin{table}[htb]
\centering
\caption{Estimated peer effects with $\mathcal{R}=\{r_1,r_2,r_3,r_4\}=\{111, 112, 122, 222\}$. $\hat{\tau}, \hat{\tau}_{[1]}$ and $\hat{\tau}_{[2]}$ are the point estimators, and the numbers in the parentheses are the estimated standard errors. The bold numbers correspond to those peer effects significantly different from zero at level 0.05.}\label{table:est_pe}
\begin{tabular}{cccccccc}
\toprule
Department & Estimator & $(r_1,r_2)$ & $(r_1,r_3)$ & $(r_1,r_4)$ & $(r_2,r_3)$ & $(r_2,r_4)$ & $(r_3,r_4)$\\[5pt]
\midrule
Informatics & $\hat{\tau}$ & $-0.117$ & $0.008$ & ${\bf -0.313}$ & $0.125$ & ${\bf -0.196}$ & ${\bf -0.321}$\\
&  & {\small (0.086)} & {\small (0.109)} & {\small (0.071)} & {\small (0.101)} & {\small (0.059)} & {\small (0.089)}\\
 & $\hat{\tau}_{[1]}$ & $-0.117$ &  $0.074$ &  ${\bf -0.285}$ & $0.191$ & ${\bf -0.168}$ & ${\bf -0.359}$\\
& & {\small (0.112)} & {\small (0.152)} & {\small (0.087)} & {\small (0.145)} & {\small (0.074)} & {\small (0.127)}\\
& $\hat{\tau}_{[2]}$ & $-0.119$ & $-0.125$ & ${\bf -0.369}$ & $-0.006$ & ${\bf -0.250}$ & ${\bf -0.244}$\\
& & {\small (0.126)} & {\small (0.120)} & {\small (0.123)} & {\small (0.092)} & {\small (0.096)} &  {\small (0.088)}\\[5pt]
\midrule
Physics & $\hat{\tau}$ & 0.172 & $-0.108$ & $-0.095$ & ${\bf -0.280}$ & $-0.267$ & 0.013\\
& & {\small (0.142)} & {\small (0.140)} & {\small (0.177)} & {\small (0.103)} & {\small (0.150)} & {\small (0.148)}\\
 & $\hat{\tau}_{[1]}$ &  $0.329$ & $-0.099$ & $0.017$ & ${\bf -0.427}$ & $-0.311$ & $0.116$\\
& & {\small (0.185)} & {\small (0.167)} & {\small (0.268)} & {\small (0.145)} & {\small (0.255)} & {\small (0.243)} \\
& $\hat{\tau}_{[2]}$ & $-0.007$ & $-0.119$ & $-0.222$ & $-0.112$ & $-0.215$ & $-0.103$\\
& & {\small (0.219)} & {\small (0.231)} & {\small (0.225)} & {\small (0.145)} & {\small (0.135)} & {\small (0.153)}\\
\bottomrule
\end{tabular}
\end{table}

Table \ref{table:est_pe} shows the estimated average peer effects for these two departments with estimated standard errors based on Corollary \ref{cor:var_est_com_res}. Treatment $r_4$ is significantly better than other treatments for students in the informatics department. Treatment $r_3$ is significantly better than treatment $r_2$ for students from Gaokao in the physics department.

\begin{figure}[t]
\centering
\begin{subfigure}{.35\textwidth}
  \centering
  \includegraphics[width=1\linewidth]{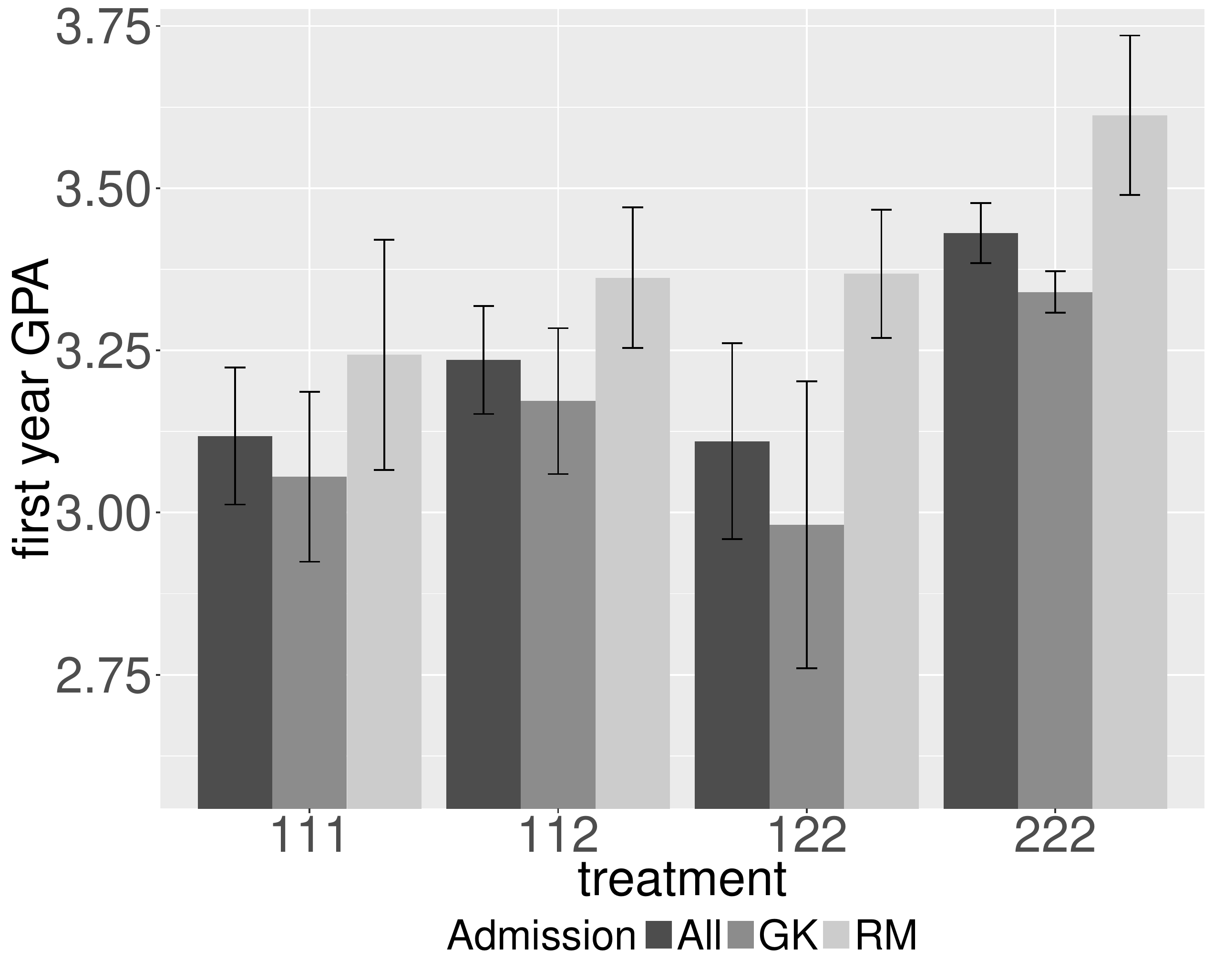}
   \caption{
    \centering Informatics}
  \label{fig:star_coverage_all}
\end{subfigure}%
\begin{subfigure}{.35\textwidth}
  \centering
  \includegraphics[width=1\linewidth]{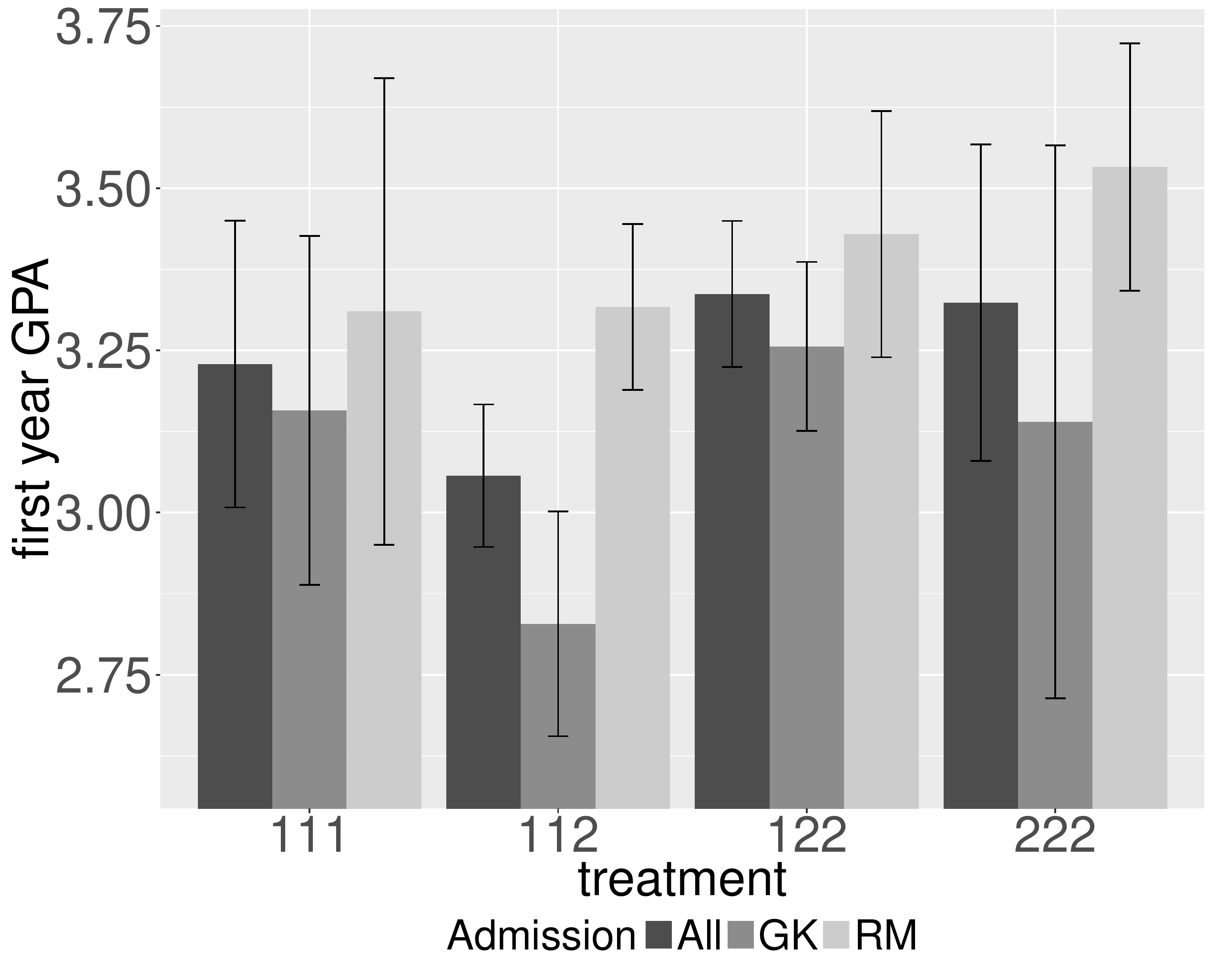}
   \caption{
    \centering Physics}
    \label{fig:star_ci_length}
\end{subfigure}
\caption{
Estimated average potential outcomes with $\mathcal{R}  = \{111, 112, 122, 222\}$. 
The black, grey and light-grey bars correspond to all students, students from Gaokao, and students from recommendation, respectively. The solid lines denote the $95\%$ confidence intervals. 
}\label{fig:est_ave_po}
\end{figure}

Figure \ref{fig:est_ave_po} shows the estimated  average potential outcomes $\hat{Y}(r)$ and $\hat{Y}_{[a]}(r)$, as well as their 95\% confidence intervals for all possible $a$ and $r$. 
It displays some interesting results. Students from recommendation in both departments have higher average GPAs if they have more roommates from recommendation. However, this monotonic pattern does not apply to students from Gaokao: in the informatics department, the average GPA drops when the number of their peers from recommendation increases from $1$ to $2$ (i.e., the treatment moves from $r_2=112$ to $r_3=122$); in the physics department, the average GPA drops when the number of their peers from recommendation increases from $0$ to $1$ (i.e., the treatment moves from $r_1=111$ to $r_2=112$) and drops again when the number of their peers from recommendation increases from $2$ to $3$ (i.e., the treatment moves from $r_3=122$ to $r_4=222$). We observe some treatment effect heterogeneity in different subgroups, although many results in Table \ref{table:est_pe} are insignificant due to small sample sizes.

\subsection{Optimal roommate assignment}\label{sec::optimal}
We derive the optimal roommate assignment mechanism for the same population as those in the informatics or physics departments, separately. We estimate the optimal complete randomization and obtain the fiducial distribution of $l_{\text{opt}}$. Table \ref{tab:opt_assign} shows the estimators and fiducial distributions for the optimal roommate assignments.

\begin{table}[t]
\centering
\caption{Fiducial distributions of optimal roommate assignments.
Estimated optimal roommate assignments are in bold with $\mathcal{G}=\{g_1,\ldots, g_5\}=\{1111, 1112, 1122, 1222, 2222\}$.
The ``Prob." columns denote the fiducial probability of the treatment assignment based on $10^4$ draws, and the ``Outcome" columns denote the unbiased estimator for the expected total GPA under any treatment assignment mechanism. 
}\label{tab:opt_assign}
\begin{tabular}{c|c|ccccc||c|c|ccccc}
\toprule
\multicolumn{6}{c}{Informatics} & \multicolumn{6}{c}{Physics}\\
\midrule
Prob. & Outcome & ${l}_{1}$ & ${l}_{2}$ & ${l}_{3}$ & ${l}_{4}$ & ${l}_{5}$ & Prob. & Outcome & ${l}_{1}$ & ${l}_{2}$ & ${l}_{3}$ & ${l}_{4}$ & ${l}_{5}$\\[5pt]
0.488 & 505.60 &  {\bf 26}  &  {\bf 0}  &  {\bf 0}  &  {\bf 0}  & {\bf 13}  &
 0.555 & 306.29  & {\bf 12} &   {\bf 0} &   {\bf 0}  &  {\bf 1} &  {\bf 10}\\
0.209 & 504.96  &  2 &  32  &  0  &  0  &  5  & 0.160 & 301.61 &   2 &   0 &  20  &  1 &   0
\\
0.159 & 504.27  &  0  & 34  &  1 &   0 &   4  &0.132  & 302.36 &  9  &  0  &  0  & 13  &  1
\\
 0.091 & 504.03 & 0 &  34 &   0  &  2  &  3 & 0.059 & 300.49 &  1  &  1  & 21  &  0  &  0
\\
0.015 & 502.70  &  0  & 33 &   0  &  5 &   1 & 0.049 & 301.89  &  8 &   0 &   2  & 13 &   0
\\
0.009 & 496.64 &  0  & 26 &  13 &   0 &   0& 0.022 & 305.17 &  11 &   1   & 1 &   0 &  10
\\
0.009 & 488.71 &  13  &  0  & 26  &  0  &  0  & 0.010 & 300.82 &   8  &  1  &  0  & 14  &  0
\\
0.008 & 498.42 & 22  &  0  &  0  & 16  &  1  & 0.004  & 288.48 & 1 &  15   & 0  &  0   & 7
\\
0.007 & 497.95 & 21  &  1  &  0  & 17  &  0 & 0.004 & 302.99  & 10   & 3  &  0 &   0 &  10
\\
0.003 & 501.62 & 0  & 32  &  1  &  6 &   0 &  0.003 & 287.62 &  0 &  13   & 0  & 10 &   0
\\
0.002 & 501.84 &  1 &  31  &  0 &   7 &   0 & 0.002 & 298.31 &  0  &  3 &  20 &   0  &  0\\
\bottomrule
\end{tabular}
\end{table}

As Table \ref{tab:opt_assign} suggests, assigning students with the same type together can maximize the academic performance of all students. This may encourage separating students admitted through different channels. However, the optimal treatment assignment has huge uncertainty. It is worth taking a look at the optimal treatment assignments with the top five fiducial probabilities. Most of them do suggest mixing students with different attributes. The decision may be misleading based on a point estimate of the optimal treatment assignment. Moreover, in practice, the average GPA is just a single measure of the students' performance, and other criteria may come into play in practice. For example, if we consider both the average GPA and the diversity of students in each room, it is better to mix different types of students.

\subsection{Future research directions based on this data set}\label{sec::future}

There are several interesting future research directions based on this data set. First, we used the largest two departments of the university for randomization-based inference, because large sample approximations for other small departments are unlikely to be reliable. Second, we analyzed the data from different departments separately. It would be interesting to analyze the data set of the whole university simultaneously, allowing the smaller departments to borrow information from other larger departments. Third, the data set also contains other background information. It is our future research to leverage these covariates to improve estimation efficiency.

\section{Discussion: inference without Assumptions \ref{asp:sutva_pe} or \ref{asp:exres_pe}}\label{sec:discussion}

Assumptions \ref{asp:sutva_pe} and \ref{asp:exres_pe} may be too strong and may not hold in some applications. 
Below we discuss alternative inferential strategies without them. We summarize the main results below and relegate the technical details to the Supplementary Material.

\subsection{Randomization test}

Without Assumptions \ref{asp:sutva_pe} or \ref{asp:exres_pe}, we can still use the randomization tests under the sharp null hypothesis that the treatment $Z$ does not affect any units. This preserves the type one error in finite samples. However, rejecting the sharp null hypothesis may not be informative for understanding peer effects. It is worth extending previous randomization test strategies \citep{rosenbaum2007interference, luo2012inference, aronow2012general, bowers2013reasoning, rigdon2015exact, athey2015exact, basse2017exact} to our setting.

\subsection{Other estimands of interest}\label{subsec::generalD}

We can unbiasedly estimate some other estimands without Assumptions \ref{asp:sutva_pe} or \ref{asp:exres_pe}. For example, let $Y_i^{\mathcal{D}}(r) = \sum_{z\in \mathcal{Z}} \text{pr}(Z=z\mid R_i=r)Y_i(z)$ be a weighted average of unit $i$'s potential outcomes, where the superscript $\mathcal{D}$ denotes the design. For example, $\mathcal{D} = $ RP for random partitioning and $\mathcal{D} = $ CR for complete randomization. 
Because 
the weight is nonzero only for assignment $z$ such that $R_i(z_i)=r$, 
we can view $Y_i^{\mathcal{D}}(r)$ as a summary of the potential outcomes $Y_i(z)$'s when unit $i$ has $K$ peers with attributes $r$. 
Moreover, if Assumption \ref{asp:sutva_pe} holds, $Y_i^{\text{RP}}(r)=Y_i^{\text{CR}}(r)$ reduces to the average of the $Y_i(z_i)$'s for $z_i$ such that $R_i(z_i)=r$. Thus, we can view $\tau_i^{\mathcal{D}}(r,r') = Y_i^{\mathcal{D}}(r)-Y_i^{\mathcal{D}}(r')$ as an individual peer effect comparing treatments $r$ and $r'$. 
Define $\bar{Y}_{[a]}^{\mathcal{D}}(r)$ and $\tau^{\mathcal{D}}_{[a]}(r,r')$ as the averages of $Y_i^{\mathcal{D}}(r)$'s and $\tau_i^{\mathcal{D}}(r,r')$'s for units with attribute $a$, and $\tau^{\mathcal{D}}(r,r')$ as the average of $\tau_i^{\mathcal{D}}(r,r')$'s for all units. 
These estimands depend on the design as emphasized by the superscript $\mathcal{D}$. 
In contrast, the estimands $\tau_{[a]}(r,r')$ and $\tau(r,r')$ in previous sections do not depend on the design.
Under treatment assignment mechanisms satisfying Assumption \ref{asmp:indist}, we can show that
the estimators $\hat{Y}_{[a]}(r), \hat{\tau}_{[a]}(r,r')$ and $\hat{\tau}(r,r')$
in \eqref{eq:sub_ave_pot_hat} and \eqref{eq:tau_hat} are still unbiased for $\bar{Y}_{[a]}^{\mathcal{D}}(r), \tau^{\mathcal{D}}_{[a]}(r,r')$ and $\tau^{\mathcal{D}}(r,r')$, respectively. 
However, we do not have replications for any treatment levels to evaluate their uncertainty. 

In sum, Assumptions \ref{asp:sutva_pe} and \ref{asp:exres_pe} can be strong in practice. Without them, it is challenging to conduct repeated sampling inference although it is possible to obtain meaningful point estimates.  

\subsection{Distributional assumptions on potential outcomes}

An alternative approach imposes some distributional assumptions on the potential outcomes. 
In particular, instead of assuming that the potential outcomes depend only on the peer attribute set as in Assumption \ref{asp:exres_pe}, we allow for  some deviations but need an additional distributional assumption on the error terms. 

\begin{assumption}\label{asp:dist}
The potential outcome can be decomposed as
$Y_i(z) = Y_i(R_i(z_i)) + \varepsilon_i(z)$ with the error terms satisfying 
\begin{itemize}
	\item[(i)] $(\varepsilon_1(z), \ldots, \varepsilon_n(z))$ are mutually independent with zero mean, for any peer assignment $z \in \mathcal{Z}$;
	\item[(ii)] all $\{\varepsilon_i(z): A_i = a, R_i(z_i) = r\}$ have the same variance $\sigma^2_{[a]r}$, for any attribute $1\leq a\leq H$ and any peer attribute set $r\in \mathcal{R}$.
 \end{itemize}
\end{assumption}

Assumption \ref{asp:dist} is weaker than Assumption \ref{asp:exres_pe}, and reduces to Assumption \ref{asp:exres_pe} when the $\varepsilon_i(z)$'s are all zero, i.e., 
$\sigma^2_{[a]r}=0$ for all $a$ and $r$. 
Under Assumption \ref{asp:dist}, $\tau_{[a]}(r,r')$ in \eqref{eq:sub_ave_peer_effect} is still a meaningful estimand, although it depends only on the main terms instead of the potential outcomes. 
Moreover, it equals the expectation of $\tau_{[a]}^\mathcal{D}(r,r')$ defined in Section \ref{subsec::generalD} by averaging over the random error terms. 
Under complete randomization, we show in the Supplementary Material that
$\hat{\tau}_{[a]}(r,r')$ is still an unbiased estimator for $\tau_{[a]}(r,r')$, and the variance estimator 
$\hat{V}_{[a]}({r}, {r}')$ is still conservative in expectation for the sampling variance of  $\hat{\tau}_{[a]}(r,r')$.

\subsection{Peer effects for a target subpopulation}

\subsubsection{Target subpopulation, potential outcomes, and peer effects}

We formulate the approach of \citet{rubinpeersmoking} using the potential outcomes introduced in this paper. We consider the following ideal setting.  First, we select a ``target" subpopulation from units with attribute $a$. We assume that the units in the target subpopulation have identity numbers $(1,2,\ldots,\underline{m})$ with $\underline{m} \leq  \min(m, n_{[a]})$. Second, we assign all units except the target subpopulation into $m$ groups, with $\underline{m}\leq m$ groups containing $K$ units and the remaining $m-\underline{m}$ groups containing $K+1$ units. Third, we assign the $\underline{m}$ units in the target subpopulation into these  $\underline{m}$ groups with $K$ units. 
Therefore, the peer assignments for units in the target subpopulation are 
from randomly permuting these $\underline{m}$ groups.

Let $(\zeta_{1}, \ldots, \zeta_{\underline{m}})$ denote the units initially assigned to the $\underline{m}$ groups, where $\zeta_{k}$ is the set consisting of the identity numbers of the $K$ units in group $k$ ($1\leq k\leq \underline{m}$). 
Therefore, the peer assignments for units in the target subpopulation, $(z_1, \ldots, z_{\underline{m}})$, is a permutation of 
$(\zeta_{1}, \ldots, \zeta_{\underline{m}})$, and the peer assignments for units in the remaining $m-\underline{m}$ groups are fixed. 
Therefore, unit $i$'s potential outcome simplifies to $Y_i(z_1, \ldots, z_{\underline{m}}, z_{\underline{m}+1}, \ldots, z_{n}) = Y_i(z_1, \ldots, z_{\underline{m}})$ for $1\leq i\leq \underline{m}$. 

Following \citet{rubinpeersmoking}, we introduce the following two assumptions.

\begin{assumption}\label{asp:ideal_sutva}
If $z_i = z'_i$, then 
$Y_i(z_1, \ldots, z_{\underline{m}}) = Y_i(z_1', \ldots, z_{\underline{m}}')$, for any two peer assignments $(z_1, \ldots, z_{\underline{m}})$ and $(z_1', \ldots, z_{\underline{m}}')$ and for any unit $1\leq i\leq \underline{m}$ in the target subpopulation. 
\end{assumption}

Assumption \ref{asp:ideal_sutva} requires that each unit's potential outcomes depend only on its own peers. Under Assumption \ref{asp:ideal_sutva}, unit $i$'s potential outcome simplifies to $Y_i(z_1, \ldots, z_{\underline{m}}) = Y_i(z_i)$ for $1\leq i \leq \underline{m}$.

\begin{assumption}\label{asp:ideal_er}
If $R_i(z_i) = R_i(z'_i)$, then 
$Y_i(z_i) = Y_i(z_i')$, for any two peer assignments $(z_1, \ldots, z_{\underline{m}})$ and $(z_1', \ldots, z_{\underline{m}}')$ and for any unit $1\leq i\leq \underline{m}$ in the target subpopulation, 
\end{assumption}

Assumption \ref{asp:ideal_er} requires that each unit's potential outcomes depend only on the attributes of its peers. Under Assumption \ref{asp:ideal_er}, unit $i$'s potential outcome simplifies to $Y_i(z_i) = Y_i(R_i(z_i))$ for $1\leq i \leq \underline{m}$.

\citet{rubinpeersmoking} chose the attribute to be the smoking behavior and the target subpopulation to be nonsmoking freshman. They further dichotomized each of $( \zeta_{1}, \ldots, \zeta_{\underline{m}} ) $ into two categories: smoking or nonsmoking suites.

Under Assumptions \ref{asp:ideal_sutva} and \ref{asp:ideal_er}, unit $i$'s potential outcome simplifies to $Y_i(r)$ for some $r\in \mathcal{R}$, for $1\leq i \leq \underline{m}$ in the target subpopulation. 
Comparing two treatments $r,r'\in \mathcal{R}$, 
the individual peer effect for unit $i$  
is $\tau_i(r,r') = Y_i(r) - Y_i(r')$, 
and the average peer effect for units in the target subpopulation is 
$\tau_{\text{tg}}(r,r') = \underline{m}^{-1} \sum_{i=1}^{\underline{m}} \tau_i(r,r')$. 
We want to infer $\tau_{\text{tg}}$.

\subsubsection{Construction of target subpopulation and statistical inference}\label{sec:const_target}

We first construct the target subpopulation and then infer the peer effects for it. 
In the following, we assume complete randomization. 
For the observed peer assignment $Z$, let $\underline{m}$ be the number of groups with units of attribute $a$. 
For each of the $\underline{m}$ groups, we randomly pick one unit with attribute $a$ to constitute the target subpopulation, and denote the remaining units in these groups as $(\zeta_1, \ldots, \zeta_{\underline{m}})$. 
To construct the ideal setting,  
we conduct inference conditional on the group assignments for all units excluding the target subpopulation.
The remaining randomness comes solely from the peer assignments of the target subpopulation,  
which is a random permutation of $(\zeta_{1}, \ldots, \zeta_{\underline{m}})$. Moreover, under Assumptions \ref{asp:ideal_sutva} and \ref{asp:ideal_er}, the treatment assignment is a completely randomized experiment with multiple treatments taking values in $\mathcal{R} = \{r_1, \ldots, r_{|\mathcal{R}|} \}$. Therefore, for the average peer effect $\tau_{\text{tg}}(r,r')$, an unbiased estimator is the standard difference-in-means for units receiving treatments $r$ and $r'$. We relegate the sampling variance and variance estimator to the Supplementary Material.

\subsubsection{Comparison and connection to our approach}

Compared to Assumptions \ref{asp:sutva_pe} and \ref{asp:exres_pe}, Assumptions \ref{asp:ideal_sutva} and \ref{asp:ideal_er}
are weaker because they make assumptions for a subset of units and a subset of peer assignments, i.e., the unique values in $(\zeta_{1}, \ldots, \zeta_{\underline{m}})$. 
However, under this ideal setting with Assumptions \ref{asp:ideal_sutva} and \ref{asp:ideal_er}, we can only infer peer effects for the target subpopulation instead of all units.

The construction in Section \ref{sec:const_target} generates a random target subpopulation. Interestingly, averaging over all possible constructions, the point estimator for $\tau_{\text{tg}}(r,r')$ is the same as 
$\hat{\tau}_{[a]}(r,r')$ in \eqref{eq:est_pe_fix}. We prove this result in the Supplementary Material.

 \section*{Supplementary Material}
 
Appendix A1 gives supporting materials for Section \ref{sec:discussion}.
Appendix A2 gives more technical details for general treatment assignment mechanisms.
Appendix A3 gives more technical details for complete randomization. 
Appendix A4 gives more technical details for random partitioning.

 \section*{Acknowledgments}
 
The authors thank Don Rubin, Luke Miratrix, Ms. Kristen Hunter and Mr. Zach Branson  at Harvard University,
the Associate Editor and two reviewers for insightful comments. Dr. Avi Feller at UC Berkeley kindly edited our final version.

\section*{Funding}
The authors gratefully acknowledge financial support from the National Science Foundation (Peng Ding: DMS grant \# 1713152; 
Jun Liu: DMS \# 1712714)

\bibliographystyle{plainnat}
\bibliography{causal}

\newpage
\setcounter{page}{1}
\begin{center}
{\bf \huge 
Supplementary Material for \\
``Randomization Inference for Peer Effects''
}
\\
\bigskip 
{
by Xinran Li, Peng Ding, Qian Lin, Dawei Yang, and Jun Liu
}
\end{center}

\bigskip

\setcounter{equation}{0}
\setcounter{section}{0}
\setcounter{figure}{0}
\setcounter{example}{0}
\setcounter{proposition}{0}
\setcounter{corollary}{0}
\setcounter{theorem}{0}
\setcounter{table}{0}

\renewcommand {\theproposition} {A\arabic{proposition}}
\renewcommand {\theexample} {A\arabic{example}}
\renewcommand {\thefigure} {A\arabic{figure}}
\renewcommand {\thetable} {A\arabic{table}}
\renewcommand {\theequation} {A\arabic{equation}}
\renewcommand {\thelemma} {A\arabic{lemma}}
\renewcommand {\thesection} {A\arabic{section}}
\renewcommand {\thetheorem} {A\arabic{theorem}}
\renewcommand {\thecorollary} {A\arabic{corollary}}
\renewcommand {\theassumption} {A\arabic{assumption}}

\allowdisplaybreaks

Appendix \ref{app:fewer_asmp} gives supporting materials for Section \ref{sec:discussion}.

Appendix \ref{app::general} gives more technical details for general treatment assignment mechanisms.

Appendix \ref{sec::cr-appendix} gives more technical details for complete randomization. 

Appendix \ref{app:more_on_rand_part} gives more technical details for random partitioning.

\section{Analyzing peer effects without Assumptions \ref{asp:sutva_pe} or \ref{asp:exres_pe}}\label{app:fewer_asmp}

\subsection{Randomization tests}\label{sec:rand_test}
We use randomization tests for the significance of peer effects for the following reasons. 
First, they provide additional evidence for the significance of peer effects. 
Second, randomization tests are exact and valid for finite samples.  Third, randomization tests do not require Assumptions \ref{asp:sutva_pe}--\ref{asmp:indist}, as long as the assignment mechanism is known. 
Fourth, as \citet{Fisher:1935} suggested, we can use randomization tests to check the Normal approximations. 

We can use randomization tests for the sharp null hypothesis for all units:
\begin{align*}  
H_{0}: Y_i(z) = Y_i(z'), \quad \text{for all } z,z' \text{ and for all unit } i,
\end{align*}
or the null hypothesis for units with attribute $a$:
\begin{align*}  
H_{0,[a]}: Y_i(z) = Y_i(z'), \quad \text{for all } z,z' \text{ and for all unit } i \text{ such that } A_i=a.
\end{align*}
$H_{0,[a]}$ is not sharp. But we can still conduct randomization test for $H_{0,[a]}$. We choose test statistics depending only on the outcomes of units with attribute $a$, and their randomization distributions are known under $H_{0,[a]}$. We can then obtain exact $p$-values under $H_{0,[a]}$.

We first discuss the choices of test statistics for the subgroup null $H_{0,[a]}$, and then the test statistics for $H_0$.
A choice of test statistic for the subgroup null $H_{0,[a]}$ is
\begin{align}\label{eq:est_sub_max_tau}
T_{[a]} & =  \max_{r,r'} \hat{\tau}_{[a]}(r,r') = \max_{r}\hat{Y}_{[a]}(r) - \min_{r}\hat{Y}_{[a]}(r).
\end{align}
Another choice of test statistic, $F_{[a]}$, is the F statistic from the analysis of variance of the linear regression of $Y_i$ on $R_i$ among units with attribute $a$. 
For the sharp null $H_0$, we can use $T = \max_{r,r'}\hat{\tau}(r,r')$, 
the F statistic $F$ from the linear model \eqref{eq:linear_model}, 
$\max_{a}T_{[a]}$, and $\max_{a}F_{[a]}$.

\begin{table}[htb]
\centering
\caption{$p$ values of randomization tests for the subgroup null $H_{0,[a]}$
and the sharp null $H_0$, with test statistics for $H_0$ shown in the parentheses. 
}\label{table:rt_sub_sn}
\begin{tabular}{ccccccc}
\toprule
Department & Null hypothesis & \multicolumn{3}{c}{Test statistic for $H_{0,[a]}$ ($H_0$)} \\
&  & $T_{[a]}$ ($\max_a T_{[a]}$) & $F_{[a]}$ ($\max_a F_{[a]}$)  & ($T$) & ($F$)
\\[5pt]
\midrule
Informatics & $H_0$ & 0.2678 & 0.3640 &   0.1092 & 0.1468 \\
& $H_{0,[1]}$ & 0.2390 & 0.3384  & \\
&$H_{0,[2]}$ & 0.0615 & 0.2050 &  \\[5pt]
physics & $H_0$ & 0.4553 & 0.0719  & 0.3431 & {\bf 0.0366} \\
& $H_{0,[1]}$ & 0.3545 & {\bf 0.0360} &  \\
& $H_{0,[2]}$ & 0.6994 & 0.7232 & 
\\
\bottomrule
\end{tabular}
\end{table}
For the  
motivating application, 
Table \ref{table:rt_sub_sn} shows the $p$-values from randomization tests for $H_0$ and $H_{0,[a]}$ with different test statistics. At significance level 0.05, the peer effects are not significant for students in the informatics department; the subgroup average peer effects are significant for students from Gaokao in the physics department if we use the F statistic as the test statistic, and the subgroup average peer effects are not significant for students from recommendation. 
We ignore the multiple testing issue. 
Compared to Table \ref{table:est_pe}, randomization tests reject only $H_{0,[1]}$ for students from Gaokao in the physics department, which may be due to the lack of power of randomization tests \citep{ding2014paradox}.

Moreover, we check the randomization distributions of the $\hat{\tau}_{[a]}(r,r')$'s under the sharp null hypothesis.
Under the sharp null hypothesis that the potential outcomes are not affected by the treatment assignment, all potential outcomes are known and identical to the observed outcomes. Therefore, the distributions of the subgroup peer effect estimators are known under complete randomization. 
For students in the informatics and physics departments graduating in 2013, Figures \ref{fig:check_norm}(a) and \ref{fig:check_norm}(b) show, respectively, the histograms of the subgroup peer effect estimators under the sharp null hypothesis based on $10^5$ treatment assignments from complete randomization. From Figures \ref{fig:check_norm}(a) and \ref{fig:check_norm}(b), the Normal approximations work fairly well. 
For our application, we do not know all potential outcomes and thus can not directly check the Normal approximations over repeated sampling of the treatment assignments. However, we view Figures \ref{fig:check_norm}(a) and \ref{fig:check_norm}(b) as intuitive justifications for the Normal approximation of $\hat{\tau}_{[a]}(r,r')$ in the Neymanian inference.

\begin{figure}
	\centering
	\begin{subfigure}{0.65\textwidth}  
		\includegraphics[width=\textwidth]{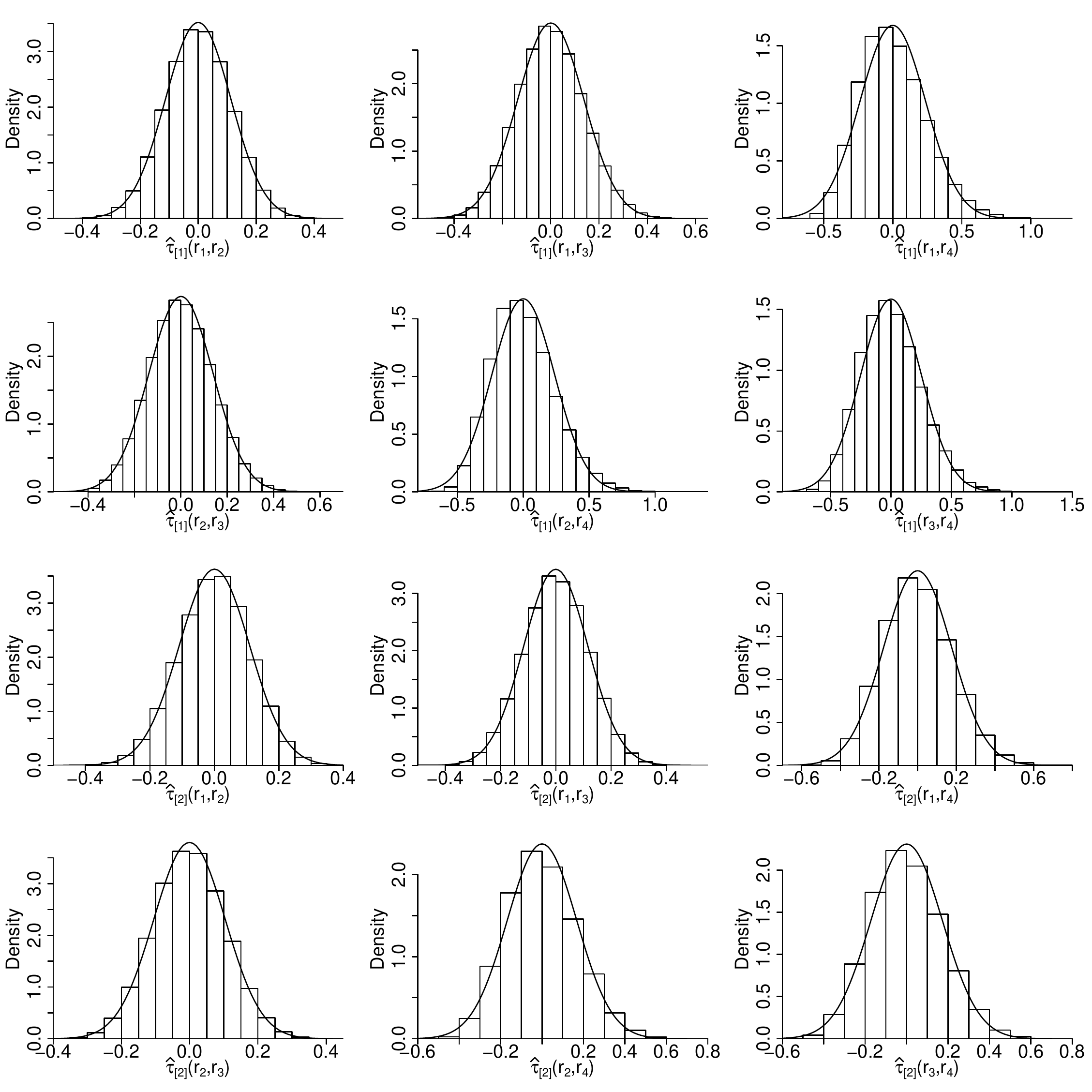}
		\caption{The informatics department}  \label{fig:check_norm_CS}
	\end{subfigure}
	\begin{subfigure}{0.65\textwidth}  
		\includegraphics[width=\textwidth]{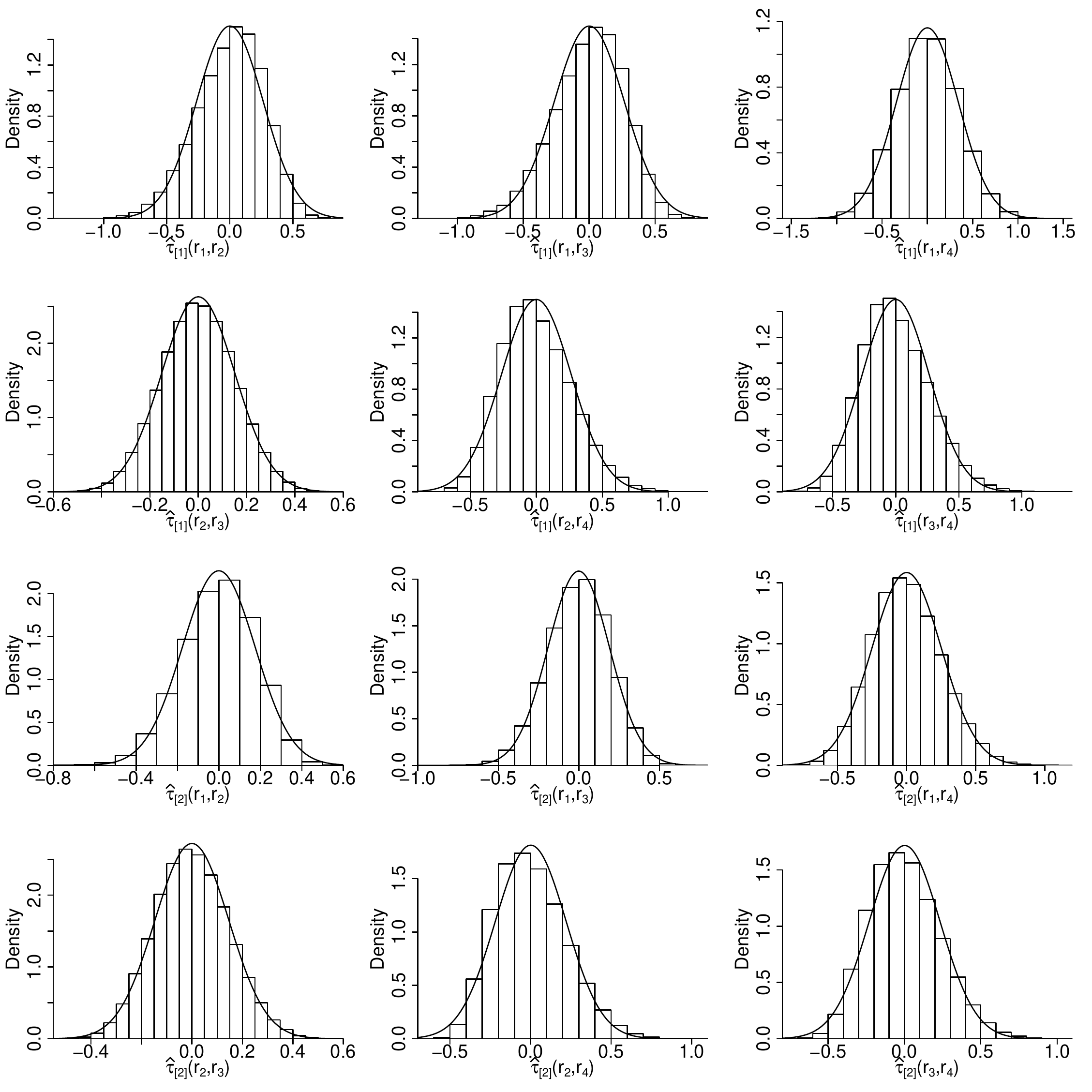}
		\caption{The physics department}  \label{fig:check_norm_PH}
	\end{subfigure}
	\caption{Histograms of subgroup peer effect estimators under the sharp null hypothesis, based on $10^5$ draws from complete randomization. The lines are densities of Normal approximations.}  \label{fig:check_norm}
\end{figure}

\subsection{Estimands, unbiased estimators, and optimal assignment mechanism}\label{sec:estimator_few_asmp}

Even if Assumptions \ref{asp:sutva_pe} or \ref{asp:exres_pe} fails, the estimators in \eqref{eq:sub_ave_pot_hat} and \eqref{eq:tau_hat} are still 
meaningful in the sense of unbiasedly estimating some average potential outcomes. Recall the definitions in Section \ref{sec:discussion}: 
\begin{eqnarray*}
& & Y_i^{\mathcal{D}}(r) = \sum_{z\in \mathcal{Z}} \text{pr}(Z=z\mid R_i=r)Y_i(z), \quad 
\tau_i^{\mathcal{D}}(r,r')= Y_i^{\mathcal{D}}(r)- Y_i^{\mathcal{D}}(r') ,\\
& & \bar{Y}_{[a]}^{\mathcal{D}}(r) = n_{[a]}^{-1} \sum_{i:A_i=a} Y_i^{\mathcal{D}}(r), \quad
\tau_{[a]}^{\mathcal{D}}(r,r') = n_{[a]}^{-1} \sum_{i:A_i=a}\tau_i^{\mathcal{D}}(r,r') ,\quad 
\tau^{\mathcal{D}}(r,r') = n^{-1} \sum_{i=1}^n\tau_i^{\mathcal{D}}(r,r') . 
\end{eqnarray*}

\begin{proposition}
\label{prop:avepotential}
Under Assumption \ref{asmp:indist}, $\hat{Y}_{[a]}(r)$ has mean $\bar{Y}_{[a]}^{\mathcal{D}}(r)$.
\end{proposition}

The conclusion follows from
\begin{align*}
E\left\{ \hat{Y}_{[a]}(r) \right\} & = \{n_{[a]}\pi_{[a]}(r)\}^{-1}\sum_{i:A_i=a} E \left\{ I(R_i=r)Y_i \right\} \\
&=
\{n_{[a]}\pi_{[a]}(r)\}^{-1}\sum_{i:A_i=a}
E\left\{
 I(R_i=r)\sum_{z}I(Z=z)Y_i(z)
\right\}\\
& = \{n_{[a]}\pi_{[a]}(r)\}^{-1}\sum_{i:A_i=a}\sum_{z}
E\left\{
 I(R_i=r)I(Z=z)
\right\}Y_i(z)\\
& = \{n_{[a]}\pi_{[a]}(r)\}^{-1}\sum_{i:A_i=a}\sum_{z}
\pr(
R_i=r, Z=z
)
Y_i(z)\\
& = \{n_{[a]}\pi_{[a]}(r)\}^{-1}\sum_{i:A_i=a}\sum_{z}
\pr(R_i=r)
\pr(
Z=z \mid R_i=r
)
Y_i(z)\\
& = n_{[a]}^{-1}\sum_{i:A_i=a}\sum_{z}
\pr(
Z=z \mid R_i=r
)
Y_i(z) = n_{[a]}^{-1}\sum_{i:A_i=a} Y_i^{\mathcal{D}}(r)
 = \bar{Y}_{[a]}^{\mathcal{D}}(r). 
\end{align*}

By the linearity of the expectation, 
$\hat{\tau}_{[a]} (r,r')$ and $\hat{\tau}(r,r')$ are unbiased for 
$\tau_{[a]}^{\mathcal{D}}(r,r')$ and $\tau^{\mathcal{D}}(r,r')$, respectively. 
However, without Assumptions \ref{asp:sutva_pe} and \ref{asp:exres_pe}, for a given treatment we do not have replications of units, making it difficult to evaluate the uncertainty of these estimators. 
Similarly, for the first type of interference, \citet{hudgens2008toward} discussed the expectations of the point estimators under general settings, but invoked ``stratified interference'' (analogous to Assumption \ref{asp:exres_pe}) to evaluate the uncertainty.

Moreover, the estimands $\tau_{[a]}^{\mathcal{D}}(r,r')$ and $\tau^{\mathcal{D}}(r,r')$ are meaningful in many situations. 
In the expression of $Y_i^{\mathcal{D}}(r)$, 
the weight $\text{pr}(Z=z\mid R_i=r)$ is nonzero only if $R_i(z_i)=r$, i.e., the attributes of unit $i$'s peers constitute $r$. Therefore, $Y_i^{\mathcal{D}}(r)$ summarizes unit $i$'s potential outcomes when he/she has $K$ peers with attributes $r$. 
Consequently, $\tau_i^{\mathcal{D}}(r,r')$ measures the difference when unit $i$ has $K$ peers with attributes $r$ rather than $r'$.
Thus we can view $\tau_i^{\mathcal{D}}(r,r')$ as the individual peer effect comparing treatments $r$ and $r'$, and $\tau_{[a]}^{ \mathcal{D}}(r,r')$ and $\tau^{\mathcal{D}}(r,r')$ as the corresponding average peer effects. 
Below we further simplify $Y_i^{\mathcal{D}}(r)$ under some special cases, making its meaning more intuitive. 
When Assumption \ref{asp:sutva_pe} holds, 
$Y_i^{\mathcal{D}}(r)$ reduces to 
$\sum_{z_i} \text{pr}(Z_i=z_i\mid R_i=r)Y_i(z_i)$; 
if further the treatment assignment mechanism is random partitioning or complete randomization, then the weight $\text{pr}(Z_i=z_i\mid R_i=r)$ is a nonzero constant for $z_i$ such that $R_i(z_i)=r$, and $Y_i^{\text{RP}}(r)=Y_i^{\text{CR}}(r)$ further reduces to the average of $Y_i(z_i)$'s for $z_i$ such that $R_i(z_i)=r$. 
When both  Assumptions \ref{asp:sutva_pe} and \ref{asp:exres_pe} hold, $Y_i^{\mathcal{D}}(r)$ is the same as $Y_i(r)$ in the main paper, which does not depend on the assignment mechanism $\mathcal{D}$. In sum, $Y_i^{\mathcal{D}}(r)$ is an extension of $Y_i(r)$.

We now consider the optimal complete randomization mechanism for the assignment of a new population of size $n'$ without Assumptions \ref{asp:sutva_pe} or \ref{asp:exres_pe}. 
Define similarly $Y_i^{\mathcal{D}'}(r)$ and $\bar{Y}_{[a]}^{\mathcal{D}'}(r)$ for the new population. 
By the same logic as \eqref{eq:ex_total_outcome_new}, 
the expected total outcome under complete randomization with ${L'}(z)=l'$ 
for the new population 
is
\begin{align}\label{eq:ex_total_outcome_new_no_asump}
E\left(\sum_{i=1}^{{n'}} Y'_i\right) 
& = 
\sum_{a=1}^H E\left(\sum_{i:A_i=a} Y'_i\right) 
=
\sum_{a=1}^H \sum_{r\in \mathcal{R}} n'_{[a]r} \bar{Y}^{\text{CR}'}_{[a]}(r).  
\end{align}

To estimate the maximizer $l'_{\text{opt}}$ of \eqref{eq:ex_total_outcome_new_no_asump}, we need to assume the new population is similar to the one in our data. 
Let $\mathcal{D}_0$ denote the assignment mechanism for our observed data.
The following assumption is similar to the one in the main paper.

\begin{assumption}\label{asp:new_pop_few_asmp}
$\bar{Y}_{[a]}^{\text{CR}'}(r) = \yscale \bar{Y}_{[a]}^{\mathcal{D}_0}(r)+\yshift$
for some constants $\yscale>0$ and $\yshift$, for all $1\leq a\leq H$ and $r\in \mathcal{R}$. 
\end{assumption}
Under Assumption \ref{asp:new_pop_few_asmp}, \eqref{eq:ex_total_outcome_new_no_asump} reduces to
\begin{align*}
E\left(\sum_{i=1}^{{n'}} Y'_i\right) 
=
\sum_{a=1}^H \sum_{r\in \mathcal{R}} n'_{[a]r} \bar{Y}^{\text{CR}'}_{[a]}(r)
= 
 \yscale \sum_{a=1}^H \sum_{r\in \mathcal{R}} n'_{[a]r} \bar{Y}_{[a]}^{\mathcal{D}_0}(r) + \sum_{a=1}^{H} n_{[a]}' \yshift,  
\end{align*}
implying that 
the maximizer $l'_{\text{opt}}$ of \eqref{eq:ex_total_outcome_new_no_asump}
is the same as that of $\sum_{a=1}^H \sum_{r\in \mathcal{R}} n'_{[a]r} \bar{Y}_{[a]}^{\mathcal{D}_0}(r)$. 
We can unbiasedly estimate $\bar{Y}_{[a]}^{\mathcal{D}_0}(r)$ by $\hat{Y}_{[a]}(r)$, and then estimate $l'_{\text{opt}}$ by simply plugging in the estimators $\hat{Y}_{[a]}(r)$'s. 
Again, it is difficult to evaluate the uncertainty of the estimator for $l'_{\text{opt}}$ 
for 
{\lxr ?, by?}
the same reason as that for the average peer effect estimators.

Below we give some comments on Assumption \ref{asp:new_pop_few_asmp}, which is key for inferring the optimal complete randomization. 
Assumption \ref{asp:new_pop_few_asmp} is a strong requirement of the similarity between the new population and the one in our data, due to the dependence of $\bar{Y}_{[a]}^{\text{CR}'}(r)$ and $\bar{Y}_{[a]}^{\mathcal{D}_0}(r)$ on the designs. Even if we assume that the new population is the same as the one in our data, Assumption \ref{asp:new_pop_few_asmp}, or, equivalently, $\bar{Y}_{[a]}^{\text{CR}}(r) = \yscale \bar{Y}_{[a]}^{\mathcal{D}_0}(r)+\yshift$, may fail because the values of $\bar{Y}_{[a]}^{\text{CR}}(r)$ and $\bar{Y}_{[a]}^{\mathcal{D}_0}(r)$ depend on the assignment mechanisms, CR and $\mathcal{D}_0$, respectively. If the new population is different from the one in our data, then Assumption \ref{asp:new_pop_few_asmp} is even less plausible.

We summarize several concerns for  inferring the optimal complete randomization in the absence of Assumptions \ref{asp:sutva_pe} or \ref{asp:exres_pe}. First, it is unnatural to infer the optimal peer assignment for a new population, because the treatment is the set of the identity numbers of units varying across populations. Second, the similarity assumption becomes stronger due to the dependence of the estimands on the design. Third, it is difficult to evaluate the uncertainty of the point estimator.

\subsection{Distributional assumptions on potential outcomes}

Under Assumption \ref{asp:dist}, we decompose $\hat{Y}_{[a]}(r)$ into two parts: 
\begin{align*}
\hat{Y}_{[a]}(r) & = n_{[a]r}^{-1} \sum_{i:A_i=a, R_i = r} Y_i = 
n_{[a]r}^{-1} \sum_{i:A_i=a, R_i = r} \left\{Y_i(r)+\varepsilon_i(Z)\right\}\\
& = 
n_{[a]r}^{-1} \sum_{i:A_i=a, R_i = r} Y_i(r)+
n_{[a]r}^{-1} \sum_{i:A_i=a, R_i = r} \varepsilon_i(Z) 
\equiv \check{Y}_{[a]}(r) + \check{\varepsilon}_{[a]}(r),
\end{align*}
where $\check{Y}_{[a]}(r)\equiv n_{[a]r}^{-1} \sum_{i:A_i=a, R_i = r} Y_i(r)$ and 
$\check{\varepsilon}_{[a]}(r)\equiv n_{[a]r}^{-1} \sum_{i:A_i=a, R_i = r} \varepsilon_i(Z)$ are the main part and deviance part, respectively. 
Let 
$\check{\tau}_{[a]}(r,r') = \check{Y}_{[a]}(r) - \check{Y}_{[a]}(r')$ and $\check{\delta}_{[a]}(r,r') = \check{\varepsilon}_{[a]}(r)-\check{\varepsilon}_{[a]}(r')$. 
Correspondingly, we can decompose the subgroup peer effect estimator $\hat{\tau}_{[a]}(r,r')$ into two parts: 
$\hat{\tau}_{[a]}(r,r') = \check{\tau}_{[a]}(r,r') + \check{\delta}_{[a]}(r,r')$. 

First, we discuss the sampling mean and variance of $\hat{\tau}_{[a]}(r,r')$. 
Under Assumption \ref{asp:dist}, we have, for any $1\leq a\leq H$ and $r\neq r'\in \mathcal{R}$, 
\begin{align}\label{eq:var_epsilon_check}
E\left\{
\check{\varepsilon}_{[a]}(r) \mid Z
\right\}
& = 
E\left\{
n_{[a]r}^{-1} \sum_{i:A_i=a, R_i = r} \varepsilon_i(Z) 
\mid Z
\right\} = 
n_{[a]r}^{-1} \sum_{i:A_i=a, R_i = r}
E\left\{
 \varepsilon_i(Z) 
\mid Z
\right\} = 0,
\nonumber
\\
\Var\left\{
\check{\varepsilon}_{[a]}(r) \mid Z
\right\}
& = 
\Var\left\{
n_{[a]r}^{-1} \sum_{i:A_i=a, R_i = r} \varepsilon_i(Z) 
\mid Z
\right\} 
= n_{[a]r}^{-2} 
\sum_{i:A_i=a, R_i = r} \Var\left\{
\varepsilon_i(Z)  \mid Z
\right\} = 
n_{[a]r}^{-1} \sigma^2_{[a]r},
\nonumber\\
\Cov\left\{\check{\varepsilon}_{[a]}(r), \check{\varepsilon}_{[a]}(r')\mid Z\right\} & = 
(n_{[a]r}n_{[a]r'})^{-1}\sum_{i:A_i=a, R_i = r}\sum_{j:A_j=a, R_j = r'} E\left\{
\varepsilon_i(Z) \varepsilon_j(Z) \mid Z 
\right\} = 0,
\end{align}
which immediately imply that 
$
E\{
\check{\delta}_{[a]}(r,r') \mid Z
\}=0$
and 
$
\Var\{
\check{\delta}_{[a]}(r,r') \mid Z
\} = n_{[a]r}^{-1} \sigma^2_{[a]r} + n_{[a]r'}^{-1} \sigma^2_{[a]r'}.
$
Thus, marginally, $\check{\delta}_{[a]}(r,r')$ has mean 0 and variance $n_{[a]r}^{-1} \sigma^2_{[a]r} + n_{[a]r'}^{-1} \sigma^2_{[a]r'}$.
Because $\check{\tau}_{[a]}(r,r')$ is constant given $Z$, we have
\begin{align*}
\Cov\left\{
\check{\tau}_{[a]}(r,r'), \check{\delta}_{[a]}(r,r')
\right\} & = 
E\left[ E\left\{
\check{\tau}_{[a]}(r,r')\check{\delta}_{[a]}(r,r') \mid Z
\right\} 
\right] = E\left[ \check{\tau}_{[a]}(r,r') E\left\{
\check{\delta}_{[a]}(r,r') \mid Z
\right\} 
\right] = 0,
\end{align*}
which implies that 
$\Var\{\hat{\tau}_{[a]}(r,r')\} = \Var\{\check{\tau}_{[a]}(r,r')\} + \Var\{\check{\delta}_{[a]}(r,r')\}$. 
Because the sampling variance of $\check{\tau}_{[a]}(r,r')$ is the same as that in Corollary \ref{cor:var_com_res} with Assumption \ref{asp:exres_pe}, we can derive that 
\begin{align*}
\Var\left\{\hat{\tau}_{[a]}(r,r')\right\} & =
\Var\left\{\check{\tau}_{[a]}(r,r')\right\} + \Var\left\{\check{\delta}_{[a]}(r,r')\right\}
= \frac{S_{[a]}^2({r})}{n_{[a]r}}
+ 
\frac{S_{[a]}^2({r}')}{n_{[a]r'}}
-\frac{S_{[a]}^2({r}\text{-}{r}')}{n_{[a]}} +
\frac{\sigma_{[a]}^2({r})}{n_{[a]r}}
+ 
\frac{\sigma_{[a]}^2({r}')}{n_{[a]r'}}\\
& = \frac{S_{[a]}^2({r})+\sigma^2_{[a]r}}{n_{[a]r}}
+ 
\frac{S_{[a]}^2({r}')+\sigma^2_{[a]r'}}{n_{[a]r'}}
-\frac{S_{[a]}^2({r}\text{-}{r}')}{n_{[a]}}.
\end{align*}

Second, we discuss the variance estimator for $\hat{\tau}_{[a]}(r,r')$. 
We decompose the sample variance of observed outcomes for units with attribute $a$ receiving treatment $r$ as
\begin{align}\label{eq:sample_var_weak_dist}
s^2_{[a]}(r) & = 
(n_{[a]r}-1)^{-1}
\sum_{i:A_i=a, R_i=r}\left\{
Y_i - \hat{Y}_{[a]}(r)
\right\}^2 = 
(n_{[a]r}-1)^{-1}\sum_{i:A_i=a, R_i=r}\left\{
Y_i (r) +\varepsilon_i(Z)- \check{Y}_{[a]}(r) - \check{\varepsilon}_{[a]}(r)
\right\}^2
\nonumber
\\
& = (n_{[a]r}-1)^{-1}\sum_{i:A_i=a, R_i=r}\left\{
Y_i (r) - \check{Y}_{[a]}(r) 
\right\}^2 + 
(n_{[a]r}-1)^{-1} \sum_{i:A_i=a, R_i=r}\left\{
\varepsilon_i(Z) - \check{\varepsilon}_{[a]}(r)
\right\}^2
\nonumber
\\
& \quad \ 
+ (n_{[a]r}-1)^{-1}\sum_{i:A_i=a, R_i=r}
\left\{
Y_i (r) - \check{Y}_{[a]}(r) 
\right\}\left\{
\varepsilon_i(Z) - \check{\varepsilon}_{[a]}(r)
\right\}.
\end{align}
Below we discuss the expectation of the three terms in \eqref{eq:sample_var_weak_dist} separately. 
The expectation of the first term is the same as that in Theorem \ref{thm:est_var_gen} with Assumption \ref{asp:exres_pe}. 
The expectation of the second term is $\sigma^2_{[a]r}$, because, conditional on $Z$, 
it is the sample variance of independent zero-mean random variables with variance $\sigma^2_{[a]r}$.
The expectation of the third term is zero,
because, conditioning on $Z$, $Y_i (r) - \check{Y}_{[a]}(r) $ is a constant and $\varepsilon_i(Z) - \check{\varepsilon}_{[a]}(r)$ has mean zero.
Above all,  
$
E\{
s^2_{[a]}(r)
\}
= S^2_{[a]}(r) + \sigma^2_{[a]r}.
$
The variance estimator is conservative because 
\begin{align*}
E\left\{ \hat{V}_{[a]}({r}, {r}') \right\}
& = 
\frac{E\{s_{[a]}^2({r}) \} }{n_{[a]r}}
+ 
\frac{ E\{ s_{[a]}^2({r}') \}}{n_{[a]r'}} = 
\frac{ S^2_{[a]}(r) + \sigma^2_{[a]r} }{n_{[a]r}}
+ 
\frac{ S^2_{[a]}(r') + \sigma^2_{[a]r'} }{n_{[a]r'}} \geq \Var\left\{\hat{\tau}_{[a]}(r,r')\right\} . 
\end{align*}

\subsection{Peer effects for a target subpopulation}
First, we study the sampling variance and its estimator for the difference-in-means estimator
$\hat{\tau}_{\text{tg}}(r,r')$. 
For any $r,r'\in \mathcal{R}$ and units in the target subpopulation, 
let 
$\bar{Y}_{\text{tg}}(r) = \underline{m}^{-1}\sum_{i=1}^{\underline{m}}Y_i(r)$ and
$\mathfrak{S}^2({r}) = (\underline{m}-1)^{-1}\sum_{i=1}^{\underline{m}}\{Y_i(r) - \bar{Y}_{\text{tg}}(r)\}^2$ 
be the finite population average and variance of individual potential outcome $Y_i(r)$'s, 
and 
$\mathfrak{S}^2({r}\text{-}{r}')$ be the finite population variance of individual peer effect $\tau_{i}(r,r')$'s. 
The number of units in the target subpopulation receiving treatment $r\in \mathcal{R}$ is $\underline{m}_{r} = \sum_{k=1}^{\underline{m}} I(\{A_j: j \in \zeta_k\}
=r)
$. 
Following \citet{Neyman:1923}, the sampling variance of $\hat{\tau}_{\text{tg}}(r,r')$ is 
\begin{align*}
\Var\{
\hat{\tau}_{\text{tg}}(r,r')
\}
& = 
\frac{\mathfrak{S}^2({r})}{\underline{m}_{r}}
+ 
\frac{\mathfrak{S}^2({r}')}{\underline{m}_{r'}}
-\frac{\mathfrak{S}^2({r}\text{-}{r}')}{\underline{m}}.
\end{align*}
For any $r\in \mathcal{R}$, let $\mathfrak{s}^2(r)$ be the sample variance of observed outcome $Y_i$'s for units receiving treatment $r$. We can show that $\mathfrak{s}^2(r)$ is an unbiased estimator for $\mathfrak{S}^2(r)$. Therefore, a conservative sampling variance estimator for $\hat{\tau}_{\text{tg}}(r,r')$ is 
\begin{align*}
\widehat{\Var}\{
\hat{\tau}_{\text{tg}}(r,r')
\}
& = 
\frac{\mathfrak{s}^2({r})}{\underline{m}_{r}}
+ 
\frac{\mathfrak{s}^2({r}')}{\underline{m}_{r'}}.
\end{align*}
Moreover, under some regularity conditions, the Wald-type confidence interval is asymptotically conservative.

Second, we show that, averaging over all possible constructions, the point estimator for $\tau_{\text{tg}}(r,r')$ is the same as 
$\hat{\tau}_{[a]}(r,r')$ in \eqref{eq:est_pe_fix}. 
Because the point estimator for $\tau_{\text{tg}}(r,r')$ is the standard difference-in-means for units receiving treatments $r$ and $r'$, 
it suffices to show that the average of 
$\underline{m}_{r}^{-1}\sum_{i=1}^{\underline{m}} I(R_i = r) Y_i$ over all configurations of the target subpopulation is the same as $n_{[a]r}^{-1}\sum_{i:A_i=a,R_i=r}Y_i$, the average observed outcome for units with attribute $a$ receiving treatment $r$. This is true because (i) a unit with attribute $a$ receiving treatment $r$ must be in a group with group attribute $\{a\}\cup r$,  
(ii) any group with group attribute $\{a\}\cup r$ has the same number of units with attribute $a$, and (iii) each configuration randomly picks one unit with attribute $a$ in these groups with group attribute $\{a\}\cup r$ and calculates their average observed outcome to get $\underline{m}_{r}^{-1}\sum_{i=1}^{\underline{m}} I(R_i = r) Y_i$.

\section{Technical details for general treatment assignments}\label{app::general}

\subsection{Lemmas}
Recall that $S_{[a]}^2(r)$ and $S_{[a]}^2(r\text{-}r')$ are the finite population variances of potential outcomes $Y_i(r)$'s and
individual peer effects $\tau_i(r,r')$'s among units with attribute $a$. 
We further define
$S_{[a]}(r,r')$ as the finite population covariance between the $Y_i(r)$'s and the $Y_i(r')$'s among units with attribute $a$, and
$
\widetilde{Y}_i(r) = Y_i(r) - \bar{Y}_{[A_i]}(r) 
$
as the centered potential outcome of unit $i$ by subtracting the average potential outcome among units with the same attribute as unit $i$.
We can then rewrite
\begin{align}\label{eq:sh2r_sh2rr_prim}
S_{[a]}^2(r) = (n_{[a]}-1)^{-1}\sum_{i:A_i=a}\widetilde{Y}_i^2(r), \quad 
S_{[a]}^2(r,r') = (n_{[a]}-1)^{-1}\sum_{i:A_i=a}\widetilde{Y}_i(r)\widetilde{Y}_i(r'),
\end{align} 
and decompose the subgroup average potential outcome estimator as
\begin{align*}
\hat{Y}_{[a]}(r) & = \{n_{[a]} \pi_{[a]}(r)\}^{-1}\sum_{i:A_i=a}I(R_i=r)\widetilde{Y}_i(r) + \{n_{[a]} \pi_{[a]}(r)\}^{-1}\sum_{i:A_i=a}I(R_i=r)\bar{Y}_{[a]}(r)
\equiv B_{[a]}(r) + C_{[a]}(r),
\end{align*}
and decompose the subgroup average peer effect estimator as 
\begin{eqnarray}\label{eq:rewrite_tau_h}
\hat{\tau}_{[a]}(r,r') 
 =  
\hat{Y}_{[a]}(r) - \hat{Y}_{[a]}(r')
 \equiv  B_{[a]}(r)+C_{[a]}(r)-B_{[a]}(r')-C_{[a]}(r').
\end{eqnarray}
The following three lemmas characterize the  covariances of the terms in \eqref{eq:rewrite_tau_h}.  
\begin{lemma}\label{lemma:var_between_centered_outcome}
For $1\leq a,a'\leq H$ and $r,r'\in \mathcal{R}$,
\begin{eqnarray*}
\Cov\{B_{[a]}(r), B_{[a']}(r')\}
& = &
\begin{cases}
0, & \text{if } a\neq a';\\
-(n_{[a]}-1)\frac{\pi_{[a][a]}(r,r')}{n_{[a]}^2  \pi_{[a]}(r) \pi_{[a]}(r')} S_{[a]}(r,r'), & \text{if } a=a', r\neq r';\\
(n_{[a]}-1)\frac{\pi_{[a]}(r)-\pi_{[a][a]}(r,r)}{n_{[a]}^2  \pi_{[a]}^2(r) }S_{[a]}^2(r), & \text{if } a=a', r = r'.
\end{cases}
\end{eqnarray*}
\end{lemma}

\begin{lemma}\label{lemma:cov_between_center_and_mean}
For $1\leq a,a'\leq H$ and $r,r'\in \mathcal{R}$,
$
\Cov\{B_{[a]}(r), C_{[a']}(r')\} 
 =0.
$
\end{lemma}

\begin{lemma}\label{lemma:var_between_mean}
For any $1\leq a,a'\leq H$ and $r,r'\in \mathcal{R}$, 
\begin{align*}
\Cov\{C_{[a]}(r), C_{[a']}(r')\} 
& =
(n_{[a]} n_{[a']})^{-1/2}
c_{[a][a']}(r,r')\bar{Y}_{[a]}(r)\bar{Y}_{[a']}(r'),
\end{align*}
where $c_{[a][a']}(r,r')$ 
is defined in \eqref{eq:c_def} in the main text. 
\end{lemma}

Recall that   
$
\Yhhrrprime =\{n_{[a]}(n_{[a]}-1)\}^{-1}\sum_{i\neq j:A_i=A_j=a}Y_i(r)Y_j(r')
$
is 
the average of the  
products of the potential outcomes 
for pairs of two different units with the same attribute $a$. 
The following lemma represents the finite population covariance $S_{[a]}^2(r,r')$ and the product of average potential outcomes $\bar{Y}_{[a]}(r)\bar{Y}_{[a]}(r')$ as functions of $S_{[a]}^2(r)$, $S_{[a]}^2(r')$, $S_{[a]}^2(r\text{-}r')$ and $\Yhhrrprime$.  
\begin{lemma}\label{lemma:represent_S_product}
For $1\leq a\leq H$, and $r\neq r'\in \mathcal{R}$, 
\begin{itemize}
\item[(a)] $2S_{[a]}(r,r') = S_{[a]}^2(r)+S_{[a]}^2(r')-S_{[a]}^2(r\text{-}r')$;
\item[(b)] $\bar{Y}_{[a]}(r)\bar{Y}_{[a]}(r') = \Yhhrrprime+
(2n_{[a]})^{-1}
\{
S_{[a]}^2(r)+S_{[a]}^2(r')-S_{[a]}^2(r\text{-}r')
\}.$
\end{itemize}
\end{lemma}

\subsection{Proofs of the lemmas}
\begin{proof}[Proof of Lemma \ref{lemma:var_between_centered_outcome}]
Based on the definitions of $\pi_{[a]}(r)$ and $\pi_{[a][a']}(r,r')$, for units $i$ and $j$ such that $A_i=a$ and $A_j=a'$, the covariance between their treatment indicators is
\begin{align}\label{eq:cov_indicator}
\Cov\{ I(R_i=r), I(R_j=r') \} & = \pr(R_i=r, R_j=r')-\pr(R_i=r) \pr(R_j=r') \nonumber\\
& = 
\begin{cases}
\pi_{[a][a']}(r,r')-\pi_{[a]}(r)\pi_{[a']}(r'), & \text{if } i\neq j,\\
-\pi_{[a]}(r)\pi_{[a]}(r'),  & \text{if } i= j, r\neq r',\\
\pi_{[a]}(r) - \pi_{[a]}^2(r) & \text{if } i = j, r= r'.
\end{cases}
\end{align}
Below we discuss three cases separately. 
(1)
When $a \neq a'$, any units $i$ and $j$ such that $A_i=a$ and $A_j=a'$ must satisfy $i\neq j$. Therefore, 
\begin{align*}
\Cov\{B_{[a]}(r), B_{[a']}(r')\}
& = \Cov\left[
\{n_{[a]} \pi_{[a]}(r)\}^{-1} \sum_{i:A_i=a}I(R_i=r) \widetilde{Y}_i(r), \{n_{[a']} \pi_{[a']}(r')\}^{-1}\sum_{j:A_j=a'}I(R_j=r')\widetilde{Y}_j(r')
\right]\\
& = \{n_{[a]} n_{[a']}  \pi_{[a]}(r) \pi_{[a']}(r')\}^{-1}\sum_{i:A_i=a}
\sum_{j:A_j=a'}\widetilde{Y}_i(r)\widetilde{Y}_j(r')\Cov\{
I(R_i=r), I(R_j=r')
\}\\
& = \frac{\pi_{[a][a']}(r,r')-\pi_{[a]}(r)\pi_{[a']}(r')}{n_{[a]} n_{[a']}  \pi_{[a]}(r) \pi_{[a']}(r')}\sum_{i:A_i=a}
\sum_{j:A_j=a'}
\widetilde{Y}_i(r)\widetilde{Y}_j(r'),
\end{align*}
where the last equality follows from \eqref{eq:cov_indicator}.
We can further simplify $\Cov\{B_{[a]}(r), B_{[a']}(r')\}$ as 
\begin{eqnarray*}
\Cov\{B_{[a]}(r), B_{[a']}(r')\} = 
\frac{\pi_{[a][a']}(r,r')-\pi_{[a]}(r)\pi_{[a']}(r')}{n_{[a]} n_{[a']}  \pi_{[a]}(r) \pi_{[a']}(r')}\sum_{i:A_i=a}\widetilde{Y}_i(r)\sum_{j:A_j=a'}\widetilde{Y}_j(r') = 0.
\end{eqnarray*}

(2)  
When $a=a'$ and $r\neq r'$, we need to consider the covariances between the treatment indicators of two different units and those of the same unit: 
\begin{align*}
\Cov\{B_{[a]}(r), B_{[a]}(r')\}
& = \Cov\left[
\{n_{[a]} \pi_{[a]}(r)\}^{-1}\sum_{i:A_i=a}I(R_i=r)\widetilde{Y}_i(r), \{n_{[a]} \pi_{[a]}(r')\}^{-1} \sum_{j:A_j=a}I(R_j=r')\widetilde{Y}_j(r')
\right]\\
& =  \{n_{[a]}^2  \pi_{[a]}(r) \pi_{[a]}(r')\}^{-1} \sum_{i\neq j:A_i=a,A_j=a}\widetilde{Y}_i(r)\widetilde{Y}_j(r')\Cov\{
I(R_i=r), I(R_j=r')
\}\\
& \quad \  +\{n_{[a]}^2  \pi_{[a]}(r) \pi_{[a]}(r')\}^{-1} \sum_{i:A_i=a}\widetilde{Y}_i(r)\widetilde{Y}_i(r')\Cov\{
I(R_i=r), I(R_i=r')
\}\\
& =  \frac{\pi_{[a][a]}(r,r')-\pi_{[a]}(r)\pi_{[a]}(r')}{n_{[a]}^2  \pi_{[a]}(r) \pi_{[a]}(r')}\sum_{i\neq j:A_i=a,A_j=a}\widetilde{Y}_i(r)\widetilde{Y}_j(r') -\frac{\pi_{[a]}(r)\pi_{[a]}(r')}{n_{[a]}^2  \pi_{[a]}(r) \pi_{[a]}(r')}\sum_{i:A_i=a}\widetilde{Y}_i(r)\widetilde{Y}_i(r'),
\end{align*}
where the last equality follows from \eqref{eq:cov_indicator}. 
We can further simplify $\Cov\{B_{[a]}(r), B_{[a]}(r')\}$ as 
\begin{align*}
\Cov\{B_{[a]}(r), B_{[a]}(r')\}
& =   \frac{\pi_{[a][a]}(r,r')-\pi_{[a]}(r)\pi_{[a]}(r')}{n_{[a]}^2  \pi_{[a]}(r) \pi_{[a]}(r')}\sum_{i:A_i=a}
\sum_{j:A_j=a}
\widetilde{Y}_i(r)\widetilde{Y}_j(r')\\
& \quad \  - 
\left(
\frac{\pi_{[a][a]}(r,r')-\pi_{[a]}(r)\pi_{[a]}(r')}{n_{[a]}^2  \pi_{[a]}(r) \pi_{[a]}(r')} + 
\frac{\pi_{[a]}(r)\pi_{[a]}(r')}{n_{[a]}^2  \pi_{[a]}(r) \pi_{[a]}(r')}
\right)
\sum_{i:A_i=a}\widetilde{Y}_i(r)\widetilde{Y}_i(r')\\
& =  0-\frac{(n_{[a]}-1) \pi_{[a][a]}(r,r')}{n_{[a]}^2  \pi_{[a]}(r) \pi_{[a]}(r')} (n_{[a]}-1)^{-1}\sum_{i:A_i=a}\widetilde{Y}_i(r)\widetilde{Y}_i(r')\\
& =  -(n_{[a]}-1)\frac{\pi_{[a][a]}(r,r')}{n_{[a]}^2  \pi_{[a]}(r) \pi_{[a]}(r')} S_{[a]}(r,r'),
\end{align*}
where the last equality follows from \eqref{eq:sh2r_sh2rr_prim}.

(3)
When $a=a'$ and $r=r'$, we similarly consider two cases with $i=j$ and $i\neq j$: 
\begin{align*}
\Var\{B_{[a]}(r)\} & =  \Cov\left[
\{n_{[a]} \pi_{[a]}(r)\}^{-1} \sum_{i:A_i=a}I(R_i=r)\widetilde{Y}_i(r), \{n_{[a]} \pi_{[a]}(r)\}^{-1} \sum_{j:A_j=a}I(R_j=r)\widetilde{Y}_j(r)
\right]\\
& =  \{n_{[a]}^2  \pi_{[a]}^2(r)\}^{-1} \sum_{i\neq j:A_i=a,A_j=a}\widetilde{Y}_i(r)\widetilde{Y}_j(r)\Cov\{
I(R_i=r), I(R_j=r)
\}\\
& \quad \   + \{n_{[a]}^2  \pi_{[a]}^2(r) \}^{-1} \sum_{i:A_i=a}\widetilde{Y}_i(r)\widetilde{Y}_i(r)\Cov\{
I(R_i=r), I(R_i=r)
\}\\
& =  \frac{\pi_{[a][a]}(r,r)-\pi_{[a]}^2(r)}{n_{[a]}^2  \pi_{[a]}^2(r) }\sum_{i\neq j:A_i=a,A_j=a}\widetilde{Y}_i(r)\widetilde{Y}_j(r)+ \frac{\pi_{[a]}(r)-\pi_{[a]}^2(r)}{n_{[a]}^2  \pi_{[a]}^2(r) }\sum_{i:A_i=a}\widetilde{Y}_i^2(r)
\end{align*}
where the last equality follows from \eqref{eq:cov_indicator}. We can further simplify $\Var\{B_{[a]}(r)\}$ as 
\begin{align*}
\Var\{B_{[a]}(r)\} & =  \frac{\pi_{[a][a]}(r,r)-\pi_{[a]}^2(r)}{n_{[a]}^2  \pi_{[a]}^2(r) }\sum_{i:A_i=a}\sum_{j:A_j=a}\widetilde{Y}_i(r)\widetilde{Y}_j(r) \\
& \quad \ + 
\left(
\frac{\pi_{[a]}(r)-\pi_{[a]}^2(r)}{n_{[a]}^2  \pi_{[a]}^2(r) } - 
\frac{\pi_{[a][a]}(r,r)-\pi_{[a]}^2(r)}{n_{[a]}^2  \pi_{[a]}^2(r) } 
\right)\sum_{i:A_i=a}\widetilde{Y}_i^2(r)
\\
& =  0 + \frac{\pi_{[a]}(r)- \pi_{[a][a]}(r,r)}{n_{[a]}^2  \pi_{[a]}^2(r) }\sum_{i:A_i=a}\widetilde{Y}_i^2(r)\\
& =  (n_{[a]}-1)\frac{\pi_{[a]}(r)- \pi_{[a][a]}(r,r)}{n_{[a]}^2  \pi_{[a]}^2(r) }S_{[a]}^2(r),
\end{align*}
where the last equality follows from \eqref{eq:sh2r_sh2rr_prim}.
\end{proof}

\begin{proof}[Proof of Lemma \ref{lemma:cov_between_center_and_mean}]
Similarly to the proof of Lemma \ref{lemma:var_between_centered_outcome}, we discuss three cases, and use \eqref{eq:cov_indicator} to calculate the covariance between two treatment indicators.
(1) When $a\neq a'$, 
\begin{align*}
\Cov\{B_{[a]}(r), C_{[a']}(r')\}  & =  \Cov\left[ \{n_{[a]} \pi_{[a]}(r)\}^{-1} \sum_{i:A_i=a}I(R_i=r)\widetilde{Y}_i(r), \{n_{[a']} \pi_{[a']}(r')\}^{-1} \sum_{j:A_j=a'}I(R_j=r')\bar{Y}_{[a']}(r') \right]\\
& =  \{n_{[a]} n_{[a']} \pi_{[a]}(r) \pi_{[a']}(r')\}^{-1} \cdot\bar{Y}_{[a']}(r') \sum_{i:A_i=a}
\sum_{j: A_j=a' }
\widetilde{Y}_i(r)\Cov\{
I(R_i=r), I(R_j=r')
\}\\
& =  \frac{\pi_{[a][a']}(r,r')-\pi_{[a]}(r) \pi_{[a']}(r')}{n_{[a]} n_{[a']} \pi_{[a]}(r) \pi_{[a']}(r')}\bar{Y}_{[a']}(r') \sum_{i:A_i=a}
\sum_{j: A_j=a' }\widetilde{Y}_i(r) 
\end{align*}
where the last equality follows from \eqref{eq:cov_indicator}. We can further simplify $\Cov\{B_{[a]}(r), C_{[a']}(r')\}$ as 
\begin{align*}
\Cov\{B_{[a]}(r), C_{[a']}(r')\} & =  \frac{\pi_{[a][a']}(r,r')-\pi_{[a]}(r) \pi_{[a']}(r')}{n_{[a]} n_{[a']} \pi_{[a]}(r) \pi_{[a']}(r')}\bar{Y}_{[a']}(r')\cdot n_{[a']}\sum_{i:A_i=a}\widetilde{Y}_i(r) =  0.
\end{align*}
(2) When $a=a'$ and $r\neq r'$,
\begin{align*}
\Cov\{B_{[a]}(r), C_{[a]}(r')\} & = \Cov\left[ \{n_{[a]} \pi_{[a]}(r)\}^{-1} \sum_{i:A_i=a} I(R_i=r)\widetilde{Y}_i(r), \{n_{[a]} \pi_{[a]}(r')\}^{-1} \sum_{j:A_j=a}I(R_j=r')\bar{Y}_{[a]}(r') \right]\\
& = \{n_{[a]}^2 \pi_{[a]}(r) \pi_{[a]}(r')\}^{-1} \cdot \bar{Y}_{[a]}(r') \sum_{i\neq j:A_i=a,A_j=a}\widetilde{Y}_i(r)\Cov\{
I(R_i=r), I(R_j=r')
\}\\
& \quad \  +\{n_{[a]}^2 \pi_{[a]}(r) \pi_{[a]}(r')\}^{-1}\cdot \bar{Y}_{[a]}(r')\sum_{i:A_i=a}\widetilde{Y}_i(r)\Cov\{
I(R_i=r), I(R_i=r')
\}\\
& = \frac{\pi_{[a][a]}(r,r')-\pi_{[a]}(r)\pi_{[a]}(r')}{n_{[a]}^2 \pi_{[a]}(r) \pi_{[a]}(r')}\bar{Y}_{[a]}(r')\sum_{i\neq j:A_i=a,A_j=a}\widetilde{Y}_i(r)\\
& \quad \ 
- \frac{\pi_{[a]}(r)\pi_{[a]}(r')}{n_{[a]}^2 \pi_{[a]}(r) \pi_{[a]}(r')}\bar{Y}_{[a]}(r')\sum_{i:A_i=a}\widetilde{Y}_i(r),
\end{align*}
where the last equality follows from from \eqref{eq:cov_indicator}. We can further simplify $\Cov\{B_{[a]}(r), C_{[a]}(r')\}$ as 
\begin{align*}
\Cov\{B_{[a]}(r), C_{[a]}(r')\} & =  \frac{\pi_{[a][a]}(r,r')-\pi_{[a]}(r)\pi_{[a]}(r')}{n_{[a]}^2 \pi_{[a]}(r) \pi_{[a]}(r')}\bar{Y}_{[a]}(r')\cdot (n_{[a]}-1) \sum_{i:A_i=a}\widetilde{Y}_i(r) - 0  = 0.
\end{align*}
(3) When $a=a'$ and $r=r'$,
\begin{align*}
\Cov\{B_{[a]}(r), C_{[a]}(r)\} & = \Cov\left[ \{n_{[a]} \pi_{[a]}(r)\}^{-1} \sum_{i:A_i=a}I(R_i=r)\widetilde{Y}_i(r), \{n_{[a]} \pi_{[a]}(r)\}^{-1}\sum_{j:A_j=a}I(R_j=r)\bar{Y}_{[a]}(r) \right]\\
& = \{n_{[a]}^2 \pi_{[a]}(r) \pi_{[a]}(r')\}^{-1}\sum_{i\neq j:A_i=a,A_j=a}\widetilde{Y}_i(r)\bar{Y}_{[a]}(r')\Cov\{
I(R_i=r), I(R_j=r)
\}\\
& \quad \  + \{n_{[a]}^2 \pi_{[a]}(r) \pi_{[a]}(r')\}^{-1} \sum_{i:A_i=a}\widetilde{Y}_i(r) \bar{Y}_{[a]}(r')\Cov\{
I(R_i=r), I(R_i=r)
\}\\
& = \frac{\pi_{[a][a]}(r,r)-\pi_{[a]}(r)\pi_{[a]}(r)}{n_{[a]}^2 \pi_{[a]}(r) \pi_{[a]}(r)}\sum_{i\neq j:A_i=a,A_j=a}\widetilde{Y}_i(r)\bar{Y}_{[a]}(r)\\
& \quad \ +\frac{\pi_{[a]}(r) - \pi_{[a]}(r)\pi_{[a]}(r) }{n_{[a]}^2 \pi_{[a]}(r) \pi_{[a]}(r)}\sum_{i:A_i=a}\widetilde{Y}_i(r)\bar{Y}_{[a]}(r),
\end{align*}
where the last equality follows from \eqref{eq:cov_indicator}. We can further simplify $\Cov\{B_{[a]}(r), C_{[a]}(r)\}$ as 
\begin{align*}
\Cov\{B_{[a]}(r), C_{[a]}(r)\} & =  \frac{\pi_{[a][a]}(r,r)-\pi_{[a]}(r)\pi_{[a]}(r)}{n_{[a]}^2 \pi_{[a]}(r) \pi_{[a]}(r)}\sum_{i:A_i=a}\widetilde{Y}_i(r)\times (n_{[a]}-1)\bar{Y}_{[a]}(r)- 0  = 0.
\end{align*}
\end{proof}

\begin{proof}[Proof of Lemma \ref{lemma:var_between_mean}]
Similary to the proofs of Lemmas \ref{lemma:var_between_centered_outcome} and \ref{lemma:cov_between_center_and_mean}, we 
discuss three cases separately, and use \eqref{eq:cov_indicator} to calculate the covariance between two treatment indicators.
(1) 
When $a\neq a'$, 
\begin{align*}
\Cov\{C_{[a]}(r), C_{[a']}(r')\}  & = \Cov\left[
\{n_{[a]} \pi_{[a]}(r)\}^{-1} \sum_{i:A_i=a}I(R_i=r)\bar{Y}_{[a]}(r),
\{n_{[a']} \pi_{[a']}(r')\}^{-1}\sum_{j:A_j=a'}I(R_j=r')\bar{Y}_{[a']}(r')
\right]\\
& = \{n_{[a]} n_{[a']} \pi_{[a]}(r) \pi_{[a']}(r')\}^{-1}\bar{Y}_{[a]}(r)\bar{Y}_{[a']}(r')\sum_{i:A_i=a}
\sum_{j: A_j=a'}
\Cov
\{
I(R_i=r), I(R_j=r')
\}\\
& = \frac{\pi_{[a][a']}(r,r')-\pi_{[a]}(r) \pi_{[a']}(r')}{n_{[a]} n_{[a']} \pi_{[a]}(r) \pi_{[a']}(r')}\bar{Y}_{[a]}(r)\bar{Y}_{[a']}(r')  n_{[a]} n_{[a']}\\
& 
=  \frac{\pi_{[a][a']}(r,r')-\pi_{[a]}(r) \pi_{[a']}(r')}{ \pi_{[a]}(r) \pi_{[a']}(r')}\bar{Y}_{[a]}(r)\bar{Y}_{[a']}(r'),
\end{align*}
where the last equality follows from \eqref{eq:cov_indicator}. 
Based on the definitions of $d_{[a][a']}(r,r')$ and $c_{[a][a']}(r,r')$ in \eqref{eq:d_def} and \eqref{eq:c_def}, 
we can further simplify $\Cov\{C_{[a]}(r), C_{[a']}(r')\}$ as 
\begin{align*}
\Cov\{C_{[a]}(r), C_{[a']}(r')\} & = 
\left(
\frac{\pi_{[a][a']}(r,r')}{ \pi_{[a]}(r)\pi_{[a']}(r')} - 1
\right)
\bar{Y}_{[a]}(r)\bar{Y}_{[a']}(r') = 
\frac{d_{[a][a']}(r,r')}{(n_{[a]} n_{[a']})^{1/2}}\bar{Y}_{[a]}(r)\bar{Y}_{[a']}(r')\\
&
= \frac{c_{[a][a']}(r,r')}{(n_{[a]} n_{[a']})^{1/2}}\bar{Y}_{[a]}(r)\bar{Y}_{[a']}(r').
\end{align*}
(2) When $a=a'$ and $r\neq r'$,
\begin{align*}
\Cov\{C_{[a]}(r), C_{[a]}(r')\} & = \Cov\left[
\{n_{[a]} \pi_{[a]}(r)\}^{-1}\sum_{i:A_i=a}I(R_i=r)\bar{Y}_{[a]}(r),
\{n_{[a]} \pi_{[a]}(r')\}^{-1}\sum_{j:A_j=a}I(R_j=r')\bar{Y}_{[a]}(r')
\right]\\
& = \{n_{[a]}^2 \pi_{[a]}(r) \pi_{[a]}(r')\}^{-1} \bar{Y}_{[a]}(r)\bar{Y}_{[a]}(r')\sum_{i\neq j:A_i=a,A_j=a}\Cov
\{
I(R_i=r), I(R_j=r')
\} \\
& \quad \  + \{n_{[a]}^2 \pi_{[a]}(r) \pi_{[a]}(r')\}^{-1}\bar{Y}_{[a]}(r)\bar{Y}_{[a]}(r')\sum_{i:A_i=a}\Cov
\{
I(R_i=r), I(R_i=r')
\}\\
& =  \frac{\pi_{[a][a]}(r,r')- \pi_{[a]}(r)\pi_{[a]}(r')}{n_{[a]}^2 \pi_{[a]}(r) \pi_{[a]}(r')}\bar{Y}_{[a]}(r)\bar{Y}_{[a]}(r') n_{[a]}(n_{[a]}-1)\\
& \quad \ 
- \frac{\pi_{[a]}(r)\pi_{[a]}(r')}{n_{[a]}^2 \pi_{[a]}(r) \pi_{[a]}(r')}\bar{Y}_{[a]}(r)\bar{Y}_{[a]}(r') n_{[a]},
\end{align*}
where the last equality follows from \eqref{eq:cov_indicator}. 
Based on the definitions of $d_{[a][a]}(r,r')$ and $c_{[a][a]}(r,r')$ in \eqref{eq:d_def} and \eqref{eq:c_def}, 
we can further simplify $\Cov\{C_{[a]}(r), C_{[a]}(r')\}$ as 
\begin{align*}
\Cov\{C_{[a]}(r), C_{[a]}(r')\}
& = 
\left\{
\left(1-n_{[a]}^{-1}\right)
\left(
\frac{\pi_{[a][a]}(r,r')}{ \pi_{[a]}(r)\pi_{[a]}(r')}-1
\right)
 - n_{[a]}^{-1}
\right\}\bar{Y}_{[a]}(r)\bar{Y}_{[a]}(r')
\\
& = 
\left\{
\left(1-n_{[a]}^{-1}\right)n_{[a]}^{-1}d_{[a][a]}(r,r') - n_{[a]}^{-1}
\right\}\bar{Y}_{[a]}(r)\bar{Y}_{[a]}(r')\\
& = 
\frac{\left(1-n_{[a]}^{-1}\right)d_{[a][a]}(r,r') - 1}{n_{[a]}}\bar{Y}_{[a]}(r)\bar{Y}_{[a]}(r')
\\
& = \frac{c_{[a][a]}(r,r')}{n_{[a]}}\bar{Y}_{[a]}(r)\bar{Y}_{[a]}(r').
\end{align*}
(3) When $a=a'$ and $r=r'$,
\begin{align*}
\Var\{C_{[a]}(r)\} & = \Cov\left[
\{n_{[a]} \pi_{[a]}(r)\}^{-1}\sum_{i:A_i=a}I(R_i=r)\bar{Y}_{[a]}(r),
\{n_{[a']} \pi_{[a']}(r')\}^{-1}\sum_{j:A_j=a'}I(R_j=r')\bar{Y}_{[a']}(r')
\right]\\
& = \{n_{[a]}^2 \pi_{[a]}^2(r) \}^{-1}\bar{Y}_{[a]}^2(r)\sum_{i\neq j:A_i=a,A_j=a}\Cov
\{
I(R_i=r), I(R_j=r)
\} \\
& \quad \  + \{n_{[a]}^2 \pi_{[a]}^2(r)\}^{-1}\bar{Y}_{[a]}^2(r)\sum_{i:A_i=a}\Cov
\{
I(R_i=r), I(R_i=r)
\}\\
& = \frac{\pi_{[a][a]}(r,r)- \pi_{[a]}^2(r)}{n_{[a]}^2 \pi_{[a]}^2(r)}\bar{Y}_{[a]}^2(r) n_{[a]}(n_{[a]}-1) +
\frac{\pi_{[a]}(r)-\pi_{[a]}^2(r)}{n_{[a]}^2 \pi_{[a]}^2(r)}\bar{Y}_{[a]}^2(r) n_{[a]},
\end{align*}
where the last equality follows from \eqref{eq:cov_indicator}. 
Based on the definitions of $d_{[a][a]}(r,r)$ and $c_{[a][a]}(r,r)$ in \eqref{eq:d_def} and \eqref{eq:c_def}, 
we can further simplify $\Var\{C_{[a]}(r)\}$ as 
\begin{align*}
\Var\{C_{[a]}(r)\} & = 
\left\{
\left(1-n_{[a]}^{-1}\right)
\left(
\frac{\pi_{[a][a]}(r,r)}{\pi_{[a]}^2(r)} - 1
\right) + n_{[a]}^{-1} \pi_{[a]}^{-1}(r) - n_{[a]}^{-1}
\right\}\bar{Y}_{[a]}^2(r)\\
& = 
\left\{
\left(1-n_{[a]}^{-1}\right)
n_{[a]}^{-1}d_{[a][a]}(r,r)
 + n_{[a]}^{-1} \pi_{[a]}^{-1}(r) - n_{[a]}^{-1}
\right\}\bar{Y}_{[a]}^2(r)
\\
& = 
\frac{\left(1-n_{[a]}^{-1}\right)
d_{[a][a]}(r,r) + \pi_{[a]}^{-1}(r) - 1
}{n_{[a]}}
\bar{Y}_{[a]}^2(r)
 = \frac{c_{[a][a]}(r,r)}{n_{[a]}} \bar{Y}_{[a]}^2(r).
\end{align*}
\end{proof}

\begin{proof}[Proof of Lemma \ref{lemma:represent_S_product}]
For $1\leq a\leq H$ and $r\neq r'\in \mathcal{R},$ 
by definition, we have 
\begin{eqnarray*}
 S_{[a]}^2(r)+S_{[a]}^2(r')-S_{[a]}^2(r\text{-}r') 
& = & (n_{[a]}-1)^{-1}\left[
\sum_{i:A_i=a}\widetilde{Y}_i^2(r)+\sum_{i:A_i=a}\widetilde{Y}_i^2(r') - \sum_{i:A_i=a}\left\{
\widetilde{Y}_i(r) - \widetilde{Y}_i(r')
\right\}^2
\right]\\
& = & (n_{[a]}-1)^{-1}\times 2\sum_{i:A_i=a}\widetilde{Y}_i(r)\widetilde{Y}_i(r')=2S_{[a]}(r,r'),
\end{eqnarray*}
and
\begin{eqnarray*}
& & \Yhhrrprime
 +(2n_{[a]})^{-1}\left\{
S_{[a]}^2(r)+S_{[a]}^2(r')-S_{[a]}^2(r\text{-}r')
\right\}\\
& = & \{n_{[a]}(n_{[a]}-1)\}^{-1}
\sum_{i\neq j:A_i=A_j=a}Y_i(r)Y_j(r')
 +n_{[a]}^{-1}S_{[a]}(r,r')\\
& = & \{n_{[a]}(n_{[a]}-1)\}^{-1} \sum_{i\neq j:A_i=A_j=a}Y_i(r)Y_j(r')+ \{n_{[a]}(n_{[a]}-1)\}^{-1}\left\{
\sum_{i: A_i=a}Y_i(r)Y_i(r') - n_{[a]} \bar{Y}_{[a]}(r)\bar{Y}_{[a]}(r')
\right\}\\
& = & \{n_{[a]}(n_{[a]}-1)\}^{-1}\left\{
\sum_{i:A_i=a}
\sum_{j:A_j=a}
Y_i(r)Y_j(r')  - n_{[a]} \bar{Y}_{[a]}(r)\bar{Y}_{[a]}(r')
\right\}\\
& = & \{n_{[a]}(n_{[a]}-1)\}^{-1} \left\{
n_{[a]}^2 \bar{Y}_{[a]}(r)\bar{Y}_{[a]}(r')  - n_{[a]} \bar{Y}_{[a]}(r)\bar{Y}_{[a]}(r')
\right\} 
 =  \bar{Y}_{[a]}(r)\bar{Y}_{[a]}(r').
\end{eqnarray*}
\end{proof}

\subsection{Proofs of the theorems for general assignment mechanism}

\begin{proof}[{\bf Proof of Theorem \ref{thm:var_gen}}]
First, we calculate the sampling variance of estimated subgroup average peer effect. 
From \eqref{eq:rewrite_tau_h} and Lemmas \ref{lemma:var_between_centered_outcome}--\ref{lemma:var_between_mean}, the covariances,
$\Cov\{B_{[a]}(r), C_{[a]}(r)\}$, $\Cov\{B_{[a]}(r), C_{[a]}(r')\}$, $\Cov\{B_{[a]}(r'), C_{[a]}(r)\}$ and $\Cov\{B_{[a]}(r'), C_{[a]}(r'),$ are all zero for $r\neq r'$.  Therefore, 
the sampling variance of subgroup average peer effect estimator is
\begin{align*}
& \quad \  \Var\left\{
\hat{\tau}_{[a]}(r,r') 
\right\}\\
 & = \Var\left\{
B_{[a]}(r)+C_{[a]}(r)-B_{[a]}(r')-C_{[a]}(r')
\right\}\\
& =
\Var\left\{ B_{[a]}(r) \right\} + 
 \Var\left\{
B_{[a]}(r')
\right\} - 2\Cov\left\{
B_{[a]}(r), B_{[a]}(r')
\right\}  + \Var\left\{
C_{[a]}(r)
\right\} + 
\Var\left\{
C_{[a]}(r')
\right\}\\
& \quad \  - 2\Cov\left\{
C_{[a]}(r),
C_{[a]}(r')
\right\}\\
& =  (n_{[a]}-1)\frac{\pi_{[a]}(r)-\pi_{[a][a]}(r,r)}{n_{[a]}^2  \pi_{[a]}^2(r) }S_{[a]}^2(r) + (n_{[a]}-1)\frac{\pi_{[a]}(r')-\pi_{[a][a]}(r',r')}{n_{[a]}^2 \pi_{[a]}^2(r') }S_{[a]}^2(r') \\
& \quad \  + 2\frac{(n_{[a]}-1) \pi_{[a][a]}(r,r')}{n_{[a]}^2  \pi_{[a]}(r) \pi_{[a]}(r')} S_{[a]}(r,r')
+ n_{[a]}^{-1}\left\{
c_{[a][a]}(r,r)\bar{Y}_{[a]}^2(r)+c_{[a][a]}(r',r')\bar{Y}_{[a]}^2(r') - 
2c_{[a][a]}(r,r')\bar{Y}_{[a]}(r)\bar{Y}_{[a]}(r')
\right\}.
\end{align*}
Replacing $2S_{[a]}(r,r')$ and $\bar{Y}_{[a]}(r)\bar{Y}_{[a]}(r')$ by their expressions in Lemma \ref{lemma:represent_S_product}, we can rewrite the sampling variance of $\hat{\tau}_{[a]}(r,r')$ as
\begin{eqnarray}\label{eq:proof_var_sub}
\Var\left\{
\hat{\tau}_{[a]}(r,r') 
\right\}
& = & (n_{[a]}-1)\frac{\pi_{[a]}(r)- \pi_{[a][a]}(r,r)}{n_{[a]}^2  \pi_{[a]}^2(r)}S_{[a]}^2(r) + (n_{[a]}-1)\frac{\pi_{[a]}(r')- \pi_{[a][a]}(r',r')}{n_{[a]}^2 \pi_{[a]}^2(r') }S_{[a]}^2(r') 
\nonumber
\\
&  & 
+ \frac{(n_{[a]}-1) \pi_{[a][a]}(r,r')}{n_{[a]}^2  \pi_{[a]}(r) \pi_{[a]}(r')} \left\{
S_{[a]}^2(r) + S_{[a]}^2(r') - S_{[a]}^2(r\text{-}r')
\right\}
\nonumber
\\
& & + n_{[a]}^{-1}\left\{
c_{[a][a]}(r,r)\bar{Y}_{[a]}^2(r)+c_{[a][a]}(r',r')\bar{Y}_{[a]}^2(r') - 
2c_{[a][a]}(r,r')
\Yhhrrprime \right\}
\nonumber\\
& & - \frac{2c_{[a][a]}(r,r')}{n_{[a]}}\frac{
S_{[a]}^2(r)+S_{[a]}^2(r')-S_{[a]}^2(r\text{-}r')
}{2n_{[a]}}.
\end{eqnarray}
We then combine the terms and calculate the cofficients of $S_{[a]}^2(r), S_{[a]}^2(r')$ and $S_{[a]}^2(r\text{-}r')$ in \eqref{eq:proof_var_sub}, separately. From the definitions of $c_{[a][a]}(r,r')$ and $b_{[a]}(r)$, we can simplify the coefficient of $S_{[a]}^2(r)$ as 
\begin{eqnarray*}
n_{[a]}^{-1}
\left\{
(1-n_{[a]}^{-1})
\left(\pi_{[a]}^{-1}(r) - \frac{\pi_{[a][a]}(r,r)}{\pi_{[a]}^2(r)} + \frac{\pi_{[a][a]}(r,r')}{\pi_{[a]}(r)\pi_{[a]}(r')} 
\right)- n_{[a]}^{-1}c_{[a][a]}(r,r')
\right\} = n_{[a]}^{-1}b_{[a]}(r), 
\end{eqnarray*}
and similarly simplify the coefficient of  $S_{[a]}^2(r')$ as $n_{[a]}^{-1}b_{[a]}(r').$
We can also simplify the coefficient of $S_{[a]}^2(r\text{-}r')$ as 
\begin{eqnarray*}
- \frac{(n_{[a]}-1) \pi_{[a][a]}(r,r')}{n_{[a]}^2  \pi_{[a]}(r) \pi_{[a]}(r')} + \frac{c_{[a][a]}(r,r')}{n_{[a]}^2} = - n_{[a]}^{-1}.
\end{eqnarray*}
Therefore, 
\eqref{eq:proof_var_sub} reduces to 
\begin{eqnarray*}
\Var\left\{
\hat{\tau}_{[a]}(r,r') 
\right\}
& = & n_{[a]}^{-1}\left\{
b_{[a]}(r)S_{[a]}^2(r) + b_{[a]}(r')S_{[a]}^2(r') - S_{[a]}^2(r\text{-}r')
\right\}\\
& & + n_{[a]}^{-1}\left\{
c_{[a][a]}(r,r)\bar{Y}_{[a]}^2(r)+c_{[a][a]}(r',r')
\bar{Y}_{[a]}^2(r') - 
2c_{[a][a]}(r,r')
\Yhhrrprime \right\}.
\end{eqnarray*}

Second, we calculate the covariance between two estimated subgroup average peer effects. For $a\neq a'$ and $r\neq r'\in \mathcal{R},$ 
according to Lemmas \ref{lemma:var_between_centered_outcome}--\ref{lemma:var_between_mean}, 
the covariances, $\Cov\{B_{[a]}(r),$ $B_{[a']}(r)\}$, $\Cov\{B_{[a]}(r),$ $B_{[a']}(r')\}$, $\Cov\{B_{[a]}(r'),$  $B_{[a']}(r)\}$, $\Cov\{B_{[a]}(r'),$ $B_{[a']}(r')\}$, $\Cov\{B_{[a]}(r),$ $C_{[a']}(r)\}$, $\Cov\{B_{[a]}(r),$  $C_{[a']}(r')\}$, $\Cov\{B_{[a]}(r'), C_{[a']}(r)\}$ and $\Cov\{B_{[a]}(r'), C_{[a']}(r')\}$, 
are all zero.  Therefore
the sampling covariance between $\hat{\tau}_{[a]}(r,r')$ and $\hat{\tau}_{[a']}(r,r')$ is
\begin{eqnarray*}
& & \Cov\left\{
\hat{\tau}_{[a]}(r,r'), \hat{\tau}_{[a']}(r,r') 
\right\}\\
& = & \Cov\left\{ B_{[a]}(r)+C_{[a]}(r)-B_{[a]}(r')-C_{[a]}(r'), B_{[a']}(r)+C_{[a']}(r)-B_{[a']}(r')-C_{[a']}(r')
\right\}
\\
& = & \Cov\left\{
C_{[a]}(r), C_{[a']}(r)
\right\} + \Cov\left\{
C_{[a]}(r'), C_{[a']}(r')
\right\} - \Cov\left\{
C_{[a]}(r),C_{[a']}(r')
\right\} -  \Cov\left\{
C_{[a]}(r'),
C_{[a']}(r)
\right\}\\
& = & 
(n_{[a]} n_{[a']})^{-1/2}
\left\{
c_{[a][a']}(r,r)\bar{Y}_{[a]}(r)\bar{Y}_{[a']}(r) + 
c_{[a][a']}(r',r')\bar{Y}_{[a]}(r')\bar{Y}_{[a']}(r')
\right. \\
& & \quad \quad \quad \quad \quad \quad \ 
\left.
- c_{[a][a']}(r,r')\bar{Y}_{[a]}(r)\bar{Y}_{[a']}(r') -
c_{[a][a']}(r',r)\bar{Y}_{[a]}(r')\bar{Y}_{[a']}(r)\right\},
\end{eqnarray*}
where the last equality follows from Lemma \ref{lemma:var_between_mean}. 

Third, we calculate the variance of the average peer effect estimator: 
\begin{align*}
\Var\left\{ \hat{\tau}(r,r') \right\}
 =  \Var\left\{ \sum_{a=1}^H w_{[a]} \hat{\tau}_{[a]}(r,r') \right\}
 = \sum_{a=1}^H w_{[a]}^2 \Var\left\{
\hat{\tau}_{[a]}(r,r')
\right\} + \sum_{a=1}^H\sum_{a'\neq a}w_{[a]} w_{[a']}\Cov\left\{
\hat{\tau}_{[a]}(r,r'), \hat{\tau}_{[a']}(r,r')
\right\}. 
\end{align*}
Using the variances and covariances between the subgroup average peer effect estimators, we can simplify the variance of $\hat{\tau}(r,r')$ as
\begin{eqnarray*}
\Var\left\{ \hat{\tau}(r,r') \right\}
& = & n^{-1}\sum_{a=1}^H w_{[a]} \left\{
b_{[a]}(r)S_{[a]}^2(r) + b_{[a]}(r')S_{[a]}^2(r') - S_{[a]}^2(r\text{-}r')
\right\}\\
& & + n^{-1}\sum_{a=1}^H w_{[a]} \left\{
c_{[a][a]}(r,r)\bar{Y}_{[a]}^2(r)+c_{[a][a]}(r',r')\bar{Y}_{[a]}^2(r') - 
2c_{[a][a]}(r,r')
\Yhhrrprime \right\}\\
& & + n^{-1} \sum_{a=1}^H\sum_{a'\neq a} (w_{[a]} w_{[a']})^{1/2} \left\{
c_{[a][a']}(r,r)\bar{Y}_{[a]}(r)\bar{Y}_{[a']}(r) + 
c_{[a][a']}(r',r')\bar{Y}_{[a]}(r')\bar{Y}_{[a']}(r')
\right.\\
& & \left. \quad \quad  \quad \quad  \quad \quad  \quad  \quad \quad  \quad \quad - c_{[a][a']}(r,r')\bar{Y}_{[a]}(r)\bar{Y}_{[a']}(r')  -
c_{[a][a']}(r',r)\bar{Y}_{[a]}(r')\bar{Y}_{[a']}(r)
\right\}.
\end{eqnarray*}
\end{proof}

\begin{proof}[{\bf Proof of Theorem \ref{thm:est_var_gen}}]
First, we prove that $s_{[a]}^2(r)$ defined in Section \ref{sec:est_var_gen} is unbiased for the finite population variance $S_{[a]}^2(r).$ 
Note that 
\begin{eqnarray*}
E\left\{\hat{Y}_{[a]}^2(r)\right\} & = &
E\left\{
\{n_{[a]} \pi_{[a]}(r)\}^{-1}\sum_{i:A_i=a}I(R_i=r)Y_i(r)
\times \{n_{[a]} \pi_{[a]}(r)\}^{-1}\sum_{j:A_j=a}I(R_j=r)Y_j(r)
\right\}\nonumber\\
& = & \{n_{[a]}^2 \pi_{[a]}^2(r)\}^{-1}\sum_{i\neq j: A_i=A_j=a}Y_i(r)Y_j(r)\pr(R_i=r, R_j=r) \nonumber \\
& & + \{n_{[a]}^2 \pi_{[a]}^2(r)\}^{-1}\sum_{i: A_i=a}Y_i^2(r)\pr(R_i=r, R_i=r)\nonumber\\
& = & \frac{\pi_{[a][a]}(r,r)}{n_{[a]}^2 \pi_{[a]}^2(r)}\sum_{i\neq j: A_i=A_j=a}Y_i(r)Y_j(r) + \frac{ \pi_{[a]}(r)}{n_{[a]}^2 \pi_{[a]}^2(r)}\sum_{i: A_i=a}Y_i^2(r),
\end{eqnarray*}
where the last equality follows from the definitions of $\pi_{[a]}(r)$ and $\pi_{[a][a]}(r,r)$. 
We can then simplify $E\{\hat{Y}_{[a]}^2(r)\}$ as 
\begin{eqnarray}\label{eq:Y_hat_square}
E\left\{\hat{Y}_{[a]}^2(r)\right\} & = & \frac{\pi_{[a][a]}(r,r)}{n_{[a]}^2 \pi_{[a]}^2(r)}\sum_{i: A_i=a}\sum_{j:A_j=a}Y_i(r)Y_j(r)+ \frac{\pi_{[a]}(r)- \pi_{[a][a]}(r)}{n_{[a]}^2 \pi_{[a]}^2(r)}\sum_{i: A_i=a}Y_i^2(r)\nonumber \\
& = & \frac{\pi_{[a][a]}(r,r)}{\pi_{[a]}^2(r)}\bar{Y}_{[a]}^2(r)+ \frac{\pi_{[a]}(r)- \pi_{[a][a]}(r)}{n_{[a]}^2 \pi_{[a]}^2(r)}\sum_{i: A_i=a}Y_i^2(r). 
\end{eqnarray}		
The mean of $s_{[a]}^2(r)$ is
\begin{eqnarray*}
E\left\{s_{[a]}^2(r)\right\} & = &
\frac{n_{[a]} \pi_{[a]}^2(r)}{(n_{[a]}-1) \pi_{[a][a]}(r,r)}
\left[
\frac{n_{[a]}+c_{[a][a]}(r,r)}{n_{[a]}^2 \pi_{[a]}(r)} \sum_{i:A_i=a}E\left\{I(R_i=r) \right\} Y_i^2(r)-
E\left\{\hat{Y}^2_{[a]}(r)\right\}
\right]\\
& = & 
\frac{n_{[a]} \pi_{[a]}^2(r)}{(n_{[a]}-1) \pi_{[a][a]}(r,r)}
\left[
\frac{n_{[a]}+c_{[a][a]}(r,r)}{n_{[a]}^2} \sum_{i:A_i=a}Y_i^2(r)-
E\left\{\hat{Y}^2_{[a]}(r)\right\}
\right]
\\
& = & 
\frac{n_{[a]} \pi_{[a]}^2(r)}{(n_{[a]}-1) \pi_{[a][a]}(r,r)}
\left[
\frac{(n_{[a]}-1) \pi_{[a][a]}(r,r) + \pi_{[a]}(r)}{n_{[a]}^2 \pi_{[a]}^2(r)} \sum_{i:A_i=a}Y_i^2(r)-
E\left\{\hat{Y}^2_{[a]}(r)\right\}
\right],
\end{eqnarray*}
where the last equality follows from the definition of $c_{[a][a]}(r,r)$ in \eqref{eq:c_def}. 
Using \eqref{eq:Y_hat_square}, we can further simplify $E\{s_{[a]}^2(r)\}$ as 
\begin{eqnarray*}
E\left\{s_{[a]}^2(r)\right\} & = & 
\frac{n_{[a]} \pi_{[a]}^2(r)}{(n_{[a]}-1) \pi_{[a][a]}(r,r)}
\left\{
\frac{(n_{[a]}-1) \pi_{[a][a]}(r,r) + \pi_{[a]}(r)}{n_{[a]}^2 \pi_{[a]}^2(r)} \sum_{i:A_i=a}  Y_i^2(r)
\right.\\
&& \quad \quad \quad \quad \quad \quad \quad  \quad \quad 
\left. -
\frac{\pi_{[a][a]}(r,r)}{\pi_{[a]}^2(r)}\bar{Y}_{[a]}^2(r)- \frac{\pi_{[a]}(r)- \pi_{[a][a]}(r)}{n_{[a]}^2 \pi_{[a]}^2(r)}\sum_{i: A_i=a}Y_i^2(r)
\right\}\\
& = & \frac{n_{[a]} \pi_{[a]}^2(r)}{(n_{[a]}-1) \pi_{[a][a]}(r,r)}
\left\{
\frac{\pi_{[a][a]}(r,r)}{n_{[a]} \pi_{[a]}^2(r)}\sum_{i: A_i=a}Y_i^2(r) - 
\frac{\pi_{[a][a]}(r,r)}{\pi_{[a]}^2(r)}\bar{Y}_{[a]}^2(r)
\right\} \\
& = & 
(n_{[a]}-1)^{-1}
\left\{
\sum_{i: A_i=a}Y_i^2(r) - n_{[a]}
\bar{Y}_{[a]}^2(r)
\right\}
=  S_{[a]}^2(r).
\end{eqnarray*}

Second, we prove the unbiasedness of the estimator for $\bar{Y}_{[a]}^2(r)$. For $1\leq a\leq H$ and $r\in \mathcal{R}$, accoring to \eqref{eq:Y_hat_square} and the unbiasedness of $s_{[a]}^2(r)$ for $S_{[a]}^2(r)$,
\begin{eqnarray*}
& & E\left[
\frac{n_{[a]}\hat{Y}_{[a]}(r)^2-\{b_{[a]}(r)-1\}s_{[a]}^2(r)
}{n_{[a]}+c_{[a][a]}(r,r)}
\right]\\
& = & 
\frac{n_{[a]} E\{\hat{Y}_{[a]}(r)^2\}-\{b_{[a]}(r)-1\}E\{s_{[a]}^2(r)\}
}{n_{[a]}+c_{[a][a]}(r,r)}\\
& = & \frac{n_{[a]}}{n_{[a]}+c_{[a][a]}(r,r)}
\left\{
\frac{\pi_{[a][a]}(r,r)}{\pi_{[a]}^2(r)}\bar{Y}_{[a]}^2(r)+ \frac{\pi_{[a]}(r)- \pi_{[a][a]}(r)}{n_{[a]}^2 \pi_{[a]}^2(r)}\sum_{i: A_i=a}Y_i^2(r)
\right\} - \frac{b_{[a]}(r)-1}{n_{[a]}+c_{[a][a]}(r,r)}S_{[a]}^2(r),
\end{eqnarray*}
which, based on the definitions of $c_{[a][a]}(r,r)$ and $b_{[a]}(r)$, further reduces to 
\begin{eqnarray*}
& & 
\frac{n_{[a]}}{(n_{[a]}-1) \pi_{[a][a]}(r,r)+ \pi_{[a]}(r)}
\left\{
\pi_{[a][a]}(r,r)\bar{Y}_{[a]}^2(r)+ \frac{\pi_{[a]}(r)- \pi_{[a][a]}(r)}{n_{[a]}^2}\sum_{i: A_i=a}Y_i^2(r)
\right\} \\
& & - 
\frac{\{\pi_{[a]}(r)- \pi_{[a][a]}(r,r)\}}{n_{[a]}\{(n_{[a]}-1) \pi_{[a][a]}(r,r)+ \pi_{[a]}(r)\}}
\left\{
\sum_{i:A_i=a}Y_i^2(r) - n_{[a]} \bar{Y}_{[a]}^2(r)
\right\}\\
& = & 
\frac{
n_{[a]} \pi_{[a][a]}(r,r)\bar{Y}_{[a]}^2(r)  +  \{\pi_{[a]}(r)- \pi_{[a][a]}(r,r)\}\bar{Y}_{[a]}^2(r)
}
{(n_{[a]}-1) \pi_{[a][a]}(r,r)+ \pi_{[a]}(r)}
=  \bar{Y}_{[a]}^2(r). 
\end{eqnarray*}

Third, we prove the unbiasedness of the 
estimator for $\Yhhrrprime$. For $1\leq a\leq H$ and $r\neq r'\in \mathcal{R}$,
\begin{align*}  
& \quad \  E\left\{  
\frac{n_{[a]}}{n_{[a]}-1}
\frac{\pi_{[a]}(r) \pi_{[a]}(r')}{\pi_{[a][a]}(r,r')}\hat{Y}_{[a]}(r)\hat{Y}_{[a]}(r') \right\}\nonumber\\
& = \frac{n_{[a]}}{n_{[a]}-1}
\frac{\pi_{[a]}(r)\pi_{[a]}(r')}{\pi_{[a][a]}(r,r')}E\left\{
\{n_{[a]} \pi_{[a]}(r)\}^{-1}\sum_{i:A_i=a}I(R_i=r)Y_i(r)
\times \{n_{[a]} \pi_{[a]}(r')\}^{-1}\sum_{j:A_j=a}I(R_j=r')Y_j(r')
\right\}\nonumber\\
& = \frac{n_{[a]}}{n_{[a]}-1}
\frac{1}{n_{[a]}^2 \pi_{[a][a]}(r,r')}
\sum_{i,j:A_i=A_j =a} Y_i(r)Y_j(r')\pr (R_i=r, R_j=r')\nonumber\\
& = \frac{1}{n_{[a]}(n_{[a]}-1) \pi_{[a][a]}(r,r')}
\left\{
\sum_{i\neq j:A_i=A_j=a}Y_i(r)Y_j(r') \pr (R_i=r, R_j=r') + \sum_{i:A_i=a}Y_i(r)Y_i(r') \pr (R_i=r, R_i=r')
\right\}\nonumber \\
& = \frac{1}{n_{[a]}(n_{[a]}-1) \pi_{[a][a]}(r,r')} \left\{
\pi_{[a][a]}(r,r')\sum_{i\neq j:A_i=A_j=a}Y_i(r)Y_j(r') + 0 
\right\} =  \Yhhrrprime, 
\end{align*}
where the second last equality holds because $\pr (R_i=r, R_i=r')=0$ for $r\neq r'\in \mathcal{R}$.

Fourth, we prove that, for $a\neq a'$ and $r,r'\in \mathcal{R}$, $\pi_{[a]}(r) \pi_{[a']}(r')/\pi_{[a][a']}(r,r') \cdot \hat{Y}_{[a]}(r)\hat{Y}_{[a']}(r')$ is unbiased for $\bar{Y}_{[a]}(r)\bar{Y}_{[a']}(r').$  For $a\neq a'$ and $r,r'\in \mathcal{R}$, any units ($i,j$) such that $A_i=a$ and $A_j=a'$ must satisfy $i\neq j$, and therefore
\begin{eqnarray*}
& & E\left\{
\frac{\pi_{[a]}(r) \pi_{[a']}(r')}{\pi_{[a][a']}(r,r')}
\hat{Y}_{[a]}(r)\hat{Y}_{[a']}(r')
\right\}\\
& = & \frac{\pi_{[a]}(r) \pi_{[a']}(r')}{\pi_{[a][a']}(r,r')} E\left\{
n_{[a]}^{-1}\sum_{i:A_i=a} \pi_{[a]}^{-1}(r)I(R_i=r)Y_i(r)
\times n_{[a']}^{-1}\sum_{j:A_j=a'}\pi_{a'}^{-1}(r')I(R_j=r')Y_j(r')
\right\}\\
& = & \{n_{[a]} n_{[a']}\pi_{[a][a']}(r,r')\}^{-1} \sum_{i:A_i=a}
\sum_{j:A_j=a'}
Y_i(r)Y_j(r')\pr (
R_i=r, R_j=r'
)\\
& = & (n_{[a]} n_{[a']})^{-1} \sum_{i:A_i=a}
\sum_{j:A_j=a'}Y_i(r)Y_j(r') =   \bar{Y}_{[a]}(r)\bar{Y}_{[a']}(r').
\end{eqnarray*}
\end{proof}

\section{More technical details about complete randomization}\label{sec::cr-appendix}

\begin{proof}[{\bf Proof of Proposition \ref{prop:CR_stratified}}]
We first show that the numerical implementation in Section \ref{sec:comp_rand} generates treatment assignments under complete randomization. 
For any treatment assignment $z$ with $L(z)=z$, by definition, there are $l_t$ groups with group attribute set $g_t$ for $1\leq t\leq T$. 
Thus, there are $\prod_{t=1}^{T}l_t!$ ways to arrange these $m$ groups such that the first $l_1$ groups have group attribute set $g_1$, the next $l_2$ groups have group attribute $g_2$, ..., and the last $l_T$ groups have group attribute $g_T$. 
Each of the $\prod_{t=1}^{T}l_t!$ arrangements is a group assignment from the numerical implementation, and all group assignments from the numerical implementation have the same probability.   
Therefore, 
under the numerical implementation, 
any assignment $z$ with $L(z)=z$ corresponds to $\prod_{t=1}^{T}l_t!$ realizations and
will have the same probability.

We then prove Proposition \ref{prop:CR_stratified}. From the above discussion, it is equivalent to consider the distribution of $(R_1, \ldots, R_n)$ under the group assignment generated from the numerical implementation. 
Under the numerical implementation, for each $1\leq a\leq H$ and the $n_{[a]}$ units with attribute $a$, 
the $l_1\times g_1(a)$ units in the first $l_1$ groups must receive treatment $g_1\setminus \{a\}$, 
the next $l_2\times g_2(a)$ units in the next $l_2$ groups must receive treatment $g_2\setminus \{a\}$, ..., 
and the last $l_T\times g_T(a)$ units in the last $l_T$ groups must receive treatment  $g_T\setminus \{a\}$, 
and the assignments of the $n_{[a]}$ units into these $m$ groups have the same probability. 
From the relationship between $n_{[a]r}$ and $L_t(z)g_t(a)$ in \eqref{eq:n_hr}, the first conclusion (1) in Proposition  \ref{prop:CR_stratified} holds. 
The second conclusion (2) in Proposition  \ref{prop:CR_stratified} follows directly from the independence among the group assignments for units with different attributes under the numerical implementation. 
\end{proof}

As a direct consequence of Proposition \ref{prop:CR_stratified}, we have the following results characterizing the probability law of complete randomization,  and we will use them in later proofs.

\begin{proposition}\label{prop:a_com_res}
Under Assumptions \ref{asp:sutva_pe} and \ref{asp:exres_pe}, and under
complete randomization defined in Section \ref{sec:comp_rand},  
for $1\leq a,a'\leq H$ and $r, r'\in \mathcal{R},$ we have $\pi_{[a]}(r)=n_{[a]r}/n_{[a]}$, $b_{[a]}({r})  = n_{[a]}/n_{[a]r}$, $c_{[a][a']}(r,r') = 0,$
and 
\begin{align*}
\pi_{[a][a']}(r,r')  =
\begin{cases}
\frac{n_{[a]r}n_{[a']r'}}{n_{[a]} n_{[a']}}; \\
\frac{n_{[a]r}n_{[a]r'}}{n_{[a]}(n_{[a]}-1)}; \\
\frac{n_{[a]r}(n_{[a]r}-1)}{n_{[a]}(n_{[a]}-1)}; 
\end{cases}
d_{[a][a']}(r,r')  = 
\begin{cases}
0, &\quad \quad \text{if } a\neq a';\\
\frac{n_{[a]}}{n_{[a]}-1}, &\quad \quad \text{if } a = a', r\neq r';\\
-\frac{n_{[a]}(n_{[a]}-n_{[a]r})}{(n_{[a]}-1)n_{[a]r}}, &\quad \quad \text{if } a = a', r= r'.
\end{cases}
\end{align*} 
\end{proposition}

The formulas of $\pi_{[a]}(r)$ and $\pi_{[a][a]}(r,r')$ are standard in completely randomized experiments with multiple treatments, the formula of $\pi_{[a][a']}(r,r')$ with $a\neq a'$ follows from the independence between treatments of units with different attributes, and the formulas of $d_{[a][a']}(r,r'), c_{[a][a']}(r,r')$ and $b_{[a]}(r)$ follow from their definitions in \eqref{eq:d_def}--\eqref{eq:b_def}.

\begin{proof}[{\bf Proof of Theorem \ref{thm:clt_comp_rand}}]
We first consider the point and interval estimator for the subgroup average peer effect. 
Under complete randomization, for the $n_{[a]}$ units with attribute $a$, we are essentially conducting a complete randomized experiments with $n_{[a]r}$ units receiving treatment $r$. 
Based on Lemma \ref{lemma:represent_S_product} and the regularity condition (ii) in Condition \ref{cond:fp}, 
the finite population covariance between potential outcomes $S_{[a]}(r,r')$ has a limit. 
From \citet[][Theorem 5]{fcltxlpd2016}, $\hat{\tau}_{[a]}(r,r')$ is asymptotically Normal:
\begin{eqnarray*}
\sqrt{n_{[a]}}\left\{
\hat{\tau}_{[a]}(r,r') - \tau_{[a]}(r,r')
\right\} 
& \converged & 
\mathcal{N}\left(
0, \lim_{n\rightarrow \infty} n_{[a]} \Var\{\hat{\tau}_{[a]}(r,r')\}
\right),
\end{eqnarray*}
where $\lim_{n\rightarrow \infty} n_{[a]} \Var\{\hat{\tau}_{[a]}(r,r')\}$ exists due to the convergence of proportions of units receiving different treatments $n_{[a]r}/n_{[a]}$ and the finite population variances of potential outcomes and individual peer effects $S^2_{[a]}(r)$ and $S^2_{[a]}(r\text{-}r').$ 

Moreover, according to \citet[][Proposition 3]{fcltxlpd2016}, the sample variance of observed outcomes in the subgroup consisting of units with attribute $a$ receiving treatment $r$, $s_{[a]}^2(r),$ is consistent for the population analogue $S_{[a]}^2(r)$, in the sense that $s_{[a]}^2(r)-S_{[a]}^2(r)\convergep 0.$ Thus, the variance estimator $\hat{V}_{[a]}(r,r')$ satisfies   
\begin{align*}
n_{[a]} \hat{V}_{[a]}(r,r') - n_{[a]} \Var\{\hat{\tau}_{[a]}(r,r')\} - S_{[a]}^2(r\text{-}r') = 
\frac{n_{[a]}}{n_{[a]r}}\left\{
s_{[a]}^2(r)-S_{[a]}^2(r)
\right\} +
\frac{n_{[a]}}{n_{[a]r'}}\left\{
s_{[a]}^2(r')-S_{[a]}^2(r')
\right\} 
\convergep  0 .
\end{align*}
Therefore, the Wald-type confidence interval for $\tau_{[a]}(r,r')$ is asymptotically conservative. 

Second, we consider the confidence interval for the average peer effect $\tau(r,r').$ Based on Slutsky's theorem, $\hat{\tau}(r,r')$ is asymptotically Normal: 
\begin{eqnarray*}
\sqrt{n}
\left\{\hat{\tau}(r,r') - \tau(r,r') 
\right\}
 =  
\sum_{a=1}^H \sqrt{w_{[a]}}\sqrt{n_{[a]}} \{\hat{\tau}_{[a]}(r,r') - \tau_{[a]}(r,r')\}
 \converged 
\mathcal{N}\left(
0, \ 
\lim_{n\rightarrow \infty} n \Var\{\hat{\tau}(r,r')\}
\right). 
\end{eqnarray*}
Moreover, the variance estimator $\hat{V}(r,r')$ satisfies that 
\begin{eqnarray*}
n\hat{V}({r}, {r}') - n \Var\{\hat{\tau}(r,r')\} - \sum_{a=1}^{H}w_{[a]} S_{[a]}^2(r\text{-}r')
 =  \sum_{a=1}^H  w_{[a]} 
\left\{
n_{[a]}\hat{V}_{[a]}({r}, {r}') - n_{[a]} \Var\{\hat{\tau}_{[a]}(r,r')\} - S_{[a]}^2(r\text{-}r')
\right\} \convergep 0. 
\end{eqnarray*}
Therefore, the Wald-type confidence interval for $\tau(r,r')$ is asymptotically conservative. 
\end{proof}

\begin{proof}[{\bf Proof of Theorem \ref{thm:asym_adjust_treat_effect}}]
We prove the three conclusions in Theorem \ref{thm:asym_adjust_treat_effect} as follows.

First, let $|\mathcal{R}|$ dimensional column vectors
$$
Y_i(\mathcal{R}) = (Y_i(r_1), \ldots, Y_i(r_{\mathcal{R}}) )^\top,\quad 
\bar{Y}_{[a]}(\mathcal{R}) = 
(\bar{Y}_{[a]}(r_1), \ldots, \bar{Y}_{[a]}(r_{|\mathcal{R}|}))^\top , \quad 
\hat{Y}_{[a]}(\mathcal{R})=(\hat{Y}_{[a]}(r_1), \ldots, \hat{Y}_{[a]}(r_{|\mathcal{R}|}))^\top 
$$
consist of unit $i$'s all potential outcomes, all subgroup average potential outcomes, and all subgroup average potential outcome estimators, respectively. Then 
$$
\theta_i(\mathcal{R}) = \Gamma Y_i(\mathcal{R}), \quad 
\theta_{[a]}(\mathcal{R}) = \Gamma \bar{Y}_{[a]}(\mathcal{R}), \quad 
\hat{\theta}_{[a]}(\mathcal{R}) =  \Gamma \hat{Y}_{[a]}(\mathcal{R}) .
$$
Based on the equivalence relationship in Proposition \ref{prop:CR_stratified} and the variance formula in \citet[][Theorem 3]{fcltxlpd2016}, the sampling covariance matrix of $\hat{Y}_{[a]}(\mathcal{R})$ under complete randomization is 
\begin{eqnarray*}
\Cov\{\hat{Y}_{[a]}(\mathcal{R})\} = 
\text{diag}\left\{
\frac{S^2_{[a]}(r_1)}{n_{[a]r_1}}, \ldots, 
\frac{S^2_{[a]}(r_{|\mathcal{R}|})}{n_{[a]r_{|\mathcal{R}|}}}
\right\}
- 
\frac{1}{n_{[a]}(n_{[a]}-1)}
\sum_{i:A_i=a}
\{Y_i(\mathcal{R})-\bar{Y}_{[a]}(\mathcal{R})\}
\{Y_i(\mathcal{R})-\bar{Y}_{[a]}(\mathcal{R})\}^\top,
\end{eqnarray*}
which implies the sampling covariance of 
$\hat{\theta}_{[a]}(\mathcal{R}) =  \Gamma \hat{Y}_{[a]}(\mathcal{R})$.

Second, the regularity conditions of Theorem \ref{thm:asym_adjust_treat_effect} and \citet[][Theorem 5]{fcltxlpd2016} immediately imply the asymptotic Normality of
$\sqrt{n_{[a]}}\{\hat{Y}_{[a]}(\mathcal{R})-\bar{Y}_{[a]}(\mathcal{R})\}$, which further implies the asymptotic Normality of 
$\sqrt{n_{[a]}}\{\hat{\theta}_{[a]}(\mathcal{R})-\theta_{[a]}(\mathcal{R})\}$.

Third, according to \citet[][Proposition 3]{fcltxlpd2016}, $s_{[a]}^2(r)$ is consistent for $S_{[a]}^2(r)$. Moreover, the second term in the covariance formula of $\hat{\theta}_{[a]}(\mathcal{R})$ is a positive semi-definite matrix. Therefore, the Wald-type confidence set using variance estimator \eqref{eq:est_var_adj_treat_effect} is asymptotically conservative.

Note that the second term in the covariance formula of $\hat{\theta}_{[a]}(\mathcal{R})$ is actually the finite population covariance matrix of 
$
(\theta_i(r_1), \theta_i(r_2), \ldots, \theta_i(r_{|\mathcal{R}|}))^\top
$
for units with attribute $a$ scaled by $n_{[a]}^{-1}$, and these centered individual potential outcomes $\theta_i(r)$'s can be represented as linear functions of the individual peer effects $\tau_i(r,r')$'s. Thus, when the individual peer effects for units with the same attribute are additive, 
the finite population covariance matrix of 
$
(\theta_i(r_1), \theta_i(r_2), \ldots, \theta_i(r_{|\mathcal{R}|}))^\top
$
for units with attribute $a$ are zero, 
and 
the Wald-type confidence sets for $\theta_{[a]}(\mathcal{R})$ become asymptotically exact. 
\end{proof}

\section{More on random partitioning}\label{app:more_on_rand_part}
In this section, we discuss details of random partitioning. In particular, we give the formulas for $\pi_{[a]}(r)$ and $\pi_{[a][a']}(r,r')$, based on which we can get the formulas for $d_{[a][a']}(r,r'), c_{[a][a']}(r,r')$ and  $b_{[a]}(r)$, 
the unbiased point estimators for peer effects, 
the sampling variances of peer effects estimators,  and the corresponding variance estimators.

Each $r\in\mathcal{R}$ is a set containing $K$ unordered but replicable elements from $\{1,2,\ldots,H\}$. 
Let $r(a)$ be the number of elements in set $r$ that are equal to $a$.
If $a$ itself belongs to $r$, let $r\setminus\{a\}$ be the set containing the remaining $K-1$ elements, by deleting an element $a$ from the set $r.$
\begin{theorem}\label{prop:a_rand_partition}
Under random partitioning, for $1\leq a,a'\leq H$ and $r,r'\in \mathcal{R}$, the probability that a unit $i$ with attribute $A_i=a$ receives treatment $r$ is
\begin{align}\label{eq:a_h_rand_partition}
\pi_{[a]}(r) = \pr(R_i=r)
= 
\frac{\binom{n_{[a]}-1}{r(a)} \prod_{1\leq q\leq H, q\neq a} \binom{n_{[q]}}{r(q)}}{
\binom{n-1}{K}
}, 
\end{align}
and the probability that two different units ($i\neq j$) with attributes $A_i=a$ and $A_j=a'$ receive treatments $r$ and $r'$ is
\begin{align}\label{eq:a_h_hprime_r_rprime_random_partition}
\pi_{[a][a']}(r,r') = \pr(R_i=r, R_j=r')
= \frac{K}{n-1} \cdot \psi_{[a][a']}(r,r') + \frac{n-K-1}{n-1} \cdot \phi_{[a][a']}(r,r'), 
\end{align}
where 
\begin{align}\label{eq:psi}
\psi_{[a][a']}(r,r') & = 
\begin{cases}
		\frac{
		\binom{n_{[a]}-1}{r(a)}\binom{n_{[a']}-1}{r(a')-1}
		\prod_{1\leq q\leq H, q\neq a,a'}\binom{n_{[q]}}{r(q)}}{\binom{n-2}{K-1}}, & \text{if } a'\in r, a\in r', r\setminus\{a'\}=r'\setminus\{a\},  a\neq a',\\
		\frac{
		\binom{n_{[a]}-2}{r(a)-1}
		\prod_{1\leq q\leq H, q\neq a}\binom{n_{[q]}}{r(q)}}{
		\binom{n-2}{K-1}}, & \text{if } a'\in r, a\in r', r\setminus\{a'\}=r'\setminus\{a\}, a= a', \\
		0, & \text{otherwise},
\end{cases}
\end{align}
and
\begin{align}\label{eq:phi}
\phi_{[a][a']}(r,r') & = 
\begin{cases}
		\frac{
			\binom{n_{[a]}-1}{r(a)}\binom{n_{[a]}-1-r(a)}{r'(a)}
			\binom{n_{[a']}-1}{r(a')}\binom{n_{[a']}-1-r(a')}{r'(a')}
			\prod_{1\leq q\leq H, q\neq a,a'}
			\left\{\binom{n_{[q]}}{r(q)}
			\binom{n_{[q]}-r(q)}{r'(q)}\right\}
		}{\binom{n-2}{K}\binom{n-2-K}{K}}, & \text{if } a\neq a',\\
		\frac{
			\binom{n_{[a]}-2}{r(a)}\binom{n_{[a]}-2-r(a)}{r'(a)}
			\prod_{1\leq q \leq H, q\neq a}\left\{\binom{n_{[q]}}{r(q)}
			\binom{n_{[q]}-r(q)}{r'(q)}\right\}}
		{\binom{n-2}{K}\binom{n-2-K}{K}}, & \text{if } a=a'.
\end{cases}
\end{align}
\end{theorem}

We give some intuition to explain the formulas of $\pi_{[a]}(r)$ and $\pi_{[a][a']}(r,r')$ under random partitioning. 
First, in \eqref{eq:a_h_rand_partition}, 
the denominator $\binom{n-1}{K}$ denotes the total number of possible peers for unit $i$, and
the 
numerator denotes the number of possible peers such that unit $i$ receives treatment $r$. 
Second, 
for any two different units $i$ and $j$ with attributes $a$ and $a'$, we consider two cases according to whether units $i$ and $j$ are in the same group or not. The coefficients $K/(n-1)$ and $(n-K-1)/(n-1)$ in \eqref{eq:a_h_hprime_r_rprime_random_partition} are the probabilities that units $i$ and $j$ are in the same group and not in the same group, respectively. 
Correspondingly,  
$\psi_{[a][a']}(r,r')$ and $\phi_{[a][a']}(r,r')$ represent the conditional probabilities that units $i$ and $j$ receive treatments $r$ and $r'$ given that $i$ and $j$ are and are not in the same group. 

When units $i$ and $j$ are in the same group, 
they have $K-1$ common peers, and therefore, 
the treatment $R_i$ of unit $i$ consists of unit $j$'s attribute and the $K-1$ common peers' attributes, and the treatment $R_j$ consists of unit $i$'s attribute and the $K-1$ common peers' attributes. Therefore, $\psi_{[a][a']}(r,r')$ is positive if and only if $a'\in r, a\in r'$ and  $r\setminus\{a'\}=r'\setminus\{a\}$. 
In \eqref{eq:psi}, when $\psi_{[a][a']}(r,r')\neq 0$, 
the denominator $\binom{n-2}{K-1}$  counts the number of possible $K-1$ units in the same group as units $i$ and $j$, and the numerator counts the number of possible $K-1$ units in the same group as units $i$ and $j$ such that units $i$ and $j$ receive treatments $r$ and $r'$. 
In \eqref{eq:phi}, the denominator 
$\binom{n-2}{K}\binom{n-2-K}{K}$ counts the number of possible peers for units $i$ and $j$, and the numerator counts the number of possible peers for units $i$ and $j$ such that units $i$ and $j$ receive treatments $r$ and $r'$.

\begin{proof}[{\bf Proof of Theorem  \ref{prop:a_rand_partition}}]
First, we calculate $\pi_{[a]}(r)$. 
Assume that unit $i$ has attribute $A_i=a.$ 
The total number of possible peers of unit $i$ is $\binom{n-1}{K}$, and the total number of possible peers of unit $i$ such that unit $i$ receives treatment $r$ is 
\begin{align*}
\binom{n_{[a]}-1}{r(a)} \prod_{1\leq q\leq H, q\neq a} \binom{n_{[q]}}{r(q)},
\end{align*}
where $\binom{n_{[a]}-1}{r(a)}$ counts the number of possible choice of the peers of unit $i$ with attribute $a$, and $\binom{n_{[q]}}{r(q)}$ counts the number of possible choice of the peers of unit $i$ with attribute $q\neq a$. 
Under random partitioning, any other $K$ units have the same probability to be in the same group as unit $i.$ 
Therefore, \eqref{eq:a_h_rand_partition} holds.

Second, we calculate $\pi_{[a][a']}(r,r').$ Assume that $i\neq j$ are two units with attributes $A_i=a$ and $A_j=a'$. Under random partitioning, we can decompose the probability $\pr (R_i=r, R_j = r')$ into two parts according to whether units $i$ and $j$ are in the same group:
\begin{align}\label{eq:cal_a_random_partition}
 \pi_{[a][a']}(r,r')  & = \pr (R_i=r, R_j = r')  = \pr (j \in Z_i, R_i=r, R_j = r') + \pr (j \notin Z_i, R_i=r, R_j = r') \nonumber\\
& = 
\pr (j \in Z_i) \pr ( R_i=r, R_j = r' \mid j \in Z_i) + \pr (j \notin Z_i) \pr ( R_i=r, R_j = r' \mid j \notin Z_i).
\end{align}
Because any other $K$ units have the same probability to be the peers of unit $i$,  by symmetry, $\pr (j \in Z_i)=K/(n-1)$ and $\pr (j \notin Z_i)=1-K/(n-1).$
We then consider the two conditional probabilities
$
\psi_{[a][a']}(r,r') \equiv \pr ( R_i=r, R_j = r' \mid j \in Z_i), 
$
and
$
\phi_{[a][a']}(r,r') \equiv 
\pr ( R_i=r, R_j = r' \mid j \notin Z_i).
$

When units $i$ and $j$ are in the same group, units $i$ and $j$ are peers of each other and they have $K-1$ common peers. 
Therefore, 
they have positive probability to receive treatments $r$ and $r'$ if and only if $r$ and $r'$ satisfy $a'\in r$, $a\in r'$, and $r\setminus\{a'\}=r'\setminus\{a\}$. 
When $a'\in r$, $a\in r'$, and $r\setminus\{a'\}=r'\setminus\{a\}$, 
the total number of possible peers of units $i$ and $j$ such that units $i$ and $j$ receive treatments $r$ and $r'$ is 
\begin{align*}
\begin{cases}
\binom{n_{[a]}-1}{{r}(a)}\binom{n_{[a']}-1}{{r}(a')-1}
\prod_{1\leq q\leq H, q\neq a,a'}\binom{n_{[q]}}{{r}(q)}, & \text{if } a\neq a',\\
\binom{n_{[a]}-2}{{r}(a)-1}
\prod_{1\leq q\leq H, q\neq a}\binom{n_{[q]}}{{r}(q)}, & \text{if } a= a'.
\end{cases}
\end{align*}
Note that the total number of possible peers of units $i$ and $j$ is $\binom{n-2}{K-1}$. 
Because any possible peers of units $i$ and $j$ have the same probability, by symmetry, 
\begin{align*}
\psi_{[a][a']}(r,r') & =  \pr ( R_i=r, R_j = r' \mid j \in Z_i)\\
& = 
\begin{cases}
\binom{n-2}{K-1}^{-1}
\binom{n_{[a]}-1}{r(a)}\binom{n_{[a']}-1}{r(a')-1}
\prod_{1\leq q\leq H, q\neq a,a'}\binom{n_{[q]}}{{r}(q)}, & \text{if } a'\in r, a\in r', r\setminus\{a'\}=r'\setminus\{a\}, a\neq a',\\
\binom{n-2}{K-1}^{-1}
\binom{n_{[a]}-2}{r(a)-1}
\prod_{1\leq q\leq H, q\neq a}\binom{n_{[q]}}{{r}(q)}, & \text{if } a'\in r, a\in r', r\setminus\{a'\}=r'\setminus\{a\},  a= a', \\
0, & \text{otherwise.}
\end{cases}
\end{align*}

When units $i$ and $j$ are not in the same group, the total number of their possible peers is 
$
\binom{n-2}{K}\binom{n-2-K}{K},
$
and the total number of their possible peers such that units $i$ and $j$ receive treatments $r$ and $r'$ is
\begin{align*}
\begin{cases}
\binom{n_{[a]}-1}{r(a)}\binom{n_{[a]}-1-r(a)}{r'(a)}
\binom{n_{[a']}-1}{r(a')}\binom{n_{[a']}-1-r(a')}{r'(a')}
\prod_{1\leq q\leq H, q\neq a,a'}
\left\{\binom{n_{[q]}}{r(q)}
\binom{n_{[q]}-r(q)}{r'(q)}\right\}, & \text{if } a\neq a',\\
\binom{n_{[a]}-2}{r(a)}\binom{n_{[a]}-2-r(a)}{r'(a)}
\prod_{1\leq q \leq H, q\neq a}\left\{\binom{n_{[q]}}{r(q)}
\binom{n_{[q]}-r(q)}{r'(q)}\right\}, & \text{if } a=a'.
\end{cases}
\end{align*}
Because any possible peers of units $i$ and $j$ have the same probability, 
by symmetry, 
\begin{align*}
\phi_{[a][a']}(r,r') & =  \pr ( R_i=r, R_j = r' \mid j \notin Z_i)\\
& = 
\begin{cases}
\frac{
\binom{n_{[a]}-1}{r(a)}\binom{n_{[a]}-1-r(a)}{r'(a)}
\binom{n_{[a']}-1}{r(a')}\binom{n_{[a']}-1-r(a')}{r'(a')}
\prod_{1\leq q\leq H, q\neq a,a'}
\left\{\binom{n_{[q]}}{r(q)}
\binom{n_{[q]}-r(q)}{r'(q)}\right\}
}{\binom{n-2}{K}\binom{n-2-K}{3}}, & \text{if } a\neq a',\\
\frac{
\binom{n_{[a]}-2}{r(a)}\binom{n_{[a]}-2-r(a)}{r'(a)}
\prod_{1\leq q \leq H, q\neq a}\left\{\binom{n_{[q]}}{r(q)}
\binom{n_{[q]}-r(q)}{r'(q)}\right\}}
{\binom{n-2}{K}\binom{n-2-K}{K}}, & \text{if } a=a'.
\end{cases}
\end{align*}
We have computed the four terms in \eqref{eq:cal_a_random_partition}, and Theorem  \ref{prop:a_rand_partition} follows directly. 
\end{proof}

\end{document}